\documentclass[11pt]{article} 

\usepackage[utf8]{inputenc}
\usepackage[T1]{fontenc}
\usepackage[pdftex,left=1in,top=1in,bottom=1in,right=1in]{geometry}

\usepackage{amsmath, amssymb, dsfont, mathrsfs,amsthm}

\usepackage{graphicx}
\usepackage{caption}
\usepackage{bm}
\usepackage{xcolor}
\usepackage[breaklinks,colorlinks=true,citecolor=blue,linkcolor=blue]{hyperref}
\hypersetup{
    colorlinks,
    linkcolor={red!50!black},
    citecolor={blue!50!black},
    urlcolor={blue!80!black}
}
\usepackage[nameinlink, noabbrev, capitalize]{cleveref}

\crefname{prop}{Property}{Properties}

\usepackage{thm-restate}

\usepackage{tikz}
\usetikzlibrary{calc,positioning,shapes,backgrounds, decorations.pathreplacing}
\usepackage{pgfplots}
\pgfplotsset{compat=1.18}
\usepackage{subcaption}
\allowdisplaybreaks

\newcommand{\N}{\mathds{N}}

\newcommand{\R}{\mathbb{R}}
\newcommand{\E}{\mathds{E}}
\newcommand{\cA}{\mathcal{A}}

\newcommand{\cB}{\mathcal{B}}
\newcommand{\cC}{\mathcal{C}}

\newcommand{\cX}{\mathcal{X}}
\newcommand{\cY}{\mathcal{Y}}

\newcommand{\bbE}{\mathbb{E}}

\newcommand{\eps}{\varepsilon}

\newtheorem{theorem}{Theorem}
\newtheorem{proposition}{Proposition}

\newtheorem{claim}{Claim}
\newtheorem{lemma}{Lemma}

\newtheorem{corollary}{Corollary}

\newtheorem{definition}{Definition}
\newtheorem{remark}{Remark}

\newtheorem{construction}{Construction}

\newcommand{\bits}{\{0,1\}}

\newcommand{\Enc}{\mathsf{Enc}}
\newcommand{\Dec}{\mathsf{Dec}}

\newcommand{\ICap}{\mathsf{ICap}}
\newcommand{\CCap}{\mathsf{CCap}}
\newcommand{\SCap}{\mathsf{SCap}}

\newcommand{\ceil}[1]{ \left \lceil {#1} \right \rceil}
\newcommand{\floor}[1]{ \left \lfloor {#1} \right \rfloor}
\newcommand{\Rin}{\mathcal{R_{\textup{in}}}}
\newcommand{\Rout}{\mathcal{R_{\textup {out}}}}

\newcommand{\delout}{\delta_{\textup{out}}}
\newcommand{\epsout}{\epsilon_{\textup{out}}}

\newcommand{\Cout}{C_{\textup{out}}}
\newcommand{\Cin}{C_{\textup{in}}}

\newcommand{\ed}{\textup{ED}}
\newcommand{\zo}{\{0,1\}}

\newcommand{\Bin}{\textup{Bin}}
\newcommand{\sigout}{\sigma^{(\sf out)}}
\newcommand{\sigouttil}{\tilde{\sigma}^{(\sf out)}}

\newcommand{\bdcRLB}{\textup{BDC-RL-Bounded}(d, \mu, M)}
\newcommand{\bdcMTRLB}{\textup{BDC-MT-RL-Bounded}(d, T, \mu, M)}

\newcommand{\bdctzerod}{\textup{BDC-Thr}(\tau,d)}

\newcommand{\Sigmain}{\Sigma_{\mathsf{in}}}

\newcommand{\Wpre}{W_{\mathsf{pre}}}

\let\originalleft\left
\let\originalright\right
\renewcommand{\left}{\mathopen{}\mathclose\bgroup\originalleft}
\renewcommand{\right}{\aftergroup\egroup\originalright}

\makeatletter
\def\blfootnote{\xdef\@thefnmark{}\@footnotetext}
\makeatother

\title{Channels with Input-Correlated Synchronization Errors
}
\author{Roni Con\thanks{Department of Computer Science, Technion - Israel  Institute  of  Technology, Haifa. \texttt{roni.con93@gmail.com}} \and  João Ribeiro\thanks{Instituto de Telecomunicações and Departamento de Matemática, Instituto Superior Técnico, Universidade de Lisboa. \texttt{jribeiro@tecnico.ulisboa.pt}}
}

\date{}

\begin{document}
	
\maketitle
	
\begin{abstract}
    ``Independent and identically distributed'' errors do not accurately capture the noisy behavior of real-world data storage and information transmission technologies.
    Motivated by this, we study channels with \emph{input-correlated} synchronization errors, meaning that the distribution of synchronization errors (such as deletions and insertions) applied to the $i$-th input $x_i$ may depend on the whole input string $x$.
    
    We begin by identifying conditions on the input-correlated synchronization channel under which the channel's information capacity is achieved by a stationary ergodic input source and is equal to its coding capacity.
    These conditions capture a wide class of channels, including channels with correlated errors observed in DNA-based data storage systems and their multi-trace versions, and generalize prior work.
    To showcase the usefulness of the general capacity theorem above, we combine it with techniques of Pernice-Li-Wootters (ISIT 2022)
    and Brakensiek-Li-Spang (FOCS 2020)
    to obtain explicit capacity-achieving codes for 
    multi-trace
    channels with \emph{runlength-dependent deletions}, motivated by error patterns observed in DNA-based data storage systems.\blfootnote{The work of R.\ Con was funded by the European Union (DiDAX, 101115134). The work of J.\ Ribeiro was funded by the European Union (LESYNCH, 101218842) and by national funds through FCT – Fundação para a Ciência e a Tecnologia, I.P., and, when eligible, co-funded by EU funds under project/support UID/50008/2025 – Instituto de Telecomunicações, with DOI \href{https://doi.org/10.54499/UID/50008/2025}{10.54499/UID/50008/2025}.
    Views and opinions expressed are however those of the authors only and do not necessarily reflect those of the European Union or the European Research Council Executive Agency. Neither the European Union nor the granting authority can be held responsible for them.}
    
\end{abstract}

\newpage

\tableofcontents
\newpage

\section{Introduction}\label{sec:intro}

Errors which cause loss of synchronization between sender and receiver, such as deletions, insertions, and replications, occur in various communications and data storage technologies, with DNA-based data storage being a notable recent example.
Despite considerable effort, pinning down the capacity and designing efficient nearly-optimal codes for channels with synchronization channels remain major problems in information and coding theory.
The surveys by Mitzenmacher~\cite{Mit09}, Mercier, Bhargava, and Tarokh~\cite{MBT10}, and Cheraghchi and Ribeiro~\cite{CR20} provide in-depth discussions of the many challenges encountered when dealing with synchronization errors.

Most prior work on channels with synchronization errors has focused on i.i.d.\ errors.
However, synchronization errors in real-world systems do not satisfy this assumption. 
For example, it has been observed in empirical analyses of widely used DNA sequencing technologies~\cite{RRC13,HMG19} that short substrings of DNA strands (which are written using a $4$-symbol alphabet $A,C,G,T$) with either a very high (above $75\%$) or very low (below $25\%$) concentration of $G$'s and $C$'s experience higher deletion rates~\cite[Figure 4]{RRC13}, and that longer runs of the same symbol in DNA strands experience higher deletion rates than shorter runs during sequencing~\cite[Figure 5]{RRC13}.
More broadly, a long series of works has explored, among many other things, the types of errors and error rates that occur in such systems, using different synthesis and sequencing technologies~\cite{Gol13,RRC13,Gra15,Yaz15,Bor16,EZ17,YGM17,Org18,HMG19,Pre20,SGPY21,weindel2023embracing}.

Motivated by this, we study a general class of channels with synchronization errors where the error distribution of the $i$-th input symbol $x_i$ may depend on the whole input $x$.
We also consider ``multi-trace'' versions of these channels, where the input $x$ is sent through multiple independent channels, generating multiple channel outputs (\emph{traces}) at the receiver end, which is especially relevant in DNA-based data storage.

\subsection{Our contributions}

\paragraph{Capacity theorems for channels with input-correlated synchronization errors}
Our first contribution, in \cref{sec:shannon-gen}, is a capacity theorem for a general class of channels with input-correlated synchronization errors, which we call \emph{admissible channels}.
The precise definition of an admissible channel is given in \cref{sec:admissible}.
As we show afterwards in \cref{sec:special-cases}, this class includes as special cases the channel model of Mao, Diggavi, and Kannan~\cite{MDK18}, multi-trace channels with input-correlated synchronization errors, and, more concretely, channels with runlength-dependent deletions where bits in runs of length $\ell$ are deleted independently with probability $d(\ell)$, for an arbitrary function $d:\N\to[0,1]$.
Channels with runlength-dependent deletions will be our concrete running example throughout this paper to showcase the usefulness of this result.

\begin{theorem}[Informal, see \cref{thm:cap-gen} for a formal statement]\label{thm:informal-cap-gen}
    Let $Z$ be an admissible channel (see \cref{sec:admissible}).
    Then, its information capacity equals its coding capacity, and the information capacity is achieved by stationary ergodic sources.
\end{theorem}

\paragraph{Efficient capacity-achieving codes for single- and multi-trace runlength-dependent deletion channels.}
As observed by Pernice, Li, and Wootters~\cite{PLW22}, a standard but quite useful consequence of \cref{thm:informal-cap-gen} is that admissible channels have capacity-achieving codes with additional structure.
In particular, admissible channels $Z$ with binary input alphabet have capacity-achieving codes $\cC$ such that every short substring of every codeword $c\in\cC$ has not-too-small Hamming weight (see \cref{thm:dense-code-gen} for a formal statement).

To showcase the usefulness of this consequence of \cref{thm:informal-cap-gen}, we use it to obtain efficient capacity-achieving codes for channels with ``bounded'' runlength-dependent deletions, \emph{in both the single-trace and multi-trace settings}.
More precisely, a bounded runlength-dependent deletion channel is a runlength-dependent deletion channel with a non-decreasing deletion probability function $d:\N\to[0,1]$, and such that there exists an integer $M$ such that $d(\ell)=d(M)<1$ for all $\ell\geq M$ (see \Cref{def:run-length-bounded-channel} for a formal definition).
Its \emph{$T$-trace version} is the channel that on input $x$ outputs $Y_1,\dots,Y_T$ (called \emph{traces}), where the $Y_i$'s are independent and identically distributed like outputs of the bounded runlength-dependent deletion channel on input $x$.
We remark that channels with runlength-dependent \emph{substitutions} have been studied before, also motivated by correlated errors in DNA-based data storage~\cite{weindel2023embracing}.

In the single-trace setting (\cref{sec:efficient-single-trace}) we combine the structured codes resulting from \cref{thm:informal-cap-gen} with an approach originally applied to channels with i.i.d.\ synchronization errors in~\cite{PLW22} to obtain efficient capacity-achieving codes for any bounded runlength-dependent deletion channel.

\begin{theorem}[Efficient capacity-achieving single-trace codes, informal. See \cref{thm:efficient-bounded-rl} for a formal statement]\label{thm:informal-single-trace}
    Let $\mathsf{Ch}$ be a bounded runlength-dependent deletion channel. 
    Denote its capacity by $\mathsf{Cap}$ and let $\eps>0$ be an arbitrary constant.
	Then, there exists a family of binary codes $\{C_n\}_{n=1}^{\infty}$, where $C_n$ has blocklength $n$ and rate $R_n > \mathsf{Cap} - \eps$ for all large enough $n$, that is robust for this channel. 
	Moreover, the $C_n$'s are encodable in linear time and decodable in quasi linear time in the message length.
\end{theorem}

As we discuss in more detail in \cref{sec:efficient-single-trace}, this result can be generalized well beyond channels with bounded runlength-dependent deletions. 
Broadly speaking, the theorem can be extended to any runlength-dependent deletion channel that satisfies two key conditions.
First, with high probability, a \emph{very long} run of zeros will remain long after passing through the channel. These long zero runs help maintain some degree of synchronization between the sender and the receiver.
Second, in any interval that is not too short and contains a high density of 1s, the probability that \emph{all} of the 1s are deleted is very small.
In this paper, we focus on the bounded runlength-dependent deletions (where the deletion function is non-decreasing) to keep the analysis simpler and because it is more realistic in practice for longer runs to experience higher deletion rates.

In the multi-trace setting (\cref{sec:multi-trace}) we again rely on \cref{thm:informal-cap-gen}, which in particular implies the existence of structured codes attaining the capacity of the \emph{multi-trace} bounded runlength-dependent channel. 
Then, to achieve efficient codes with matching rate, we combine it with the main idea
of Brakensiek, Li, and Spang~\cite{BLS20}, with some minor modifications, who showed how to compile average-case trace reconstruction algorithms into efficient rate $1 - o(1)$ codes for the coded trace reconstruction problem.
This yields efficient capacity-achieving codes for any $t$-trace bounded runlength-dependent deletion channel. 

\begin{theorem}[Efficient capacity-achieving multi-trace codes, informal. See \cref{thm:efficient-bounded-multi-trace-rl} for a formal statement]\label{thm:informal-multi-trace}
    Let $\mathsf{Ch}$ be a $T$-trace bounded runlength-dependent deletion channel, for an arbitrary constant integer $T\geq 1$. Denote its capacity by $\mathsf{Cap}$ and let $\eps>0$ be an arbitrary constant.
	Then, there exists a family of binary codes $\{C_n\}_{n=1}^{\infty}$, where $C_n$ has blocklength $n$ and rate $R_n > \mathsf{Cap} - \eps$ for all large enough $n$, that is robust for this channel. 
	Moreover, the $C_n$'s are encodable in linear time and decodable in quadratic time in the message length.
\end{theorem}

The difference between \cref{thm:informal-single-trace} and \cref{thm:informal-multi-trace} is that the former guarantees better decoding complexity in the single-trace ($T=1$) setting.
We see \cref{thm:informal-single-trace} as a natural warmup towards the multi-trace result in \cref{thm:informal-multi-trace}, as its proof is significantly simpler.

\paragraph{Capacity lower bounds for runlength-dependent deletion channels.}

Our results presented above effectively allow us to turn  any capacity lower bound for bounded runlength-dependent deletion channels into efficiently encodable and decodable codes with that rate.
To complement this, in \cref{sec:lowerbounds} we study concrete (single-trace) capacity lower bounds on arguably the simplest class of bounded runlength-dependent deletion channels: For a threshold $\tau\geq 1$ and $d\in[0,1]$, consider the runlength-dependent deletion channel that independently deletes each bit in a run of length at least $\tau$ with probability $d$, and does not apply any deletions to bits in runs of length less than $\tau$.
The case $\tau=1$ recovers the standard i.i.d.\ deletion channel.

A naive approach towards lower bounding the capacity of these channels is to take the largest \emph{runlength-limited} code, whose codewords only have runs of length less than $\tau$.
The maximal rate of such codes as a function of $\tau$ is well known.
We obtain capacity lower bounds that improve on this baseline approach for a large range of $d$, as illustrated in \cref{fig:tau-2,fig:tau-3} for thresholds $\tau=2$ and $\tau=3$, respectively.
By \cref{thm:informal-single-trace} we automatically get nearly-linear time encodable and decodable codes with that rate.

\subsection{Related work}

\paragraph{Capacity theorems for channels with synchronization errors.}
The first work to study this topic was by Dobrushin~\cite{Dob67}, who obtained capacity theorems for channels that apply i.i.d.\ synchronization errors.
Recently, there has been interest in extending such capacity theorems beyond i.i.d.\ errors.
Mao, Diggavi, and Kannan~\cite{MDK18} consider channels combining synchronization errors and (bounded) intersymbol interference as a model of nanopore-based sequencing.
More precisely, the behavior of the channel on input $x_i$ is some function of $x_i,x_{i-1},\dots,x_{i-\ell}$, for some memory threshold $\ell$.
They show that the information and coding capacity of these channels coincide, generalizing Dobrushin's result~\cite{Dob67}, but do not show that capacity is achieved by a stationary ergodic (or Markov) source.
Capacity theorems for a related (more concrete) model of nanopore-based sequencing with noisy duplications have also been studied by McBain, Saunderson, and Viterbo~\cite{MVS24,MSV24}.
Li and Tan~\cite{LT21} consider a channel obtained from the concatenation of a standard deletion channel with i.i.d.\ deletions and a finite-state discrete memoryless channel.
They show that the capacity of this channel is achieved by Markov processes, which implies that the polar codes developed by Tal, Pfister, Fazeli, and Vardy~\cite{TPFV22} achieve capacity on the i.i.d.\ deletion channel (a special case of this result was proved earlier in~\cite{TPFV22}).
Morozov and Duman~\cite{MD24,MD25} show that information and coding capacities coincide for channels that introduce deletions and insertions with Markovian memory, in the sense that the behavior of the channel on the $i$-th input bit depends on the current state of an underlying stationary ergodic finite state Markov chain (whose states are updated independently of past inputs).
In \cite{MD25} they also give capacity upper bounds for the special case of a deletion channel with Markov memory where the deletion probability applied independently to each input bit varies between a ``low value'' and a ``high value'' according to a 2-state Markov chain (that evolves independently of the channel input).

The models we study are incomparable to those of~\cite{LT21,MD24,MD25}, and our models and capacity theorems generalize those of~\cite{MDK18}.
We discuss the relationship to the model and results of~\cite{MDK18} in more detail.
In~\cite{MDK18}, the channel behavior on the $i$-th input bit $x_i$ may depend only on a \emph{bounded} window of input bits.
In contrast, in our channel model the error distribution for the $i$-th input bit $x_i$ is some function of \emph{the whole input $x$}, satisfying some additional assumptions.
We show in \cref{sec:capture-MDK} that the channel model from~\cite{MDK18} satisfies the assumptions required for the application of our capacity theorems, and so our results generalize the corresponding results of \cite{MDK18}.
Moreover, we show that stationary ergodic sources achieve the information capacity of these channels, which is particularly relevant for constructing efficient capacity-achieving codes.

Furthermore, our framework also captures interesting scenarios that fall outside the scope of~\cite{MDK18}.
In \cref{sec:cap-multi-trace} we show that our framework implies capacity theorems for \emph{multi-trace} channels with correlated synchronization errors, where on input $x$ the receiver learns multiple i.i.d.\ realizations of the channel output $Z(x)$.
This setting is especially relevant in the context of DNA-based data storage systems with nanopore-based sequencing~\cite{CGMR20,BLS20}.
Also, in \cref{sec:assumptions-runlength} we show that our framework includes as special cases deletion channels where the deletion of a bit $x_i$ may depend arbitrarily on the length of the run where $x_i$ is included (in particular, beyond a bounded window around $x_i$).

\paragraph{Efficient coding for channels with synchronization errors.}
There has been significant interest in the design of efficiently encodable and decodable codes for channels with synchronization errors.
We discuss the progress most relevant to our work.
Guruswami and Li~\cite{GL19} and later Con and Shpilka~\cite{con2022improved} obtained efficient codes for the i.i.d.\ binary deletion channel with rate $\Theta(1-d)$, where $d$ is the deletion probability.
Later, Tal, Pfister, Fazeli, and Vardy~\cite{TPFV22,PT21} and later Tal and Arava~\cite{arava2023stronger} (combined with a result from~\cite{LT21}) designed efficient polar codes achieving the capacity of a family of channels with i.i.d.\ insertions, deletions, and substitutions, generalizing the i.i.d.\ deletion channel.
Other constructions of efficient capacity-achieving codes for channels with i.i.d.\ synchronization errors were presented in~\cite{Rub22,PLW22}, which achieve slightly faster decoding and slightly smaller decoding error probability than the polar coding constructions.
Of particular note, the general framework of Pernice, Li, and Wootters~\cite{PLW22} yields efficient capacity-achieving codes for a large class of \emph{repeat channels} -- these are channels that independently replicate each input bit according to some replication distribution over the naturals (e.g., Bernoulli, Poisson, geometric).\footnote{The existence of such efficient capacity-achieving codes does not mean that we now can \emph{determine} the capacity of these channels. We see the contribution of~\cite{TPFV22,PT21,Rub22,PLW22} mainly as turning capacity lower bounds into efficient codes with the corresponding rate.}
This was accomplished by combining a marker-based construction with the capacity theorem of Dobrushin~\cite{Dob67}, which applies to i.i.d.\ repeat channels.

Some works~\cite{CGMR20,BLS20} have studied efficiently encodable and decodable codes for the multi-trace i.i.d.\ deletion channel.
Their focus is on the case where the number of traces is allowed to grow with the blocklength of the code. 
In contrast, in this work we consider the number of traces to be a fixed constant, and we are then interested in devising efficient codes for that fixed number of traces with rate as large as possible.
These two settings are incomparable.
Srinivasavaradhan, Gopi, Pfister, and Yekhanin~\cite{SGPY21} study, among other things, codes for multi-trace channels with i.i.d.\ insertions, deletions, and substitutions with a fixed number of traces.
However, their focus is different from ours and they only provide heuristic reconstruction procedures.

A common feature of the works discussed above is that they only consider channels with i.i.d.\ synchronization errors.
In contrast, we study channels with correlated synchronization errors.
In particular, we obtain a capacity theorem that applies to a wide class of (multi-trace) channels with correlated synchronization errors, which we show can be used to design efficient capacity-achieving codes for synchronization channels with relevant correlations.

\section{Preliminaries}

\subsection{Notation}
We denote random variables by uppercase roman letters such as $X$, $Y$, and $Z$.
In this work, we will only work with random variables supported on discrete sets.
We use $X\to Y\to Z$ to denote the fact that these three random variables form a Markov chain (i.e., $Z$ is conditionally independent of $X$ given $Y$), and write $X\sim Y$ if random variables $X$ and $Y$ follow the same distribution.
We use $\E[X]$ to denote the expected value of a random variable $X$ supported on a subset of $\R$, and $H(X)$ to denote its Shannon entropy.

For a sequence $x=(x_i)_{i\in \N}$, we use $x_m^n$ to denote the subsequence $x_m,x_{m+1},\dots,x_n$.
We use $\log$ to denote the base-$2$ logarithm.
For an integer $n\geq 1$, we write $[n]=\{1,2,\dots,n\}$.

\subsection{Channels}\label{sec:channel-def}

In general, we may define a (discrete-input/discrete-output) channel $Z$ as a \emph{channel sequence} $(Z_n)_{n\in\N}$ where $Z_n$ has countable input and output spaces $\cX_n$ and $\cY_n$, respectively, and, given $x\in\cX_n$ as input, outputs $y\in\cY_n$ with some probability $p_n(y|x)$.
We will focus on channels where $\cX_n=\Sigma_{\mathsf{in}}^n$ for some non-empty finite set $\Sigma_{\mathsf{in}}$ and $\cY_n$ is a countable set.
We call $\Sigma_{\mathsf{in}}$ the \emph{input alphabet} of the channel.
Common choices for $\cY_n$ in examples we study are $\cY_n=\Sigma_{\mathsf{out}}^*$ for some non-empty finite set $\Sigma_{\mathsf{out}}$, where $\Sigma_{\mathsf{out}}^*=\bigcup_{m=0}^\infty \Sigma_{\mathsf{out}}^m$ is the set of finite-length strings over $\Sigma_{\mathsf{out}}$, and $\cY_n=(\Sigma_{\mathsf{out}}^*)^T$ for some fixed integer $T\geq 2$.
The latter setting is relevant for the study of \emph{multi-trace channels}, which we discuss below.

To simplify notation, we drop the subscript $n$ and, given $x\in\cX_n$, denote by $Z(x)\in\cY_n$ the output of channel $Z_n$ on input $x$.
Later on we will impose constraints on the behavior of $Z$ to obtain capacity theorems.
For two inputs $x\in\Sigma_{\mathsf{in}}^n$ and $x'\in\Sigma_{\mathsf{in}}^{n'}$ (of possibly different lengths), we will also write $(Z(x),Z(x'))$ to mean that $Z$ is applied independently to $x$ and $x'$, and that the corresponding outputs can be separated according to the comma that we assume is a special symbol never appearing in an output of $Z$.
This convention extends to more than two inputs.

As mentioned above, a special case of interest are multi-trace channels.
Informally, a $t$-trace channel $Z$ is specified by $T$ channels $Z_1,\dots,Z_T$ with the same input spaces $(\cX_n)_{n\in\N}$.
On input $x$, the channel sends $x$ independently through the $t$ channels $Z_1,\dots,Z_T$ and outputs the resulting $T$ outputs (called \emph{traces}) from these channels.
More formally, on input $x$ the channel $Z$ outputs $(y^{(1)},\dots,y^{(T)})$ with probability $\prod_{i=1}^T p^{(i)}(y^{(i)}|x)$, where $p^{(i)}$ denotes the conditional probability rule of channel $Z_i$.
We will also write the output of $Z$ as $Z(x)=(Z_1(x),\dots,Z_T(x))$.
In particular, the $T$ output traces are conditionally independent given $x$.

\subsection{Entropy and information rates for stochastic processes and notions of capacity}

In this section, we define various types of ``channel capacity'', and state some basic bounds.

\begin{definition}[Entropy rate]
    For a process $X=(X_i)_{i\in\N}$, we define the \emph{entropy rate} of $X$, denoted by $H(X)$, as
    \begin{equation*}
        H(X)= \liminf_{n\to \infty} \frac{H(X_1^n)}{n}.
    \end{equation*}
\end{definition}

\begin{definition}[Information rate]
    For a channel $Z$ and input process $X=(X_i)_{i\in\N}$, we define the \emph{information rate} achievable by $X$ over $Z$ as
    \begin{equation*}
        I(X;Z(X)) = \liminf_{n\to \infty} \frac{I(X_1^n;Z(X_1^n))}{n}.
    \end{equation*}
\end{definition}

\begin{definition}[Information capacity]
    Given a channel $Z$, we define its \emph{information capacity}, denoted by $\ICap(Z)$, as
    \begin{equation*}
        \ICap(Z) = \liminf_{n\to\infty} \sup_{P_{X_1^n}} \frac{I(X_1^n;Z(X_1^n))}{n},
    \end{equation*}
    where the supremum inside the limit is taken over all length-$n$ input probability distributions $X_1^n$.
\end{definition}

We now introduce some useful definitions about stochastic processes.
\begin{definition}[Block-independent process]
    We say that a stochastic process $X=(X_i)_{i\in\N}$ is \emph{block-independent} with blocklength $b$ if for any integer $t\geq 1$ and $n=tb$ we have
    \begin{equation*}
        \Pr[X_1^n = x_1^n]=\prod_{i=1}^t \Pr[X_1^b=x_{(i-1)b+1}^{ib}].
    \end{equation*}
\end{definition}

\begin{definition}[Stationary ergodic process]
    We say that a stochastic process $X=(X_i)_{i\in\N}$ is \emph{stationary} if $(X_{1},\dots,X_{n})\sim (X_{1+\tau},\dots,X_{n+\tau})$ for any integers $n,\tau\in\N$.
    Moreover, we say that $X$ is \emph{stationary ergodic} if $X$ is stationary and for any function $f\in L^1$ we almost surely have
    \begin{equation*}
        \E[f(X_1)] = \lim_{n\to\infty}\frac{1}{n}\sum_{j=1}^n f(X_j).
    \end{equation*}
\end{definition}

\begin{definition}[Stationary capacity]
    Given a channel $Z$, we define its \emph{stationary capacity}, denoted by $\SCap(Z)$, as
    \begin{equation*}
        \SCap(Z) = \sup_X I(X;Z(X)),
    \end{equation*}
    where the supremum is taken over all stationary ergodic input processes $X=(X_i)_{i\in\N}$.
\end{definition}

An input process $X=(X_i)_{i\in\N}$ is \emph{$m$-th order Markov} if 
\begin{equation*}
    \Pr[X_n = x_n | X_{n-1}=x_{n-1},\dots,X_1=x_1]=\Pr[X_n=x_n|X_{n-1}=x_{n-1},\dots,X_{n-m}=x_{n-m}]
\end{equation*}
for any $n$ and $x_1,\dots,x_n$.
\begin{definition}[$m$-th order Markov capacity]
    Given a channel $Z$, we define its \emph{$m$-th order Markov capacity}, denoted by $\SCap^{(m)}(Z)$, as
    \begin{equation*}
        \SCap^{(m)}(Z) = \sup_X I(X;Z(X)),
    \end{equation*}
    where the supremum is taken over all $m$-th order stationary Markov input processes $X=(X_i)_{i\in\N}$.
\end{definition}

\begin{remark}
    \em
    We have $\ICap(Z)\geq \SCap(Z)\geq \SCap^{(m)}(Z)$ for any channel $Z$ and any integer $m\geq 0$.
\end{remark}

Before we define the coding capacity of a channel $Z$, we need some auxiliary definitions.

\begin{definition}[$(n,R,\eps)$-code for a channel]
    Let $Z$ be a channel with non-empty finite input alphabet $\Sigma_{\mathsf{in}}$ and output spaces $(\cY_n)_{n\in\N}$.
    We say that $\cC\subseteq\Sigma_{\mathsf{in}}^n$ is an \emph{$(n,R,\eps)$-code for $Z$} if $|\cC|\geq 2^{Rn}$ and there exists a deterministic function $\Dec:\cY_n\to \Sigma_{\mathsf{in}}^n$ such that $\Pr[\Dec(Z(C))\neq C]\leq \eps$, where $C$ is uniformly distributed over $\cC$ (i.e., the average decoding error probability of $\cC$ is at most $\eps$).  
\end{definition}

\begin{definition}[Achievable rate]
    Let $Z$ be a channel.
    We say that a real number $R>0$ is an \emph{achievable rate for $Z$} if there exists a family of codes $\{\cC_n\}_{n\in\N}$ and an integer $n_0$ such that each $\cC_n$ is an $(n,R_n,\eps_n)$-code for $Z$ with $R_n\geq R$ for all $n\geq n_0$ and $\lim_{n\to\infty} \eps_n=0$.
\end{definition}

\begin{definition}[Coding capacity]
    Given a channel $Z$, we define its \emph{coding capacity}, denoted by $\CCap(Z)$, as the supremum of all $R\geq 0$ that are achievable rates for $Z$.
\end{definition}

\begin{remark}\label{rem:ccap-leq-icap}
    \em
    It follows via Fano's inequality that $\CCap(Z)\leq \ICap(Z)$ for any channel $Z$. See, e.g.,~\cite[Theorem 19.7]{PW24}.
\end{remark}

\subsection{A strengthening of Fekete's lemma}
We will use the following strengthening of Fekete's lemma due to de Bruijn and Erd\H{o}s~\cite{dBE52} (also used in~\cite{MD24}).
See~\cite{FR20} for an excellent discussion on this topic.
\begin{lemma}[Strengthened Fekete's lemma~\cite{dBE52}]\label{lem:fekete}
    Let $(a_n)_{n\in\N}$ be a sequence that is ``almost'' subadditive, in the sense that
    \begin{equation*}
        a_{n+m}\leq a_n+a_m+f(n+m)
    \end{equation*}
    for all $n\leq m\leq 2n$ and some $f$ such that $\sum_{n\in\N} f(n)/n^2$ converges.
    Then, $\lim_{n\to\infty} a_n/n$ exists.
\end{lemma}

\subsection{Concentration inequalities}
In this paper, we shall use the following standard concentration inequalities.

\begin{lemma}[\protect{Multiplicative Chernoff bound; see, e.g., \cite[Theorems 4.4 and 4.5]{mitzenmacher2017probability}}]\label{lem:chernoff}
    Suppose that $X_1, \ldots, X_n$ are i.i.d.\ random variables taking values in $\zo$. Let $X = \sum_{i=1}^n X_i$ and $\mu = \bbE[X]$. Then, for any $0<\alpha < 1$ we have
    \[
    \Pr[X > (1 + \alpha) \mu] < e^{-\frac{\mu  \alpha^2}{3}}
    \]
    and
    \[
    \Pr[X < (1 - \alpha) \mu] < e^{-\frac{\mu  \alpha^2}{2}}.
    \]
\end{lemma}

\begin{lemma}[\protect{Hoeffding's inequality; see, e.g., \cite[Theorem 2.2.6]{vershynin_high-dimensional_2018}}]\label{lem:hoeff}
    Suppose that $X_1, \ldots, X_n$ are independent random variables with finite first and second moments and $a_i\leq X_i \leq b_i$ for $1\leq i \leq n$. Let $X = \sum_{i=1}^n X_i$ and $\mu = \bbE[X]$. Then, for any $t > 0$ we have 
    \[
    \Pr[X - \mu > t] < \exp \left( -\frac{2t^2}{\sum_{i=1}^n (b_i - a_i)^2}\right).
    \]
\end{lemma}

\section{Capacity theorems for channels with input-correlated synchronization errors}\label{sec:shannon-gen}

\subsection{Admissible channels}\label{sec:admissible}

Recall from \cref{sec:channel-def} that we focus on channels with input spaces $\cX_n=\Sigma_{\mathsf{in}}^n$ for some non-empty finite set $\Sigma_{\mathsf{in}}$ and countable output spaces $\cY_n$.
We will show capacity theorems for all such channels satisfying a certain \emph{admissibility} property.
To define this property we must first introduce the notion of a \emph{well-behaved} channel.

\begin{definition}[Well-behaved channel]
    We say that a channel $Z^\star$ is \emph{well-behaved} if it satisfies the following properties: 
    \begin{enumerate}
        \item \textbf{Bounded Output Entropy:} \label[prop]{it:ent-bound}
        There exists a constant $c>0$ such that for every input process $X=(X_i)_{i\in\N}$ and every integer $n\geq 1$ we have $H(Z^\star(X_1^n))\leq cn$.

        \item \textbf{Concatenation:} \label[prop]{it:concat}
        For any process $X=(X_i)_{i\in \N}$ and indices $1\leq m\leq n$, we have that 
        \begin{equation*}
            X_1^{n}\to Z^\star(X_1^{m}),Z^\star(X_{m+1}^{n})\to Z^\star(X_1^{n}).
        \end{equation*}

         \item \textbf{Partition:} \label[prop]{it:partition}

        \begin{enumerate}
            \item \textbf{Prefix/Suffix-Partition:} \label[prop]{it:single-part}
            Fix any integer $\tau\geq 1$. 
            Then, there exists a non-decreasing sequence $\gamma_m=o(m)$ such that $\sum_{m\in\N}\gamma_m/m^2$ converges and for any process $X=(X_i)_{i\in \N}$ and integer $n\geq \tau$ there exist random variables $W_{\mathrm{pre}}$ and $W_{\mathrm{suf}}$ such that (1) $X_1^n\to Z^\star(X_1^{n}),W_{\mathrm{pre}}\to Z^\star(X_1^\tau),Z^\star(X_{\tau+1}^n)$, (2) $X_1^n\to Z^\star(X_1^{n}),W_{\mathrm{suf}}\to Z^\star(X_1^{n-\tau}),Z^\star(X_{n-\tau+1}^n)$, and (3) $H(W_{\mathrm{pre}}),H(W_{\mathrm{suf}})\leq \gamma_n$.

            \item \textbf{Amortized Block-Partition:} \label[prop]{it:multi-part}
            There exists a sequence $\alpha_{m}=o(m)$ such that for any process $X=(X_i)_{i\in \N}$, blocklength $b$, and number of blocks $t$ there exists a random variable $W$ such that $X_1^{tb}\to Z^\star(X_1^{tb}),W\to Z^\star(X_1^{b}),\dots,Z^\star(X_{(t-1)b+1}^{tb})$, and $H(W)\leq t\cdot \alpha_{b}$.
        \end{enumerate}
 
        \item \textbf{Amortized Preimage Size:} \label[prop]{it:preimg} 
        There exists a sequence $\beta_m = o(m)$ such that for any process $X=(X_i)_{i\in \N}$, blocklength $b$, and $n=tb+r$ for integers $t\geq 1$ and $0\leq r<b$, there exists a random variable $Y$ with countable support such that\footnote{If $r=0$ then $X_{tb+1}^{tb+r}$ is the empty string, and we take $Z^\star(X_{tb+1}^{tb+r})$ to be the empty string as well.}
        \begin{equation*}
            X_1^{n}\to Z^\star(X_1^b),\dots,Z^\star(X_{(t-1)b+1}^{tb}),Z^\star(X_{tb+1}^{tb+r})\to Y
        \end{equation*}
        and a deterministic function $\phi$ such that
        \begin{equation*}
            Z^\star(X_1^{n})=\phi(Z^\star(X_1^b),\dots,Z^\star(X_{(t-1)b+1}^{tb}),Z^\star(X_{tb+1}^{tb+r}),Y)
        \end{equation*}
        and $\log |\phi^{-1}(z)|\leq t\cdot \beta_{b}$ for all $z$. 
    \end{enumerate}
\end{definition}

We are now ready to define our notion of admissibility.
\begin{definition}[Admissible channel]
    We say that a channel $Z$ is \emph{admissible} if there exists a well-behaved channel $Z^\star$ such that for any process $X=(X_i)_{i\in\N}$ we have
    \begin{equation*}
        I(X;Z(X))=I(X;Z^\star(X)).
    \end{equation*}
    If this holds, we may also say that $Z$ is an \emph{admissible channel with respect to $Z^\star$}.
\end{definition}

We will prove the following results for admissible channels.

\begin{theorem}[Capacity theorem for admissible channels]\label{thm:cap-gen}
    Suppose that $Z$ is an admissible channel.
    Then,
    \begin{equation*}
        \ICap(Z)=\SCap(Z)=\lim_{m\to\infty} \SCap^{(m)}(Z)=\CCap(Z).
    \end{equation*}
\end{theorem}

The proof of \cref{thm:cap-gen} proceeds via a series of theorems, where we adapt and generalize approaches from~\cite{Dob67,LT21,MD24}. The equality $\ICap(Z)=\SCap(Z)$ corresponds to \cref{thm:icap-scap} in \cref{sec:icap-scap}.
The equality $\ICap(Z)=\lim_{m\to\infty} \SCap^{(m)}(Z)$ corresponds to \cref{thm:icap-mcap} in \cref{sec:icap-mcap}.
The equality $\ICap(Z)=\CCap(Z)$ corresponds to \cref{thm:icap-ccap} in \cref{sec:good-code}.

The next result follows in a standard manner from \cref{thm:cap-gen} using the approach of~\cite{PLW22}, and is relevant for the construction of efficient capacity-achieving codes for admissible channels.
For simplicity, we present it for channels with binary input alphabet.
We prove this result in \cref{sec:dense-codes}.

\begin{restatable}[Dense capacity-achieving codes for admissible channels]{theorem}{densecodegen}\label{thm:dense-code-gen}
    Let $Z$ be an admissible channel with binary input alphabet such that $\ICap(Z)>0$.
    Then, for any $\eps,\zeta>0$ there exist $\gamma\in(0,1/2)$ and integers $b=b(\eps,\zeta)$ and $t(\eps,\zeta)$ depending only on $\eps$ and $\zeta$ such that for all $t\geq t(\eps,\zeta)$ there exists a code $\cC$ with blocklength $n=t\cdot b$, rate $R\geq \ICap(Z)-\eps$, and maximal decoding error probability $\eps$ over $Z$ such that for all codewords $c\in\cC$ we have $\gamma\zeta n\leq w(c_i^{i+\zeta n})\leq (1-\gamma)\zeta n$ for all $i\in[(1-\zeta)n]$, where $w(\cdot)$ denotes the Hamming weight.
\end{restatable}

The properties for admissibility laid out above are sufficient for \cref{thm:cap-gen,thm:dense-code-gen} to hold, but it is conceivable that they are not necessary.
We leave it as an interesting direction for future work to simplify the set of sufficient properties.

\subsection{Existence of relevant limits for admissible channels}

Let $Z$ be an arbitrary admissible channel with respect to a well-behaved channel $Z^\star$.
Via applications of Fekete's lemma (\cref{lem:fekete}), we begin by showing that the limit inferior in the definitions of capacities and information rates can be replaced by a limit.

\begin{lemma}\label{lem:icap-exist-gen}
    If $Z$ is admissible with respect to a well-behaved channel $Z^\star$, then
    \begin{equation*}
        \ICap(Z)=\ICap(Z^\star)=\lim_{n\to \infty} \sup_{X_1^n} \frac{I(X_1^n;Z^\star(X_1^n))}{n},
    \end{equation*}
    and the limit on the right-hand side exists.
\end{lemma}
\begin{proof}
    By \cref{lem:fekete}, it suffices to show that the sequence
    \begin{equation*}
        a_n = \sup_{X_1^n} I(X_1^n;Z^\star(X_1^n))
    \end{equation*}
    is subadditive, i.e., $a_{n+m}\leq a_n+a_m$ for all $n,m$.
    Then,
    \begin{align*}
        a_{n+m}&=\sup_{X_1^{n+m}} I(X_1^{n+m};Z^\star(X_1^{n+m})) \\
        &\leq \sup_{X_1^{n+m}} I(X_1^{n+m};Z^\star(X_1^{n}),Z^\star(X_{n+1}^{n+m}))\\
        &= \sup_{X_1^{n+m}} [I(X_1^{n+m};Z^\star(X_1^{n})) + I(X_{1}^{n+m};Z^\star(X_{n+1}^{n+m})|Z^\star(X_1^{n}))]\\
        &= \sup_{X_1^{n+m}} [I(X_1^{n};Z^\star(X_1^{n})) + I(X_{1}^{n+m};Z^\star(X_{n+1}^{n+m})|Z^\star(X_1^{n}))]\\
        &\leq \sup_{X_1^{n+m}} [I(X_1^{n};Z^\star(X_1^{n})) + I(X_{n+1}^{n+m};Z^\star(X_{n+1}^{n+m}))]\\
        &\leq \sup_{X_1^n} I(X_1^{n};Z^\star(X_1^{n})) + \sup_{X_{n+1}^{n+m}} I(X_{n+1}^{n+m};Z^\star(X_{n+1}^{n+m}))\\
        &=a_n+a_m.
    \end{align*}
    The first inequality uses the Concatenation Property (\cref{it:concat}).
    The second inequality uses the fact that $X_1^{n},X_{n+1}^{n+m}\to X_{n+1}^{n+m} \to Z^\star(X_{n+1}^{n+m})$, valid for all channels as in \cref{sec:channel-def}, and so
    \begin{align*}
        I(X_1^{n},X_{n+1}^{n+m};Z^\star(X_{n+1}^{n+m})|Z^\star(X_1^{n})) &= H(Z^\star(X_{n+1}^{n+m})|Z^\star(X_1^{n})) - H(Z^\star(X_{n+1}^{n+m})| X_1^{n},X_{n+1}^{n+m},Z^\star(X_1^{n}))\\
        &= H(Z^\star(X_{n+1}^{n+m})|Z^\star(X_1^{n})) - H(Z^\star(X_{n+1}^{n+m})|X_{n+1}^{n+m})\\
        &\leq H(Z^\star(X_{n+1}^{n+m})) - H(Z^\star(X_{n+1}^{n+m})|X_{n+1}^{n+m})\\
        &=I(X_{n+1}^{n+m};Z^\star(X_{n+1}^{n+m})).
    \end{align*}
    The third inequality holds because the quantity on the fifth line is maximized by taking $X_1^n$ and $X_{n+1}^{n+m}$ to be independent, since $Z^\star$ is acting independently on each of $X_1^n$ and $X_{n+1}^{n+m}$.
\end{proof}

\begin{lemma}\label{lem:scap-exist-gen}
    Let $X$ be an arbitrary stationary process and $Z$ an admissible channel with respect to a well-behaved channel $Z^\star$.
    Then,
    \begin{equation*}
        I(X;Z(X)) = I(X;Z^\star(X)) = \lim_{n\to\infty} \frac{I(X_1^n;Z^\star(X_1^n))}{n},
    \end{equation*}
    the limit on the right-hand side exists, and $\SCap(Z)=\SCap(Z^\star)$.
\end{lemma}
\begin{proof}
    Again, by \cref{lem:fekete} it suffices to show that $a_n = I(X_1^n; Z^\star(X_1^n))$
    is a subadditive sequence.
    We have
    \begin{align*}
        a_{n+m}&= I(X_1^{n+m};Z^\star(X_1^{n+m})) \\
        &\leq I(X_1^{n+m};Z^\star(X_1^{n}),Z^\star(X_{n+1}^{n+m}))\\
        &= I(X_1^{n},X_{n+1}^{n+m};Z^\star(X_1^{n})) + I(X_1^{n},X_{n+1}^{n+m};Z^\star(X_{n+1}^{n+m})|Z^\star(X_1^{n}))\\
        &= I(X_1^{n};Z^\star(X_1^{n})) + I(X_1^{n},X_{n+1}^{n+m};Z^\star(X_{n+1}^{n+m})|Z^\star(X_1^{n}))\\
        &\leq I(X_1^{n};Z^\star(X_1^{n})) + I(X_{n+1}^{n+m};Z^\star(X_{n+1}^{n+m}))\\
        &= I(X_1^{n};Z^\star(X_1^{n})) + I(X_{1}^{m};Z^\star(X_{1}^{m}))\\
        &=a_n+a_m,
    \end{align*}
    as desired.
    The first inequality uses the Concatenation Property (\cref{it:concat}).
    The second equality follows from the chain rule for mutual information.
    The third equality uses the fact that $X_1^{n},X_{n+1}^{n+m}\to X_1^{n} \to Z^\star(X_1^{n})$.
    The second inequality uses the fact that $X_1^{n},X_{n+1}^{n+m}\to X_{n+1}^{n+m} \to Z^\star(X_{n+1}^{n+m})$, analogously to the proof of \cref{lem:icap-exist-gen}. 
    The fourth equality uses the fact that $X$ is stationary, and so $X_{n+1}^{n+m}\sim X_{1}^{m}$.
\end{proof}

\begin{lemma}\label{lem:blockind-exist}
    Let $X$ be a block independent process and suppose that $Z$ is admissible with respect to a well-behaved channel $Z^\star$.
    Then, 
    \begin{equation*}
        I(X;Z(X)) = I(X;Z^\star(X)) =  \lim_{n\to\infty} \frac{I(X_1^n;Z^\star(X_1^n))}{n},
    \end{equation*}
    and the limit on the right-hand side exists.
\end{lemma}
\begin{proof}
    Consider the sequence $a_n = I(X_1^n;Z^\star(X_1^n))$.
    Let $b$ be the blocklength of $X$.
    Then, by the Concatenation Property (\cref{it:concat}), for any $m,n>b$ we have
    \begin{align*}
        a_{n+m} &= I(X_1^{n+m};Z^\star(X_1^{n+m}))\\
        & \leq I(X_1^{n+m};Z^\star(X_1^{n}),Z^\star(X_{n+1}^{n+m}))\\
        & \leq I(X_1^{n};Z^\star(X_1^{n})) + I(X_{n+1}^{n+m};Z^\star(X_{n+1}^{n+m})).
    \end{align*}
    Construct $X'$ by trimming bits from the beginning of $X_{n+1}^{n+m}$ so that $X'=X_{tb+1}^{n+m}$ for an integer $t$ for which $n+1 -(tb+1)\leq b$.
    If $m'$ denotes the length of $X'$, from the block independence of $X$ we get that
    \begin{equation*}
        I(X';Z^\star(X')) = I(X_1^{m'};Z^\star(X_1^{m'})).
    \end{equation*}
    Since we trim $r<b$ symbols from $X_{n+1}^{n+m}$ to obtain $X'$, we have
    \begin{align*}
        I(X_{n+1}^{n+m};Z^\star(X_{n+1}^{n+m}))&\leq I(X_{n+1}^{n+m};Z^\star(X_{n+1}^{tb}),Z^\star(X_{tb+1}^{n+m}))\\
        & \leq I(X_{tb+1}^{n+m};Z^\star(X_{t b+1}^{n+m})) + H(Z^\star(X_{n+1}^{tb}))\\
        & \leq I(X_{tb+1}^{n+m};Z^\star(X_{t b+1}^{n+m})) + cb\\
        & = I(X_1^{m'};Z^\star(X_1^{m'}))+ cb\\
        & \leq I(X_1^{m};Z^\star(X_1^{m'}),Z^\star(X_{m'+1}^{m}))+ cb\\
        &\leq I(X_1^{m};Z^\star(X_1^{m}))+\gamma_m+ cb
    \end{align*}
    for a non-decreasing sequence $\gamma_m=o(m)$ such that $\sum_{m\in\N}\gamma_m/m^2$ converges.
    The first inequality uses the Concatenation Property (\cref{it:concat}).
    The third inequality uses the Bounded Entropy Property (\cref{it:ent-bound}), since $H(Z^\star(X_{n+1}^{tb}))\leq c(tb-(n+1))\leq cb$ for a fixed constant $c>0$.
    The first equality uses the block independence of $X$.
    The fourth inequality follows from the chain rule for mutual information.
    The fifth inequality uses the Prefix/Suffix-Partition Property (\cref{it:single-part}) with $\tau=r<b$.

    Therefore, we conclude that for any $b<n\leq m\leq 2n$ we have $a_{n+m}\leq a_n +a_m + f(n+m)$, where $f(n) = \gamma_n+cb$ (here, we use that $\gamma_m$ is non-decreasing).
    Since $\sum_{n\in\N}f(n)/n^2$ converges because $\sum_{n\in\N} \gamma_n/n^2$ converges by  the Prefix/Suffix-Partition Property (\cref{it:single-part}), \cref{lem:fekete} implies that $\lim_{n\to\infty} a_n/n$ exists, as desired.
\end{proof}

\subsection{Information capacity of admissible channels is achieved by stationary ergodic process}\label{sec:icap-scap}

We show the following.
\begin{theorem}\label{thm:icap-scap}
    Suppose that the channel $Z$ is admissible. Then,
    \begin{equation*}
        \ICap(Z)=\SCap(Z).
    \end{equation*}
\end{theorem}
Suppose that $Z$ is admissible with respect to a well-behaved channel $Z^\star$.
Then, it suffices to show that \cref{thm:icap-scap} holds for $Z^\star$.
Fix an arbitrary $\eps>0$. 
Recalling \cref{lem:icap-exist-gen}, let $b$ be a sufficiently large integer so that
\begin{equation*}
    \sup_{\tilde{X}_1^b} \frac{I(\tilde X_1^b;Z^\star(\tilde X_1^b))}{b} \geq \lim_{n\to\infty} \sup_{\tilde{X}_1^n} \frac{I(\tilde X_1^n;Z^\star(\tilde X_1^n))}{n} - \eps = \ICap(Z^\star)-\eps.
\end{equation*}
Furthermore, let $X_1^b$ be such that
\begin{equation}\label{eq:guarantee-X-gen}
    \frac{I(X_1^b;Z^\star(X_1^b))}{b}\geq \sup_{\tilde{X}_1^b} \frac{I(\tilde X_1^b;Z^\star(\tilde X_1^b))}{b}-\eps \geq \ICap(Z^\star)-2\eps.
\end{equation}
Let $p_{X_1^b}(\cdot)$ denote the correspoding PMF.
Using the approach of~\cite{Fei59,LT21}, consider the following process $\overline{X}$.
First, define the block independent process $\hat{X}=\{\hat{X}_i\}_{i\in\N}$ with blocklength $b$ and probability mass function
\begin{equation*}
    p_{\hat{X}_1^{tb}}(x_1^{tb}) = \prod_{i=1}^t p_{X_1^b}(x_{(i-1)b+1}^{ib}).
\end{equation*}
Write $\hat{X}^{[j]}_i = \hat{X}_{i+j}$.
Let $V$ be uniformly distributed over $\{0,1,\dots,b-1\}$, and set $\overline{X}_i = \hat{X}^{[V]}_i = \hat{X}_{i+V}$ for every $i\in\N$.
Then, it holds that $\overline{X}$ is stationary and ergodic.

By \cref{lem:scap-exist-gen}, we know that $I(\overline{X};Z^\star(\overline{X}))=\lim_{n\to\infty} \frac{I(\overline{X}_1^n;Z^\star(\overline{X}_1^n))}{n}$ since $\overline{X}$ is stationary.
We will show that
\begin{equation}\label{eq:maingoal-cap-gen}
    I(\overline{X};Z^\star(\overline{X})) \geq \frac{I(X_1^b;Z^\star(X_1^b))}{b} - \eps.
\end{equation}
Combined with \cref{eq:guarantee-X-gen}, this implies that
\begin{equation*}
    \SCap(Z^\star)\geq I(\overline{X};Z^\star(\overline{X})) \geq \ICap(Z^\star)-3\eps.
\end{equation*}
Since $\eps$ was arbitrary, it follows that $\SCap(Z^\star) = \ICap(Z^\star)$, and so also $\SCap(Z)=\ICap(Z)$.
This yields \cref{thm:icap-scap}.

It remains to show \cref{eq:maingoal-cap-gen}.
We do this by a combination of two lemmas.

\begin{lemma}\label{lem:convex-comb-gen}
    We have
    \begin{equation*}
        I(\overline{X}; Z^\star(\overline{X})) = \sum_{j=0}^{b-1} \frac{1}{b} I(\hat{X}^{[j]};Z^\star(\hat{X}^{[j]})) = I(\hat{X};Z^\star(\hat{X})),
    \end{equation*}
    and these limits exist.
\end{lemma}

\begin{proof}
    
    First, we prove that
    \begin{equation*}
        I(\hat{X}^{[j]};Z^\star(\hat{X}^{[j]}))=I(\hat{X};Z^\star(\hat{X}))    
    \end{equation*}
    for all $j\in\{0,1,\dots,b-1\}$ (in particular, these quantities exist for all $0\leq j<b$, since $\hat X$ is stationary).
    For the sake of exposition we focus on $j=1$.
    The argument is analogous for other choices of $j$.
    First, note that by the Prefix/Suffix-Partition Property (\cref{it:single-part}) with $\tau=b-1$ there exists a sequence $\gamma_m$ such that $\gamma_m=o(m)$ and $\sum_{m\in\N}\gamma_m/m^2$ converges and a random variable $W$ such that $H(W)\leq \gamma_{tb}$ and
    $\hat X_1^{[1]tb-1}\to Z^\star(\hat X_1^{[1]tb-1}),W\to Z^\star(\hat X_1^{[1]b-1}), Z^\star(\hat X_{b}^{[1] tb-1})$.
    Furthermore, by the Concatenation Property (\cref{it:concat}), $\hat{X}_1^{[1] tb-1}\to Z^\star(\hat X_1^{[1]b-1}), Z^\star(\hat X_{b}^{[1] tb-1})\to Z^\star(\hat X_1^{[1]tb-1})$.
    This means that
    \begin{equation}\label{eq:ub-MI-Xhat}
        I(\hat X_1^{[1]tb-1};Z^\star(\hat X_1^{[1] tb-1}))\leq I(\hat X_1^{[1] tb-1};Z^\star(\hat X_1^{[1] b-1}), Z^\star(\hat X_{b}^{[1] tb-1}))    
    \end{equation}
    and
    \begin{align}
        I(\hat X_1^{[1]tb-1};Z^\star(\hat X_1^{[1] tb-1}))&\geq I(\hat X_1^{[1] tb-1};Z^\star(\hat X_1^{[1] b-1}), Z^\star(\hat X_{b}^{[1] tb-1}))-H(W)\nonumber\\
        &\geq I(\hat X_1^{n+m};Z^\star(\hat X_1^n), Z^\star(\hat X_{n+1}^{n+m}))-\gamma_{tb}.\label{eq:lb-MI-Xhat}
    \end{align}
    Therefore,
    \begin{align*}
        I(\hat{X}^{[1] tb-1}_1;Z^\star(\hat{X}^{[1] tb-1}_1))
        &\geq I(\hat{X}^{[1] b-1}_1,\hat{X}^{[1] tb-1}_b;Z^\star(\hat{X}^{[1] b-1}_1), Z^\star(\hat{X}^{[1] tb-1}_b))- \gamma_{tb} \\
        & = I(\hat{X}^{[1] b-1}_1; Z^\star(\hat{X}^{[1] b-1}_1), Z^\star(\hat{X}^{[1] tb-1}_b))\\
        & + I(\hat{X}^{[1] tb-1}_b; Z^\star(\hat{X}^{[1] b-1}_1), Z^\star(\hat{X}^{[1] tb-1}_b)|\hat{X}^{[1] b-1}_1)- \gamma_{tb} \\
        & \geq I(\hat{X}^{[1] tb-1}_b; Z^\star(\hat{X}^{[1] b-1}_1), Z^\star(\hat{X}^{[1] tb-1}_b)|\hat{X}^{[1] b-1}_1) - \gamma_{tb}\\
        & = I(\hat{X}^{[1] tb-1}_b; Z^\star(\hat{X}^{[1] tb-1}_b))-\gamma_{tb}\\
        &= I(\hat{X}^{tb}_{b+1}; Z^\star(\hat{X}^{tb}_{b+1}))- \gamma_{tb}\\
        &=I(\hat{X}^{(t-1)b}_{1}; Z^\star(\hat{X}^{(t-1)b}_{1}))-\gamma_{tb}.
    \end{align*}
    The first inequality uses \cref{eq:lb-MI-Xhat}.
    The second and fourth equalities use the fact that $\hat{X}$ is block-independent with blocklength $b$, and so $\hat{X}^{[1] tb-1}_b$ is independent of $\hat{X}^{[1] b-1}_1$ and is identically distributed to $\hat{X}^{(t-1)b}_{1}$.
    Similarly,
    \begin{align*}
        I(\hat{X}^{[1] tb-1}_1;Z^\star(\hat{X}^{[1] tb-1}_1))
        &\leq I(\hat{X}^{[1] b-1}_1,\hat{X}^{[1] tb-1}_b;Z^\star(\hat{X}^{[1] b-1}_1), Z^\star(\hat{X}^{[1] tb-1}_b)) \\
        & = I(\hat{X}^{[1] b-1}_1; Z^\star(\hat{X}^{[1] b-1}_1), Z^\star(\hat{X}^{[1] tb-1}_b))\\
        & + I(\hat{X}^{[1] tb-1}_b; Z^\star(\hat{X}^{[1] b-1}_1), Z^\star(\hat{X}^{[1] tb-1}_b)|\hat{X}^{[1] b-1}_1) \\
        & \leq I(\hat{X}^{[1] tb-1}_b; Z^\star(\hat{X}^{[1] b-1}_1), Z^\star(\hat{X}^{[1] tb-1}_b)|\hat{X}^{[1] b-1}_1) +b\log|\Sigmain|\\
        & = I(\hat{X}^{[1] tb-1}_b; Z^\star(\hat{X}^{[1] tb-1}_b))+ b\log|\Sigmain|\\
        &= I(\hat{X}^{tb}_{b+1}; Z^\star(\hat{X}^{tb}_{b+1}))+b\log|\Sigmain|\\
        &=I(\hat{X}^{(t-1)b}_{1}; Z^\star(\hat{X}^{(t-1)b}_{1}))+b\log|\Sigmain|,
    \end{align*}
    where the first inequality uses \cref{eq:ub-MI-Xhat} and the second inequality uses the fact that $H(\hat{X}^{[1] b-1}_1)\leq b\log|\Sigmain|$.
    As a result, we conclude that (below, we write $a\pm \delta$ for  a real number in the interval $[a-\delta,a+\delta]$)
    \begin{align*}
        I(\hat{X}^{[1]};Z^\star(\hat{X}^{[1]})) &= \lim_{t\to\infty} \frac{I(\hat{X}^{[1] tb-1}_1;Z^\star(\hat{X}^{[1] tb-1}_1))}{tb-1}\\
        &=\lim_{t\to\infty} \frac{I(\hat{X}^{(t-1)b}_{1}; Z^\star(\hat{X}^{(t-1)b}_{1}))\pm (b\log|\Sigmain|+\gamma_{tb})}{tb-1}\\
        & = \lim_{t\to\infty} \frac{(t-1)b}{tb-1}\cdot \frac{I(\hat{X}^{(t-1)b}_{1}; Z^\star(\hat{X}^{(t-1)b}_{1}))\pm (b\log|\Sigmain|+\gamma_{tb})}{(t-1)b}\\
        & = \lim_{t\to\infty} \frac{I(\hat{X}^{(t-1)b}_{1}; Z^\star(\hat{X}^{(t-1)b}_{1}))}{(t-1)b}\\
        & = I(\hat{X};Z(\hat{X})).
    \end{align*}
    The fourth equality uses the fact that $|\Sigmain|$ is a finite constant and $\lim_{t\to\infty} \frac{\gamma_{tb}}{(t-1)b}=0$, since $\gamma_{tb}=o(tb)$.
    Furthermore, \cref{lem:blockind-exist} guarantees that $I(\hat{X};Z(\hat{X}))$ exists.

    It remains to see the leftmost inequality of the lemma statement.
    First, note that
    \begin{equation}\label{eq:cond-V}
        I(\overline X;Z^\star(\overline X)) = I(\overline X;Z^\star(\overline X)|V) = \lim_{n\to \infty} \frac{I(\overline X_1^n;Z^\star(\overline X_1^n)|V)}{n}.
    \end{equation}
    This holds since $H(V)\leq \log b$ and $b$ is a fixed constant.
    Furthermore,
    \begin{align}
        I(\overline X;Z^\star(\overline X)|V) &= \lim_{n\to \infty} \frac{I(\overline X_1^n;Z^\star(\overline X_1^n)|V)}{n}\nonumber\\
        &=\lim_{n\to \infty} \frac{\frac{1}{b}\sum_{j=0}^{b-1} I(\overline X_1^n;Z^\star(\overline X_1^n)|V=j)}{n}\nonumber\\
        &=\lim_{n\to \infty} \frac{\frac{1}{b}\sum_{j=0}^{b-1} I(\hat X_1^{[j] n};Z^\star(\hat X_1^{[j] n}))}{n}.\label{eq:conv-comb-mi}
    \end{align}
    As we saw above, the limits
    \begin{equation*}
        I(\hat X^{[j]};Z^\star(\hat X^{[j]})) = \lim_{n\to\infty} \frac{I(\hat X_1^{[j] n};Z^\star(\hat X_1^{[j] n}))}{n}
    \end{equation*}
    exist and equal $I(\hat X;Z^\star(\hat X))$ for all $j$.
    Therefore, since the sum over $j$ is finite, we can swap limit and sum and conclude from \cref{eq:cond-V,eq:conv-comb-mi} that
    \begin{align*}
        I(\overline X;Z^\star(\overline X)) &= I(\overline X;Z^\star(\overline X)|V) \\
        &= \frac{1}{b}\sum_{j=0}^{b-1} \lim_{n\to \infty} \frac{I(\hat X_1^{[j] n};Z^\star(\hat X_1^{[j] n}))}{n} \\
        &= \frac{1}{b}\sum_{j=0}^{b-1}  I(\hat X^{[j]};Z^\star(\hat X^{[j]})) \\
        &= I(\hat X;Z^\star(\hat X)),
    \end{align*}
    as desired.
\end{proof}

\begin{lemma}\label{lem:block-ind-rate-gen}
    Suppose that $\hat X$ is a block independent process with blocklength $b$.
    Then, we have
    \begin{equation*}
        \left| I(\hat{X};Z^\star(\hat{X})) - \frac{I(\hat X_1^b;Z^\star(\hat X_1^b))}{b}\right| \leq \alpha_b/b,
    \end{equation*}
    where $\lim_{b\to\infty} \frac{\alpha_b}{b}=0$.
\end{lemma}
\begin{proof}
    By the Amortized Block-Partition Property (\cref{it:multi-part}), there exists a random variable $W$ with $H(W)\leq t\cdot \alpha_b$ such that $\hat{X}_1^{tb}\to Z^\star(\hat{X}_1^{tb}),W\to Z^\star(\hat{X}_1^{b}),Z^\star(\hat{X}_{b+1}^{2b}),\dots,Z^\star(\hat{X}_{(t-1)b+1}^{tb})$.
    Therefore,
    \begin{align*}
        I(\hat{X}_1^{tb};Z^\star(\hat{X}_1^{tb})) &\geq I(\hat{X}_1^{tb};Z^\star(\hat{X}_1^{b}),Z^\star(\hat{X}_{b+1}^{2b}),\dots,Z^\star(\hat{X}_{(t-1)b+1}^{tb})) -t\cdot \alpha_b\\
        &= \sum_{i=1}^t I(\hat{X}_{(i-1)b+1}^{ib};Z^\star(\hat{X}_{(i-1)b+1}^{ib})) - t\cdot \alpha_b\\
        & = \sum_{i=1}^t I(\hat X_{1}^{b};Z^\star(\hat X_{1}^{b})) - t\cdot \alpha_b\\
        & = t(I(\hat X_{1}^{b};Z^\star(\hat X_{1}^{b}))-\alpha_b).
    \end{align*}
    Consequently,
    \begin{align*}
        I(\hat{X};Z^\star(\hat{X})) &= \lim_{t\to\infty} \frac{I(\hat{X}_1^{tb};Z^\star(\hat{X}_1^{tb}))}{tb}
        \geq \frac{I(\hat X_{1}^{b};Z^\star(\hat X_{1}^{b}))-\alpha_b}{b}.
    \end{align*}
    On the other hand, by the Concatenation Property (\cref{it:concat}) we have
    \begin{equation*}
        \frac{I(\hat X_{1}^{b};Z^\star(\hat X_{1}^{b}))}{b} = \frac{I(\hat{X}_1^{tb};Z^\star(\hat{X}_1^{b}),Z^\star(\hat{X}_{b+1}^{2b}),\dots,Z^\star(\hat{X}_{(t-1)b+1}^{tb}))}{tb} \geq \frac{I(\hat{X}_1^{tb};Z^\star(\hat{X}_1^{tb}))}{tb}
    \end{equation*}
    for all $t$ and $b$, and so $\frac{I(X_{1}^{b};Z^\star(X_{1}^{b}))}{b}\geq I(\hat{X};Z^\star(\hat{X}))$.
\end{proof}

\subsection{Information capacity of admissible channels is achieved by Markov process}\label{sec:icap-mcap}

We use \cref{thm:icap-scap} to show the following.
\begin{theorem}\label{thm:icap-mcap}
    Suppose that the channel $Z$ is admissible.
    Then,
    \begin{equation*}
        \ICap(Z) = \lim_{m\to\infty} \SCap^{(m)}(Z).
    \end{equation*}
\end{theorem}
\begin{proof}
    Suppose that $Z$ is admissible with respect to a well-behaved channel $Z^\star$.
    Given $\eps>0$, let $X$ be a stationary ergodic input process and $b$ a large enough integer such that
    \begin{equation*}
        \frac{I(X_1^b;Z^\star(X_1^b))}{b}\geq \ICap(Z)-\eps.
    \end{equation*}
    Such a process $X$ and integer $b$ are guaranteed to exist by \cref{thm:icap-scap}, since $Z$ is admissible with respect to $Z^\star$.
    Let $\hat X$ be the stationary $(b-1)$-th order Markov process satisfying
    \begin{equation*}
        p_{\hat{X}_1^b}(x_1^b) = p_{X_1^b}(x_1^b).
    \end{equation*}
    We will show that 
    \begin{equation}\label{eq:maingoal}
        I(\hat X;Z(\hat X))=I(\hat X;Z^\star(\hat X))\geq \frac{I(X_1^b;Z^\star(X_1^b))}{b} - \frac{\alpha_b}{b},
    \end{equation}
    with $\lim_{b\to\infty} \frac{\alpha_b}{b}=0$, where the first equality holds since $Z$ is admissible with respect to $Z^\star$.
    Taking $b$ to be large enough, we get that $I(\hat X;Z(\hat X))\geq \ICap(Z)-2\eps$, and, since $\eps$ was arbitrary, the theorem statement follows.

    Since $\hat{X}$ is a Markov process, we have $H(\hat X)\geq H(X)$.
    Therefore, it is enough to prove that
    \begin{equation}\label{eq:cond-entropy-goal}
        H(\hat X|Z^\star(\hat X))\leq \frac{H(X_1^b|Z^\star(X_1^b))}{b}+\frac{\alpha_b}{b}.
    \end{equation}
    We have
    \begin{align*}
        H(\hat X | Z^\star(\hat X)) &= \lim_{t\to\infty} \frac{H(\hat X_1^{tb}|Z^\star(\hat X_1^{tb}))}{tb}\\
        &\leq \lim_{t\to\infty} \frac{H(\hat X_1^{tb}|Z^\star(\hat X_1^{b}),\dots,Z^\star(X_{(t-1)b+1}^{tb}))+t\cdot \alpha_b}{tb}\\
        &= \lim_{t\to\infty} \frac{\sum_{i=1}^t H(\hat X_{(i-1)b+1}^{ib}|\hat X_{1}^{(i-1)b},Z^\star(\hat X_1^{b}),\dots,Z^\star(X_{(t-1)b+1}^{tb}))+t\cdot \alpha_b}{tb}\\
        &\leq \lim_{t\to\infty} \frac{\sum_{i=1}^t H(\hat X_{(i-1)b+1}^{ib}|Z^\star(\hat X_{(i-1)b+1}^{ib}))+t\cdot \alpha_b}{tb}\\
        &=\lim_{t\to\infty} \frac{\sum_{i=1}^t H(\hat X_{1}^{b}|Z^\star(\hat X_{1}^{b}))+t\cdot \alpha_b}{tb}\\
        &=\frac{H(X_1^b|Z^\star(X_1^b))}{b}+\frac{\alpha_b}{b},
    \end{align*}
    as desired.
    The first inequality uses the Amortized Block-Partition Property (\cref{it:multi-part}).
    The second equality uses the chain rule for conditional entropy.
    The second inequality holds since further conditioning does not increase entropy.
    The third equality uses the stationarity of $X$.
\end{proof}

\subsection{Coding capacity equals information capacity, and existence of dense codes from stationary ergodic processes}\label{sec:good-code}

In this section, we begin by establishing suitable convergence of the information density of block-independent processes to their information rate for admissible channels.
We use this result in two ways. First, we use it to show that $\CCap(Z)=\ICap(Z)$. Second, focusing on admissible channels with binary input for simplicity, we apply this result to stationary ergodic processes (which by \cref{thm:icap-scap} achieve information capacity on admissible channels) to conclude that there exist ``dense'' codes $\cC$ that achieve capacity on $Z$, where ``dense'' means that every short substring of $c\in\cC$ contains a decent fraction of $1$s.
As already pointed out in~\cite{PLW22}, this property is relevant for the construction of efficient capacity-achieving codes for these channels.

\subsubsection{Convergence of information density for block-independent process}
For two random variables $X,Y$, we define their \emph{information density} $i_{X,Y}$ as
\begin{equation*}
    i_{X,Y}(x,y) = \log\left(\frac{p_{XY}(x,y)}{p_X(x)\cdot p_Y(y)}\right).
\end{equation*}
Note that $\E_{(x,y)\sim p_{XY}}[i_{X,Y}(x,y)]=I(X;Y)$.
We show the following.
\begin{theorem}\label{thm:conv-info-density-gen}
    Let $Z$ be an admissible channel with respect to a well-behaved channel $Z^\star$.
    Fix $\eps>0$ and let $(\alpha_n)_{n\in\N}$ and $(\beta_n)_{n\in\N}$ be the sequences guaranteed by the Amortized Block-Partition Property (\cref{it:multi-part}) and the Amortized Preimage Size Property (\cref{it:preimg}).
    Let $X$ be a block-independent process with blocklength $b$ such that
    \begin{equation}\label{eq:condition-blocklength}
        \max(\alpha_b/b, \beta_b/b)\leq \eps^2/12.
    \end{equation}
    Then, there exists a constant $n_0$ (possibly depending on $\eps$ and $X$) such that for all $n\geq n_0$
    we have
    \begin{equation*}
        \Pr_{(x,z)\sim X_1^{n},Z(X_1^{n})}\left[\left|\frac{i_{X_1^{n},Z(X_1^{n})}(x,z)}{n} - I(X;Z(X))\right|\leq\eps\right]\geq 1-\eps.
    \end{equation*}
\end{theorem}

Before we prove \cref{thm:conv-info-density-gen} we establish some useful lemmas.
We will start by working with the $Z^\star$ channel.
\begin{lemma}\label{lem:conv-star-gen}
    Suppose that $X$ is a block-independent process with blocklength $b$.
    For $n=tb+r$ with $0\leq r<b$, let
    \begin{equation*}
        \eta^\star_{n}= (Z^\star(X_1^b),Z^\star(X_{b+1}^{2b}),\dots,Z^\star(X_{(t-1)b+1}^{tb}),Z^\star(X_{tb+1}^{tb+r})).
    \end{equation*}
    Then, for any $\eps>0$ there exists $n_0$ such that for all $n\geq n_0$ we have
    \begin{equation*}
        \Pr\left[\left|\frac{i_{X_1^{n},\eta^\star_{n}}(X_1^{n},\eta^\star_{n})}{n}-\frac{I(X_1^b;Z^\star(X_1^b))}{b}\right|>\eps\right]\leq \eps.
    \end{equation*}
\end{lemma}
\begin{proof}
    Write $t=\lfloor n/b\rfloor$, so that $n=tb+r$ with $0\leq r<b$.    
    First, by block-independence of $X$ (with blocklength $b$) we have
    \begin{equation*}
        p_{X_1^{n},\eta^\star_{n}}(x,z)=\left(\prod_{i=1}^t p_{X_1^b,Z^\star(X_1^b)}(x^{(i)},z^{(i)})\right)\cdot p_{X_1^r,Z^\star(X_1^r)}(x^{(t+1)},z^{(t+1)}),
    \end{equation*}
    where $x^{(i)},z^{(i)}$ denote the $i$-th blocks of $x$ and $\eta^\star_n$, respectively. If $r=0$, then the $(t+1)$-st block is empty and the corresponding probability term is $1$.
    This implies that
    \begin{equation}\label{eq:decomp-info-density-gen}
        i_{X_1^{n},\eta^\star_{n}}(x,z) = \sum_{i=1}^t i_{X_1^b,Z^\star(X_1^b)}(x^{(i)},z^{(i)}) + i_{X_1^r,Z^\star(X_1^r)}(x^{(t+1)},z^{(t+1)}),
    \end{equation}
    with the rightmost term in the sum being $0$ if $r=0$.
    Therefore, combining \cref{eq:decomp-info-density-gen} with the triangle inequality,
    \begin{align}
        &\Pr\left[\left|\frac{i_{X_1^{n},\eta^\star_{n}}(X_1^{n},\eta^\star_{n})}{n}-\frac{I(X_1^b;Z^\star(X_1^b))}{b}\right|>\eps\right] \nonumber\\
        &\leq \Pr\left[\left|\frac{1}{n}\sum_{i=1}^t i_{X_1^b,Z^\star(X_1^b)}(X_{(i-1)b+1}^{ib},Z^\star(X_{(i-1)b+1}^{ib}))-  \frac{I(X_1^b;Z^\star(X_1^b))}{b}\right|> \eps/2 \right]\nonumber\\
        &+ \Pr\left[\left|\frac{1}{n}i_{X_1^r,Z^\star(X_1^r)}(X_{tb+1}^{tb+r},Z^\star(X_{tb+1}^{tb+r}))\right|>\eps/2\right]. \label{eq:boundmain}
    \end{align}
    We bound both probabilities in the sum, starting with the first term.
    Since the random variables $i_{X_1^b,Z^\star(X_1^b)}(X_{(i-1)b+1}^{ib},Z^\star(X_{(i-1)b+1}^{ib}))$ are i.i.d.\ for all $i\in[t]$ and their expectation is $I(X_1^b;Z^\star(X_1^b))$, the law of large numbers guarantees that for any $\eps>0$ there is $t_0$ such that for all $t\geq t_0$ we have
    \begin{equation*}
        \Pr\left[\left|\frac{1}{tb}\sum_{i=1}^t i_{X_1^b,Z^\star(X_1^b)}(X_{(i-1)b+1}^{ib},Z^\star(X_{(i-1)b+1}^{ib}))-  \frac{I(X_1^b;Z^\star(X_1^b))}{b}\right|> \eps/3 \right]\leq \eps/2.
    \end{equation*}
    In turn, since $\frac{n}{tb}\to 1$ as $n\to\infty$, this implies that for any $\eps>0$ there is $n_0$ such that for all $n\geq n_0$ we have
    \begin{equation}\label{eq:bound-term1}
        \Pr\left[\left|\frac{1}{n}\sum_{i=1}^t i_{X_1^b,Z^\star(X_1^b)}(X_{(i-1)b+1}^{ib},Z^\star(X_{(i-1)b+1}^{ib}))-  \frac{I(X_1^b;Z^\star(X_1^b))}{b}\right|> \eps/2 \right]\leq \eps/2.
    \end{equation}

    To bound the second term in \cref{eq:boundmain}, note that for any $\eps>0$ there is $C_\eps>0$ such that
    \begin{equation*}
        \Pr\left[\left|i_{X_1^r,Z^\star(X_1^r)}(X_1^r,Z^\star(X_1^r))\right|>C_\eps\right]\leq \eps/2
    \end{equation*}
    simultaneously for all $0\leq r<b$ (since $b$ is fixed).
    Consequently, for any $n\geq 2C_\eps/\eps$ we have (recall that $X_{tb+1}^{tb+r}$ is distributed like $X_1^r$)
    \begin{equation}\label{eq:bound-term2}
        \Pr\left[\left|\frac{1}{n}i_{X_1^r,Z^\star(X_1^r)}(X_{tb+1}^{tb+r},Z^\star(X_{tb+1}^{tb+r}))\right|>\eps/2\right]\leq \eps/2.
    \end{equation}
    Combining \cref{eq:boundmain,eq:bound-term1,eq:bound-term2} yields the desired result for all $n\geq \max(n_0,2C_\eps/\eps)$.
\end{proof}

For a given input $X_1^{n}$, we will couple the $\eta^\star_n$ and $Z^\star(X_1^{n})$ processes with the help of the Amortized Preimage Size Property (\cref{it:preimg}).
By this property, there exists a random variable $Y_n$ with countable support such that $X_1^{n}\to \eta^\star_n\to Y_n$ and for which there exists a deterministic function $\phi$ such that $Z^\star(X_1^{n})=\phi(\eta^\star_n,Y_n)$ and $\log |\phi^{-1}(z)|\leq t \cdot \beta_b$ for all $z$, where we recall that $t=\lfloor n/b\rfloor$.

\begin{lemma}\label{lem:id-rem-star-gen}
    We have
    \begin{align*}
        &\E_{(x,z,y)\sim X_1^{n},\eta^\star_n,Y_n}\left[|i_{X_1^{n},\eta^\star_n}(x,z)-i_{X_1^{n},Z^\star(X_1^{n})}(x,\phi(z,y))\right]\\
        &=\E_{(x,z,y)\sim X_1^{n},\eta^\star_n,Y_n}\left[|i_{X_1^{n},\eta^\star_n,Y_n}(x,z,y)-i_{X_1^{n},Z^\star(X_1^{n})}(x,\phi(z,y))\right] \\
        &\leq 2t\cdot \beta_b.
    \end{align*}
\end{lemma}
The proof of this lemma relies on the following simple but useful fact about information densities, due to Dobrushin~\cite{Dob67}.
Since Dobrushin's paper is in Russian and a translation is hard to find, for completeness we present a proof in \cref{sec:proof-id-dob-gen}.
\begin{lemma}[{\cite[Equation (4.3)]{Dob67}}]\label{lem:id-dob-gen}
    Let $\phi:\cA\to\cB$ be a function where $\cA$ and $\cB$ are countable sets, and suppose that $|\phi^{-1}(z)|\leq M_\phi$ for all $z\in\cB$.
    Let $A$ be supported on $\cA$ and define $B=\phi(A)$.
    Then, for any discrete random variable $X$ arbitrarily correlated with $A$, we have\footnote{ Dobrushin~\cite[Equation (4.3)]{Dob67} claims a $\log M_\phi$ upper bound instead of the $2\log M_\phi$ upper bound stated in \cref{lem:id-dob-gen}. However, as far as we can tell, the argument presented in \cite{Dob67} only guarantees the latter weaker bound. This difference is immaterial to our application of \cref{lem:id-dob-gen}.}
    \begin{equation*}
        \E_{(x,a)\sim (X,A)}[|i_{X,A}(x,a)-i_{X,B}(x,\phi(a))|] \leq 2\log M_\phi.
    \end{equation*}
\end{lemma}

We are now ready to proceed with the proof of \cref{lem:id-rem-star-gen}.
\begin{proof}[Proof of \cref{lem:id-rem-star-gen}]
    The first equality follows from the fact that
    \begin{equation*}
        i_{X_1^{n},\eta^\star_n,Y_n}(x,z,y) = i_{X_1^{n},\eta^\star_n}(x,z)
    \end{equation*}
    for all $(x,z,y)$ in the support of $(X_1^{n},\eta^\star_n,Y_n)$, since $X_1^{n}\to \eta^\star_n\to \eta^\star_n, Y_n$.
    The inequality follows from \cref{lem:id-dob-gen} with $X=X_1^{n}$ and $A=(\eta^\star_n,Y_n)$
    and
    the hypothesis on $\phi$ and the support of $Y_n$ above, guaranteed by the Amortized Preimage Size Property (\cref{it:preimg}).
\end{proof}

We are now ready to prove \cref{thm:conv-info-density-gen}.
\begin{proof}[Proof of \cref{thm:conv-info-density-gen}]
    Since $X$ is block-independent with blocklength $b$, by \cref{lem:block-ind-rate-gen} we have
    \begin{equation}\label{eq:mi-rem-star-gen}
        \left|\frac{I(X_1^b;Z^\star(X_1^b))}{b} - I(X;Z^\star(X))\right| \leq \frac{\alpha_b}{b} \leq \eps/3,
    \end{equation}
    where the last inequality uses the hypothesis on $b$ from the theorem statement.
    Therefore, it is enough to show that
    \begin{equation}\label{eq:new-goal-gen}
        \Pr_{(x,z)\sim X_1^{n},Z^\star(X_1^{n})}\left[\left|\frac{i_{X_1^{n},Z^\star(X_1^{n})}(x,z)}{n} - \frac{I(X_1^b;Z^\star(X_1^b))}{b}\right|\leq 2\eps/3\right] \geq 1-\eps.
    \end{equation}

    First, by the triangle inequality, we have
    \begin{align}
        &\Pr_{(x,z)\sim X_1^{n},Z^\star(X_1^{n})}\left[\left|\frac{i_{X_1^{n},Z^\star(X_1^{n})}(x,z)}{n} - \frac{I(X_1^b;Z^\star(X_1^b))}{b}\right|> 2\eps/3\right]\nonumber\\
        &\leq \Pr_{(x,z^\star,y)\sim X_1^{n},\eta^\star_n,Y_n}\left[\left|\frac{i_{X_1^{n},\eta^\star_n}(x,z^\star)}{n}-\frac{i_{X_1^{n},Z^\star(X_1^{n})}(x,\phi(z^\star,y))}{n}\right|>\eps/3\right]\nonumber\\
        &+ \Pr_{(x,z^\star)\sim X_1^{n},\eta^\star_n}\left[\left|\frac{i_{X_1^{n},\eta^\star_n}(x,z^\star)}{n}- \frac{I(X_1^b;Z^\star(X_1^b))}{b}\right|>\eps/3\right]. \label{eq:prob-triangle-gen}
    \end{align}
    We analyze the two terms in the sum separately.
    For the first term, combining \cref{lem:id-rem-star-gen} with Markov's inequality yields
    \begin{align}
        \Pr_{(x,z^\star,y)\sim X_1^{n},\eta^\star_n,Y_n}\left[\left|\frac{i_{X_1^{n},\eta^\star_n}(x,z^\star)}{n}-\frac{i_{X_1^{n},Z^\star(X_1^{n})}(x,\phi(z^\star,y))}{n}\right|>\eps/3\right]
        &\leq  \frac{3\cdot 2t\cdot \beta_b}{\eps n} \nonumber\\
        &\leq 
        \frac{6\beta_b}{\eps b} \nonumber\\
        &\leq \eps/2,\label{eq:markov-gen}
    \end{align}
    where the second inequality uses the fact that $n\geq tb$ and the last inequality holds by the hypothesis on $b$ from the theorem statement.
    For the second term, by \cref{lem:conv-star-gen} with $\eps/3$ in place of $\eps$, for all $n\geq n_0$ with $n_0$ a sufficiently large constant depending on $\eps$, $b$, and $X$, we have
    \begin{equation}\label{eq:conv-star-gen}
        \Pr_{(x,z^\star)\sim X_1^{n},\eta^\star_n}\left[\left|\frac{i_{X_1^{n},\eta^\star_n}(x,z^\star)}{n}- \frac{I(X_1^b;Z^\star(X_1^b))}{b}\right|>\eps/3\right] \leq \eps/2.
    \end{equation}
    Combining \cref{eq:prob-triangle-gen} with \cref{eq:markov-gen,eq:conv-star-gen} yields \cref{eq:new-goal-gen},
    as desired.
\end{proof}

\subsubsection{Capacity-achieving codes for admissible channels}

\cref{thm:conv-info-density-gen} implies, via standard methods, that the coding capacity and information capacity of admissible channels coincide.
For completeness, we discuss this in detail.
Later in \cref{sec:dense-codes} we combine \cref{thm:icap-scap,thm:conv-info-density-gen} to show the existence of ``dense'' capacity-achieving codes suitable for bootstrapping efficient constructions.

We begin by relying on the following well-known theorem that formalizes the guarantees of MAP decoding.
\begin{lemma}[\protect{\cite[Theorem 18.5]{PW24}, adapted}]\label{lem:shannon-coding}
    Fix an input random variable $X$ supported on $\Sigmain^n$ and a channel $Z$ with input alphabet $\Sigmain$.
    Then, for any $\tau>0$ there exists a code $\cC$ with blocklength $n$, size $M$, and average decoding error probability $\eps$
    satisfying
    \begin{equation*}
        \eps \leq \Pr_{(x,z)\sim X,Z(X)}[i_{X,Z(X)}(x,z)\leq \log M + \tau] + 2^{-\tau}.
    \end{equation*}
\end{lemma}
The code guaranteed by \cref{lem:shannon-coding} is obtained by sampling $M$ codewords i.i.d.\ according to $X$.
We briefly discuss how \cref{lem:shannon-coding} can be combined with \cref{thm:conv-info-density-gen} to obtain codes with arbitrarily small decoding error probability and the desired rate for an arbitrary admissible channel.
\begin{corollary}\label{coro:code}
    Let $Z$ be an admissible channel.
    Then, for any $\delta>0$ and any block-independent input process $X$ there exists $n(\delta,X)$ such that for any $n\geq n(\delta,X)$ there is an $(n,R,\delta)$-code for $Z$ with $R\geq I(X;Z(X))-\delta$.
\end{corollary}
\begin{proof}
    Fix $\delta>0$, set $\eps=\delta/3$, and consider an arbitrary block-independent input process $X$ with blocklength $b$.
    We may assume that $I(X;Z(X))-\delta>0$ as otherwise the corollary statement is trivially true.
    Note that we can assume that $b$ is large enough so that \cref{eq:condition-blocklength} in \cref{thm:conv-info-density-gen}
    is satisfied, since a block-independent process with blocklength $b$ is also block-independent with blocklength $b'=cb$ for any integer $c\geq 1$.
    Let $n_0=n_0(\eps,X)$ be the constant guaranteed by \cref{thm:conv-info-density-gen} for this choice of $\eps$ and $X$.
    Set $M=\lceil 2^{n(I(X;Z(X))-3\eps)}\rceil$ and $\tau=\eps n-1$, which is positive whenever $n>1/\eps$.
    By \cref{lem:shannon-coding}, for all $n\geq \max(n_0,1+1/\eps)$ there exists an $(n,R=\frac{1}{n}\log M\geq  I(X;Z(X))-3\eps,\lambda)$-code for $Z$ with
    \begin{align*}
        \lambda&\leq \Pr_{(x,z)\sim X_1^n,Z(X_1^n)}[i_{X_1^n,Z(X_1^n)}(x,z)\leq \log M + \tau] + 2^{-\tau} \\
        &\leq  \Pr_{(x,z)\sim X_1^n,Z(X_1^n)}\left[\frac{i_{X_1^n,Z(X_1^n)}(x,z)}{n}\leq I(X;Z(X))-2\eps\right]+2^{-\eps n+1}\\
        &\leq 2\eps + 2^{-\eps n+1},
    \end{align*}
    where the last inequality follows from \cref{thm:conv-info-density-gen}.
    Now, we may set $n$ large enough as a function of $\eps$ so that $2^{-\eps n+1}\leq \eps$, in which case $\lambda\leq 3\eps$.
    Recalling that $\eps=\delta/3$ concludes the proof.
\end{proof}

\paragraph{Obtaining vanishing decoding error probability.}
We now argue how \cref{coro:code} can be extended to get a family of codes defined for all block lengths with vanishing decoding error probability and rate matching the information rate of any given block-independent input process.
This shows that for any admissible channel  $Z$ the coding capacity equals the information capacity.

\begin{theorem}\label{thm:icap-ccap}
    Suppose that $Z$ is an admissible channel.
    Then, $\ICap(Z)=\CCap(Z)$.
\end{theorem}
\begin{proof}

For each $k\in \N$ define $\eps^{(k)} = 1/k$.
By the discussion after \cref{eq:guarantee-X-gen}, for each $k$ there is a block independent process $X$ such that $I(X;Z(X))\geq \ICap(Z)-\eps^{(k)}$.
Then, \cref{coro:code} with this $X$ and $\delta=\eps^{(k)}$ guarantees the existence of a family $\{\cC_n^{(k)}\}_{n\in\N}$ of $(n,R^{(k)}_n,\eps_n^{(k)})$-codes, where $R^{(k)}_n\geq \ICap(Z)-\eps^{(k)}$ and $\eps_n^{(k)}\leq \eps^{(k)}$ for all $n\geq n(k)$, for some $n(k)$ (note that the choice of input process $X$ is fixed for each $k\in\N$, so $n(k)$ really only depends on $k$).

Now fix a gap to information capacity $\gamma>0$ and some integer $k^\star$ such that $1/k^\star\leq \gamma$.
We construct a family of codes $\{\cC_n\}_{n\in\N}$ as follows.
For each $n\geq n(k^\star)$ take $\cC_n$ to be $\cC_n^{(k)}$ for $k\geq k^\star$ if $n(k)\leq n<n(k+1)$.
For each $n<n(k^\star)$ take $\cC_n=\Sigmain^n$.
Then, each $\cC_n$ is an $(n,R_n,\eps_n)$-code with $R_n\geq \ICap(Z)-\frac{1}{k^\star}\geq \ICap(Z)-\gamma$ for all $n$ and $\eps_n\to 0$ as $n\to\infty$.
The result follows since $\gamma>0$ was arbitrary.
\end{proof}

\subsubsection{Dense capacity-achieving codes for admissible channels}\label{sec:dense-codes}

In the previous section we showed that the coding capacity and information capacity of an admissible channel are the same.
However, this alone is not sufficient if we wish to obtain \emph{efficiently encodable and decodable} capacity-achieving codes for an admissible channel.
In this section, following the approach of~\cite{PLW22}, we combine the fact that capacity is achieved by stationary ergodic processes (\cref{thm:icap-scap}) with \cref{thm:conv-info-density-gen} to show the existence of capacity-achieving codes with density properties for admissible channels as stated in \cref{thm:dense-code-gen}, useful for constructing efficient capacity-achieving codes.\footnote{Pernice, Li, and Wootters~\cite{PLW22} focused on channels with i.i.d.\ deletions and replications. The analog of \cref{thm:icap-scap} for these channels was already shown in~\cite{Dob67}.}
For simplicity we will focus on admissible channels with binary input alphabet, although our discussion generalizes further.
We restate \cref{thm:dense-code-gen} here for convenience.

\densecodegen*

Our proof of \cref{thm:dense-code-gen} will rely on the following lemma proved in~\cite{PLW22}.
\begin{lemma}[\protect{\cite[Proposition 3.4]{PLW22}, adapted}]\label{lem:dense-stat-erg}
    Let $X$ be a stationary ergodic process supported on $\bits$ with $\Pr[X_1=1]\in (0,1)$.
    Then, for any $\zeta>0$ there exists $\gamma\in(0,1/2)$ and an integer $n_0>0$ such that the following holds for all $n\geq n_0$.
    With probability at least $0.99$ over the sampling of $x\sim X_1^n$, we have $\gamma\zeta n\leq w(x_i^{i+\zeta n})\leq (1-\gamma)\zeta n$ for all $i\in[(1-\zeta)n]$, where $w(\cdot)$ denotes the Hamming weight.
\end{lemma}

\begin{proof}[Proof of \cref{thm:dense-code-gen}]
    Let $Z$ be an arbitrary admissible channel with binary input alphabet.
Let $X$ be a stationary ergodic process such that $I(X;Z(X))>0$.
Then, it must be the case that $\Pr[X_1=1]\in(0,1)$, and so \cref{lem:dense-stat-erg} applies to $X$ with some constants $\zeta>0$, $\gamma\in(0,1/2)$, and $n_0$.

    It will be slightly easier to work with a block-independent process.
    Fix $\eps>0$.
    From the proof of \cref{thm:icap-scap} in \cref{sec:icap-scap}, we know that there is a stationary ergodic process $\overline{X}$ such that $I(\overline{X};Z(\overline{X}))\geq \ICap(Z)-\eps$, and moreover $\overline{X}$ is created by choosing an appropriate block-independent process $\hat{X}$ with blocklength $b(\eps)$, then choosing a uniformly random starting point in the first block of $\hat{X}$, and starting $\hat{X}$ from that point.
    Since $\overline{X}$ is obtained from $\hat{X}$ by trimming at most $b$ bits from the beginning of $\hat{X}$, we conclude that $\hat{X}$ also satisfies the properties laid out in \cref{lem:dense-stat-erg} with possibly a slightly smaller $\gamma$ and slightly larger $n_0$.
    Recalling that the code $\cC$ guaranteed by \cref{coro:code} applied to $Z$ and $\hat{X}$ is obtained by sampling codewords i.i.d.\ according to $\hat{X}_1^n$ yields \cref{thm:dense-code-gen}.
    This is because with high probability more than a $0.9$-fraction of codewords $c\in\cC$ will satisfy $\gamma\zeta n\leq w(c_i^{i+\zeta n})\leq (1-\gamma)\zeta n$ for all $i\in[(1-\zeta)n]$, and throwing away all codewords of $\cC$ that do not satisfy this property will not affect the asymptotic rate.
\end{proof}

\section{Some special cases of our capacity theorems}\label{sec:special-cases}

\subsection{The Mao-Diggavi-Kannan ISI model}\label{sec:capture-MDK}

Consider the \emph{$\ell$-ISI-synchronization channel} $Z$ from~\cite{MDK18}, for an arbitrary fixed integer $\ell\geq 0$.
This channel replaces the $i$-th input bit $x_i$ by a string $y_i\in\bits^*$ with probability
\begin{equation*}
    p(y_i|x_i,x_{i-1},\dots,x_{i-\ell}).
\end{equation*}
For simplicity, we focus on the case where $p(\cdot|x_i,x_{i-1},\dots,x_{i-\ell})$ is supported on $\bits^{\leq a}=\bigcup_{j=0}^a \bits^j$ for some integer $a\geq 1$ and any choice of $x_i,x_{i-1},\dots,x_{i-\ell}$, although our argument below generalizes further.
We show that our capacity theorems apply to this channel, and so they generalize the corresponding results of~\cite{MDK18}.

Consider the special channel $Z^\star$ that behaves like $Z$, except that it does not corrupt the first $\ell$ input bits and separately outputs the last $\ell$ input bits (in particular, $Z^\star$ is noiseless on inputs $x$ of length at most $\ell$).
First, we show that $Z^\star$ is well-behaved. 
Then, we show that $Z$ is admissible with respect to $Z^\star$.
\begin{itemize}
    \item \textbf{Bounded Entropy Property (\cref{it:ent-bound}):} Since the output associated to the $i$-th input bit has length at most a fixed constant $a$, we have $H(Z^\star(X_1^n))\leq n\cdot(a+1)$.

    \item \textbf{Concatenation Property (\cref{it:concat}):} Note that $Z^\star(X_{1}^{m})$ does not corrupt the first $\ell$ bits of $X_{1}^{m}$, and furthermore it outputs $X_{m-\ell+1}^{m}$. Otherwise, it behaves exactly like $Z$. Analogously, $Z^\star(X_{m+1}^{n})$ does not corrupt the first $\ell$ bits of $X_{m+1}^{n}$ and it outputs $X_{n-\ell+1}^{n}$.
    Therefore, from $Z^\star(X_{1}^{m}),Z^\star(X_{m+1}^{n})$ we have the necessary information to apply the correct errors to the first $\ell$ bits of $X_{m+1}^n$, and we also know the last $\ell$ input bits $X_{n-\ell+1}^{n}$.
    This means that $X_1^n \to Z^\star(X_{1}^{m}),Z^\star(X_{m+1}^{n})\to Z^\star(X_1^{n})$.

    \item \textbf{Partition Property (\cref{it:partition}):} 
        \begin{enumerate}
            \item Prefix/Suffix-Partition Property (\cref{it:single-part}): Fix integers $\tau$ and $n\geq \tau$, and an input process $X$.
            Let $N$ denote the number of output bits corresponding to $X_1^\tau$ in $Z^\star(X_1^n)$.
            Also, let $\Wpre$ include for each $j\in\{\tau+1,\dots,\tau+\ell\}$ the string $v_j$ that $X_j$ was replaced by in $Z^\star(X_1^n)$.
            Then, $(Z^\star(X_1^n),N,\Wpre,X_{\tau-\ell+1}^\tau,X_{\tau+1}^{\tau+\ell})$ completely determine $Z^\star(X_1^\tau),Z^\star(X_{\tau+1}^n)$, and so  $X_1^n\to Z^\star(X_1^n),N,\Wpre,X_{\tau-\ell+1}^\tau,X_{\tau+1}^{\tau+\ell}\to Z^\star(X_1^\tau),Z^\star(X_{\tau+1}^n)$.
            Furthermore, 
            \begin{equation*}
                H(N,\Wpre,X_{\tau+1}^{\tau+\ell},X_{n-\ell+1}^n)\leq \log(\tau\cdot 2^{a+1})+\ell\cdot (a+1) + 2\ell.
            \end{equation*}
            Therefore, the prefix-partitioning half of the Prefix/Suffix-Partition Property (\cref{it:single-part}) holds with $\gamma_m=\log\tau + (a+1)+\ell(a+3)$.
            An analogous argument establishes the suffix-partitioning property with the same $\gamma_m$.

            \item Amortized Block-Partition Property (\cref{it:multi-part}): Fix a blocklength $b$, number of blocks $t$, and an input process $X$.
            For each $i\in[t]$, let $N_i$ denote the number of output bits corresponding to $X_{(i-1)b+1}^{ib}$.
            Also, let $W_i$ include, for each $j\in\{(i-1)b+1,(i-1)b+\ell\}$ the string $v_j$ that $X_j$ is replaced by in $Z^\star(X_1^{tb})$.
            Let
            \begin{equation*}
                W=(N_i,W_i,X_{ib-\ell+1}^{ib})_{i\in[t]}.
            \end{equation*}
            Then, $Z^\star(X_1^{tb}), W$ completely determine $Z^\star(X_1^{b}),\dots,Z^\star(X_{(t-1)b+1}^{tb})$, and so
            \begin{equation*}
                X_1^{tb}\to Z^\star(X_1^{tb}), W\to Z^\star(X_1^{b}),\dots,Z^\star(X_{(t-1)b+1}^{tb}).
            \end{equation*}
            Moreover,
        \begin{equation*}
            H(W)\leq \sum_{i=1}^t (\log(b\cdot 2^{a+1})+\ell\cdot (a+1)+\ell) =t\cdot(\log(b\cdot 2^{a+1})+\ell\cdot (a+1)+\ell).
        \end{equation*}
        Therefore, the Amortized Block-Partition Property (\cref{it:multi-part}) holds with $\alpha_m = \log(m\cdot 2^{a+1})+\ell\cdot (a+1)+\ell = o(m)$.
            
        \end{enumerate}

    \item \textbf{Amortized Preimage Size Property (\cref{it:preimg}):} Fix a blocklength $b$ and $n=tb+r$ with $t\geq 1$ and $0\leq r<b$.
    Define $\ell_i=\min(\ell,b)$ for $i\in[t]$ and $\ell_{t+1}=\min(\ell,r)$.
    Consider the random variable $Y=(Y_i)_{i\in[t+1]}$, where each $Y_i$ includes, for each $j\in\{(i-1)b+1,\dots,(i-1)b+\ell_i\}$, the string $v_j$ that $X_j$ should be replaced by in $Z^\star(X_1^n)$.
    Then $Y$ has countable support and $X_1^{n}\to Z^\star(X_1^b),\dots,Z^\star(X_{(t-1)b+1}^{tb}),Z^\star(X_{tb+1}^{tb+r})\to Y$, since the distribution of $Y_i$ is completely determined by $X_{(i-1)b-\ell+1}^{(i-1)b}$ and $X_{(i-1)b+1}^{(i-1)b+\ell_i}$, where the former is revealed by the output blocks $Z^\star(X_1^b),\dots,Z^\star(X_{(i-2)b+1}^{(i-1)b})$ 
    and the latter is revealed by the $i$-th output block.
    Moreover, we have $Z^\star(X_1^{n})=\phi(Z^\star(X_1^b),\dots,Z^\star(X_{(t-1)b+1}^{tb}),Z^\star(X_{tb+1}^{tb+r}),Y)$ for the deterministic function $\phi$ that 
    discards the final up to $\ell$ bits revealed by each output block except $X_{n-\ell+1}^n$,
    applies the corruptions dictated by $Y$ to each of $Z^\star(X_1^b),\dots,Z^\star(X_{(t-1)b+1}^{tb}),Z^\star(X_{tb+1}^{tb+r})$, and concatenates the blocks.

    It remains to upper bound $\log|\phi^{-1}(z)|$ for all $z$ appropriately.
    First, there are at most $\binom{(t+1)(b\cdot a+1)}{t+1}$ ways of splitting $z$ into $t+1$ blocks of length at most $b\cdot a$ each (note that each input bit is mapped to a string of length at most $a$ at the output, and each block of $X$ has length at most $b$).
    Second, for each block there are at most $2^{(a+1)\cdot\ell}$ choices for the first $\ell$ input bits and the corresponding strings in $\bits^{\leq a}$ that they were replaced by at the output, and there are at most $2^\ell$ possibilities for the last $\ell$ input bits.
    Putting these observations together implies that, for all $z$,
    \begin{equation*}
        |\phi^{-1}(z)| \leq \binom{(t+1)(b\cdot a+1)}{t+1} \cdot 2^{(a+1)\ell(t+1)}\cdot 2^{\ell(t+1)} \leq 2^{(t+1)(b\cdot a+1)h\left(\frac{1}{b\cdot a +1}\right)}\cdot 2^{(a+2)\ell(t+1)},
    \end{equation*}
    where we have used the standard inequality $\binom{n}{k}\leq 2^{n h(k/n)}$, with $h(p)=-p\log p - (1-p)\log(1-p)$ the binary entropy function.
    Therefore,
    \begin{align*}
        \log |\phi^{-1}(z)| &\leq (t+1)((b\cdot a +1)h\left(\frac{1}{b\cdot a +1}\right) +(a+2)\ell)
    \end{align*}
    for all $z$, and so the Amortized Preimage Size Property (\cref{it:preimg}) holds with $\beta_m = 2((m\cdot a+1)h\left(\frac{1}{m\cdot a+1}\right)+(a+2)\ell)=o(m)$, since $(t+1)/t\leq 2$ for all $t\geq 1$.

\end{itemize}

The argument above shows that $Z^\star$ is well-behaved.
To see that $Z$ is admissible with respect to $Z^\star$, fix an arbitrary input process $X$. 
Note that 
    \begin{equation*}
        X_1^n\to Z^\star(X_1^n)\to Z(X_1^n).
    \end{equation*}
    Let $W$ denote, for each $j\in[\ell]$, the string $v_j$ that $X_j$ was replaced by in $Z(X_1^n)$.
    Then,
    \begin{equation*}
        X_1^n\to Z(X_1^n), W,X_{n-\ell+1}^n \to Z^\star(X_1^n).
    \end{equation*}
    Note that there are at most $2^{a+1}$ choices for each $v_j$.
    Therefore, $H(W,X_{n-\ell+1}^n)\leq \ell+\ell\cdot (a+1)$, and so
    \begin{equation*}
        I(X_1^n;Z^\star(X_1^n))-\ell(a+2) \leq I(X_1^n;Z(X_1^n)) \leq I(X_1^n;Z^\star(X_1^n)).
    \end{equation*}
    As a result,
    \begin{equation*}
        \lim_{n\to \infty} \frac{|I(X_1^n;Z(X_1^n))-I(X_1^n;Z^\star(X_1^n))|}{n} \leq \lim_{n\to \infty} \frac{\ell(a+2)}{n} = 0.
    \end{equation*}

\subsection{Multi-trace channels with input-correlated synchronization errors}\label{sec:cap-multi-trace}
We argue how our capacity theorems above apply to a wide class of \emph{multi-trace} input-correlated synchronization channels.
Recalling the discussion in \cref{sec:channel-def}, fix an integer $T \geq 1$ (the number of traces) and consider the multi-trace channel $Z$ given by
\begin{equation*}
    Z(x) = (Z_1(x),\dots,Z_T(x)),
\end{equation*}
where the $Z_i$'s are possibly distinct channels with the same input spaces, and the $Z_i(x)$'s are conditionally independent given $x$.

\begin{theorem}\label{thm:multi-trace-cap}
    Suppose there exist well-behaved channels $Z^\star_1,\dots,Z^\star_T$ such that
    \begin{equation}\label{eq:multi-trace-extra}
        I(X;Z_1(X),\dots,Z_T(X)) = I(X;Z^\star_1(X),\dots,Z^\star_T(X))
    \end{equation}
    for all input processes $X$.
    Then, the $T$-trace channel $Z^\star$ given by $Z^\star(x) = (Z^\star_1(x),\dots,Z^\star_T(x))$ is well-behaved and the $T$-trace channel $Z$ given by $Z(x) = (Z_1(x),\dots,Z_T(x))$ is admissible with respect to $Z^\star$.
    In particular, we have $\ICap(Z)=\SCap(Z)=\CCap(Z)$.
\end{theorem}

The assumption in \cref{eq:multi-trace-extra} appears stronger than simply requiring that each $Z_i$ be admissible with respect to $Z^\star_i$ for $i\in\{1,\dots,T\}$.
Nevertheless, it still seems reasonable.
Concretely, it is natural (as in all of our applications, for example) that $X_1^n \to Z_i(X_1^n), W_i\to Z^\star_i(X_1^n)$ for some additional side information $W_i$ satisfying $H(W_i)=o(n)$.
In this case, we get that $X_1^n\to Z(X_1^n), W\to Z^\star(X_1^n)$ for the side information $W=(W_1,\dots,W_T)$, which satisfies $H(W)\leq \sum_{i=1}^T H(W_i) =T\cdot o(n)=o(n)$ since the number of traces $T$ is constant.
This means that
\begin{equation*}
    \frac{|I(X_1^n;Z(X_1^n))-I(X_1^n;Z^\star(X_1^n))|}{n} \leq \frac{H(W)}{n} \to 0
\end{equation*}
as $n\to\infty$.

\begin{proof}[Proof of \cref{thm:multi-trace-cap}]
    We show that $Z^\star$ is well-behaved. The fact that $Z$ is admissible with respect to $Z^\star$ is already guaranteed by \cref{eq:multi-trace-extra}.

    \begin{itemize}
        \item \textbf{Bounded Entropy Property (\cref{it:ent-bound}):} Fix an arbitrary input process $X$ and integer $n\geq 1$.
        For each $i$ we know that $H(Z_i^\star(X_1^n))\leq c_i n$ for some constant $c_i>0$.
        Let $c^\star=\max_{i\in[T]} c_i$.
        Then, $H(Z^\star(X_1^n))=H((Z^\star_i(X_1^n)_{i\in[T]})\leq \sum_{i=1}^T H(Z^\star_i(X_1^n))\leq T\cdot c^\star n$, and so the Bounded Entropy Property (\cref{it:ent-bound}) holds with constant $c=T\cdot c^\star$.

        \item \textbf{Concatenation Property (\cref{it:concat}):} For an arbitrary input process $X$, indices $1\leq m\leq n$ and each $i\in[T]$, we have that $X_1^n \to Z^\star_i(X_1^m),Z^\star_i(X_{m+1}^{n})\to Z_i^\star(X_1^n)$.
        In particular, this means that
        \begin{equation*}
            X_1^n \to (Z^\star_i(X_1^m))_{i\in[T]},(Z^\star_i(X_{m+1}^{n})_{i\in[T]} = Z^\star(X_1^m), Z^\star(X_{m+1}^n)\to (Z_i^\star(X_1^n))_{i\in[T]} = Z^\star(X_1^n).
        \end{equation*}

        \item \textbf{Partition Property (\cref{it:partition}):}
            \begin{itemize}
                \item Prefix/Suffix-Partition Property (\cref{it:single-part}): Fix any integers $\tau\geq 1$ and $n\geq \tau$.
                For each $i\in[T]$ let $\gamma^{(i)}_m=o(m)$ and $W_{\mathsf{pre},i}$ and $W_{\mathsf{suf},i}$ be the sequence and random variables guaranteed by the Prefix/Suffix-Partition Property (\cref{it:single-part}) applied to $Z_i$.
                Consider $\gamma_m=\sum_{i=1}^T \gamma^{(i)}_m$ and $W_{\mathsf{pre}}=(W_{\mathsf{pre},i})_{i\in[T]}$ and $W_{\mathsf{suf}}=(W_{\mathsf{suf},i})_{i\in[T]}$.
                Then, $\gamma_m=o(m)$ and $\sum_{m\in\N}\gamma_m/m^2$ converges, and $H(W_{\mathsf{pre}}),H(W_{\mathsf{suf}})\leq \gamma_n$.
                Furthermore,
                \begin{multline*}
                    X_1^n\to (Z^\star(X_1^n),W_{\mathsf{pre}})=(Z^\star_i(X_1^n),W_{\mathsf{pre},i})_{i\in[T]}\\
                    \to (Z^\star_i(X_1^\tau),Z^\star_i(X_{\tau+1}^n))_{i\in [T]}=(Z^\star(X_1^\tau),Z^\star(X_{\tau+1}^n)).
                \end{multline*}
                The reasoning for $W_{\mathsf{suf}}$ is analogous.

                \item Amortized Block-Partition Property (\cref{it:multi-part}): For each $i\in[T]$ let $\alpha^{(i)}_m=o(m)$ and $W_i$ be the sequence and random variable, respectively, guaranteed by the Amortized Block-Partition Property (\cref{it:multi-part}) applied to $Z_i$.
                Consider $W=(W_1,\dots,W_T)$ and $\alpha_m = \sum_{i=1}^T \alpha^{(i)}_m$.
                Then, $\alpha_m = o(m)$ and $H(W)\leq \sum_{i=1}^T H(W_i)\leq t\alpha_b$, and
                \begin{multline*}
                    X_1^{tb}\to (Z^\star(X_1^{tb}),W)=((Z^\star_i(X_1^{tb}),W_i))_{i\in[T]} \\
                    \to (Z^\star_i(X_1^b),\dots,Z^\star_i(X_{(t-1)b+1}^{tb}))_{i\in[T]} = Z^\star(X_1^b),\dots,Z^\star(X_{(t-1)b+1}^{tb}).
                \end{multline*}
            \end{itemize}

        \item \textbf{Amortized Preimage Size Property (\cref{it:preimg}):} 
        Fix $n=tb+r$ with $t\geq 1$ and $0\leq r<b$.
        For each $i\in[T]$, let $\beta^{(i)}_m=o(m)$, $Y_i$, and $\phi_i$ be the sequence, random variable, and deterministic function guaranteed by the Amortized Preimage Size Property (\cref{it:preimg}) applied to $Z_i$ for this $n$.
        Set $\beta_m=\sum_{i=1}^T \beta^{(i)}_m=o(m)$, $Y=(Y_i)_{i\in[T]}$ (which has countable support since each $Y_i$ has countable support by hypothesis), and 
        \begin{multline*}
            \phi(Z^\star(X_1^b),\dots,Z^\star(X_{(t-1)b+1}^{tb}),Z^\star(X_{tb+1}^{tb+r}),Y) \\
            = (\phi_i(Z^\star_i(X_1^b),\dots,Z^\star_i(X_{(t-1)b+1}^{tb}),Z^\star_i(X_{tb+1}^{tb+r}),Y_i))_{i\in[T]}.
        \end{multline*}
        Then,
        \begin{multline*}
            X_1^{n}\to Z^\star(X_1^b),\dots,Z^\star(X_{(t-1)b+1}^{tb}),Z^\star(X_{tb+1}^{tb+r})\\
            = (Z^\star_i(X_1^b),\dots,Z^\star_i(X_{(t-1)b+1}^{tb}),Z^\star_i(X_{tb+1}^{tb+r}))_{i\in[T]} \to (Y_i)_{i\in[T]} = Y
        \end{multline*}
        and, for any $z$,
        \begin{equation*}
            \log|\phi^{-1}(z)| \leq \log\left(\prod_{i=1}^T |\phi_i^{-1}(z_i)|\right) = \sum_{i=1}^T \log |\phi_i^{-1}(z_i)| \leq t\sum_{i=1}^T \beta^{(i)}_b = t\beta_b. \qedhere
        \end{equation*}
    \end{itemize}
\end{proof}

\subsection{Capacity theorems for trimming synchronization channels}\label{sec:trimming}

In this section, we argue that our capacity theorems are robust to additional ``trimming'' of channel outputs.
This property is relevant for the marker-based construction of efficient capacity-achieving codes based on our capacity theorems.

For concreteness, let $Z$ be an arbitrary channel with input spaces $\cX_n=\bits^n$ and output spaces $\cY_n=\bits^*$.
We consider the ``$00$-trimming'' version of $Z$, denoted by $Z_{00}$, which on input $x$ first sends $x$ through $Z$ to obtain output $Z(x)$, and then trims the runs of $0$s at the beginning and end of $Z(x)$.
More precisely, it deletes from $Z(x)$ the maximal substrings of $0$s that are a prefix or a suffix of $Z(x)$ to obtain the final output $Z_{00}(x)$ (when $Z(x)=0^\ell$ for some $\ell$ these two maximal substrings are actually the same, in which case we simply delete all bits and $Z_{00}(x)$ is the empty string).
For example, if $Z(x) = \mathbf{00}1011\mathbf{0}$, then $Z_{00}(x) = 1011$; if $Z(x) = \mathbf{0}11001$, then $Z_{00}(x) = 11001$; if $Z(x) = 1101$, then $Z_{00}(x) = 1101$; and if $Z(x) = \mathbf{0000}$, then $Z_{00}(x)$ is the empty string.

\begin{lemma}\label{lem:trimming-cap}
    Let $Z$ be admissible with respect to a well-behaved channel $Z^\star$.
    Then, $Z_{00}$ is also admissible with respect to $Z^\star$.
    In particular, the capacities of $Z$ and $Z_{00}$ are all equal.
\end{lemma}
\begin{proof}
    It suffices to show that $I(X;Z_{00}(X))=I(X;Z(X))$ for any input process $X$.
    Fix an arbitrary integer $n>0$.
    First, since $X_1^n\to Z(X_1^n)\to Z_{00}(X_1^n)$, we have that $I(X_1^n;Z_{00}(X_1^n))\leq I(X_1^n;Z(X_1^n))$.
    On the other hand, if $L_0$ and $L_1$ denote the number of $0$s trimmed by $Z_{00}$ from the beginning and end of $Z(X_1^n)$, we have that $X_1^n\to Z_{00}(X_1^n),L_0,L_1\to Z(X_1^n)$, and $H(L_0,L_1)\leq 2\log(n+1)$.
    Therefore,
    \begin{equation*}
        \lim_{n\to\infty} \frac{|I(X_1^n;Z_{00}(X_1^n))-I(X_1^n;Z(X_1^n))|}{n}\leq \lim_{n\to\infty} \frac{2\log(n+1)}{n} = 0,
    \end{equation*}
    which implies the desired result.
\end{proof}

The simple proof of \cref{lem:trimming-cap} can be easily extended to other trimming channels that trim different prefixes and suffixes.
In particular, consider the ``$01$-trimming'' version of $Z$, denoted by $Z_{01}$, which on input $x$ first sends it through $Z$ to obtain output $Z(x)$, and trims the run of $0$s at the beginning of $Z(x)$ and the run of $1$s at the end of $Z(x)$, i.e., we delete from $Z(x)$ the maximal substring of $0$s that is a prefix of $Z(x)$ and the maximal substring of $1$s that is a suffix of $Z(x)$.
A simple modification to the proof of \cref{lem:trimming-cap} yields the following.
\begin{lemma}\label{lem:01trimming-cap}
    Let $Z$ be admissible with respect to a well-behaved channel $Z^\star$.
    Then, $Z_{01}$ is also admissible with respect to $Z^\star$.
    In particular, the capacities of $Z$ and $Z_{01}$ are all equal.
\end{lemma}

\subsection{Channels with runlength-dependent deletions}\label{sec:assumptions-runlength}

In this section, we apply our framework to binary channels $Z$ with runlength-dependent deletions, in the sense that $Z$ deletes each bit in a run of length $\ell$ independently with some probability $d(\ell)$ (and so we may see $d$ as a function $d:\N\to[0,1]$).
Consider the special channel $Z^\star$ that on input $x$ behaves exactly like $Z$ except that it does not apply deletions to the first and last runs of $x$, and additionally reveals the lengths of these runs.
We show that $Z^\star$ is well-behaved. Then, we show that $Z$ is admissible with respect to $Z^\star$.

\begin{itemize}

    \item \textbf{Bounded Entropy Property (\cref{it:ent-bound}):} Fix an arbitrary input process $X$ and integer $n\geq 1$.
    Since $Z^\star$ only applies deletions, we have $H(Z^\star(X_1^n))\leq n+2\log n\leq 3n$.

    \item \textbf{Concatenation Property (\cref{it:concat}):} Fix arbitrary integers $1\leq m\leq n$ and an input process $X$.
    Note that $Z^\star(X_1^m)$ does not apply deletions to the first and last runs of $X_1^m$ and also reveals the lengths of these runs, and likewise for $Z^\star(X_{m+1}^n)$.
    To all other runs of $X_1^m$ and $X_{n+1}^m$ these channels apply the same deletion rate as $Z^\star(X_1^n)$, because these runs are not broken up by the partitioning of $X_1^n$ into $X_1^m$ and $X_{m+1}^n$.
    Furthermore, $Z^\star(X_1^m),Z^\star(X_{m+1}^n)$ reveal the lengths of the first and last runs of $X_1^n$ and do not apply deletions to these runs.
    Therefore, 
    it is enough to argue that knowing $Z^\star(X_1^m),Z^\star(X_{m+1}^n)$ allows us to apply the correct deletion rates to the last run of $X_1^m$ and first run of $X_{m+1}^n$, which may actually be part of the same run of $X_1^n$.
    
    In the special case where the last run of $X_1^m$ is also its first run (i.e., when $X_1^m=b^m$ for some $b\in\bits$) then we apply no deletions to it, nor to the first run of $X_{m+1}^n$ in case it matches the bit value of $X_1^m$.
    In this case, the length of the first run of $X_1^n$, which is part of the output of $Z^\star(X_1^n)$ can be obtained from the lengths of the first runs of $X_1^m$ and $X_{m+1}^n$, which are revealed by $Z^\star(X_1^m),Z^\star(X_{m+1}^n)$.
    An analogous argument holds for the special case where the last run of $X_{m+1}^n$ is also its first run (i.e., when $X_{m+1}^n=b^{n-m}$ for some $b\in\bits$).
    
    In all other cases, since $Z^\star(X_1^m)$ reveals the length of the last run of $X_1^m$ and $Z^\star(X_{m+1}^n)$ reveals the length of the first run of $X_{m+1}^n$, we know the length of the corresponding run(s) of $X_1^n$ and so we also know the deletion rate that must be applied.
    Moreover, since no deletions were applied to the last run of $X_1^m$ and the first run of $X_{m+1}^n$ by $Z^\star(X_1^m)$ and $Z^\star(X_{m+1}^n)$, respectively, we can perfectly emulate the behavior of $Z^\star(X_1^n)$ on the corresponding runs, and so $X_1^n \to Z^\star(X_1^m),Z^\star(X_{m+1}^n)\to Z^\star(X_1^n)$.
    
    \item \textbf{Partition Property (\cref{it:partition}):} 

        \begin{enumerate}
            \item Prefix/Suffix-Partition Property (\cref{it:single-part}): Fix integers $\tau$ and $n\geq \tau$ and an input process $X$.
            Let $N$ denote the number of bits deleted by $Z^\star(X_1^n)$ from $X_1^\tau$. Furthermore, let $L_1,B_1$ (resp.\ $L_2,B_2$) denote the length of the last run of $X_1^\tau$ (resp.\ first run of $X_{\tau+1}^n$) and the bit value of this run, respectively, and let $N_1$ (resp.\ $N_2$) denote the number of bits deleted from the last run of $X_1^\tau$ (resp.\ first run of $X_{\tau+1}^n$).
            Set $W_{\mathsf{pre}}=(N,L_1,B_1,N_1,L_2,B_2,N_2)$.
            Then, from $Z^\star(X_1^n),W_{\mathsf{pre}}$ we can exactly locate the output bits corresponding to the last run of $X_1^\tau$ and to the first run of $X_{\tau+1}^n$, and restore them to their original lengths.
            Therefore, $Z^\star(X_1^n),W_{\mathsf{pre}}$ completely determine $Z^\star(X_1^\tau),Z^\star(X_{\tau+1}^n)$, and so $X_1^n\to Z^\star(X_1^n),W_{\mathsf{pre}}\to Z^\star(X_1^\tau),Z^\star(X_{\tau+1}^n)$.
            Furthermore, $H(W_{\mathsf{pre}})\leq \log(\tau+1)+4\log(n+1)+2$, and so the Prefix/Suffix-Partition Property (\cref{it:single-part}) holds with $\gamma_m=\log(\tau+1)+4\log(m+1)+2$, which satisfies the required properties.

            An analogous argument establishes the suffix-partitioning property with the same $\gamma_m$.

            \item Amortized Block-Partition Property (\cref{it:multi-part}): Fix a blocklength $b$, number of blocks $t$, and an input process $X$.
            For each $i\in[t]$, let $N_i$ denote the number of bits deleted from $X_{(i-1)b+1}^{ib}$ by $Z^\star(X_1^n)$, let $L_{1,i},N_{1,i},B_{1,i}$  (resp.\ $L_{2,i},N_{2,i},B_{2,i}$) denote the length of the first (resp.\ last) run of $X_{(i-1)b+1}^{ib}$, the number of bits deleted from this run by $Z^\star(X_1^n)$, and the bit value of this run, respectively, and let $V_i$ denote whether the first and last runs of $X_{(i-1)b+1}^{ib}$ are distinct runs.
            Note that 
            \begin{equation*}
                H(L_{1,i},N_{1,i},B_{1,i},L_{2,i},N_{2,i},B_{2,i},V_i)\leq 4\log(b+1)+3
            \end{equation*}
            for every $i\in[t]$.
        Let     $W=(L_{1,i},N_{1,i},B_{1,i},L_{2,i},N_{2,i},B_{2,i},V_i)_{i\in[t]}$.
        Then, using a similar argument to previous items, we see that $Z^\star(X_1^{tb}), W$ completely determines the sequence $Z^\star(X_1^{b}),\dots,Z^\star(X_{(t-1)b+1}^{tb})$, and so $X_1^{tb}\to Z^\star(X_1^{tb}), W\to Z^\star(X_1^{b}),\dots,Z^\star(X_{(t-1)b+1}^{tb})$.
        Moreover,
        \begin{equation*}
            H(W)\leq \sum_{i=1}^t (4\log(b+1)+3) =t\cdot(4\log(b+1)+3).
        \end{equation*}
        Therefore, the Amortized Block-Partition Property (\cref{it:multi-part}) holds with $\alpha_m = 4\log (m+1) + 3 = o(m)$.
            
        \end{enumerate}

    \item \textbf{Amortized Preimage Size Property (\cref{it:preimg}):} Fix a blocklength $b$ and $n=tb+r$ with $t\geq 1$ and $0\leq r<b$.
    Define $X^{(i)}=X_{(i-1)b+1}^{ib}$ for $i\in[t]$ and $X^{(t+1)}=X_{tb+1}^{tb+r}$.
    Consider the random variable $Y=(L_{1,i},L_{2,i})_{i\in[t+1]}$ where $L_{1,i},L_{2,i}$ denote the number of bits to be deleted from the first and last runs of $Z^\star(X^{(i)})$ so that we can transform $Z^\star(X^{(1)}),\dots,Z^\star(X^{(t+1)})$ into $Z^\star(X_1^{n})$.
    Then $Y$ has countable support and $X_1^{n}\to Z^\star(X^{(1)}),\dots,Z^\star(X^{(t+1)})\to Y$, since the distribution of $Y$ is completely determined by the lengths of the first and last runs of the $X^{(i)}$'s, which are revealed by the $Z^\star(X^{(i)})$'s.
    Moreover, we obtain $Z^\star(X_1^{n})=\phi(Z^\star(X^{(1)}),\dots,Z^\star(X^{(t+1)}),Y)$ for the deterministic function $\phi$ that applies the deletions dictated by $Y$ to the first and last runs of $Z^\star(X^{(1)}),\dots,Z^\star(X^{(t+1)})$, concatenates these blocks, and uses the lengths of the first and last runs of each block $X^{(i)}$ revealed by $Z^\star(X^{(i)})$ to compute the lengths of the first and last runs of $X_1^n$.
    Together, these form $Z^\star(X_1^{n})$.

    It remains to upper bound $\log|\phi^{-1}(z)|$ appropriately for an arbitrary $z$.
    First, note that there are at most $\binom{(t+1)(b+1)}{t+1}$ ways of splitting $z$ into $t+1$ blocks of length at most $b$ each.
    Second, for each of the $t+1$ blocks of length at most $b$, there are at most $(b+1)^2$ possibilities for the number of bits that were deleted from the first and last input runs of each block, $(b+1)^2$ possibilities for the lengths of these runs, $2^2$ possibilities for the bit values of these runs, and $2$ possibilities for whether these two runs are distinct runs or not.
    Putting these observations together implies that
    \begin{equation*}
        |\phi^{-1}(z)| \leq \binom{(t+1)(b+1)}{t+1} \cdot (b+1)^{4(t+1)}\cdot 2^{3(t+1)} \leq 2^{(t+1)(b+1)h\left(\frac{1}{b+1}\right)}\cdot (b+1)^{2(t+1)}\cdot 2^{3(t+1)},
    \end{equation*}
    where we recall that $h$ is the binary entropy function.
    Therefore,
    \begin{align*}
        \log |\phi^{-1}(z)| &\leq (t+1)(b+1)h\left(\frac{1}{b+1}\right) + 4(t+1)\log(b+1)+3(t+1) \\
        &= (t+1)\cdot\left((b+1)h\left(\frac{1}{b+1}\right) + 4\log(b+1)+3\right)\\
        &\leq 2t\cdot\left((b+1)h\left(\frac{1}{b+1}\right) + 4\log(b+1)+3\right)
    \end{align*}
    for any $z$, since $(t+1)/t\geq 2$ for all $t\geq 1$. Therefore,
    the Amortized Preimage Size Property (\cref{it:preimg}) holds with $\beta_m = 2((m+1)h\left(\frac{1}{m+1}\right) + 4\log(m+1)+3)=o(m)$.

\end{itemize}

Now we show that $Z$ is admissible with respect to $Z^\star$.
Fix an arbitrary input process $X$. Note that $X_1^n\to Z^\star(X_1^n)\to Z(X_1^n)$.
    Let $L_1,B_1$ (resp.\ $L_2, B_2$) denote the number of bits deleted by $Z$ from the first (resp.\ last) input run and the bit value of this run.
    Let also $V$ denote whether the first and last input runs are distinct.
    Then, $(Z(X_1^n), L_1,B_1,L_2,B_2,V)$ completely determines $Z^\star(X_1^n)$.
    Therefore, since $H(L_1,B_1,L_2,B_2,V)\leq 1+2(\log(n+1)+1)$, we have
    \begin{equation*}
        I(X_1^n;Z^\star(X_1^n))-(1+2(\log(n+1)+1)) \leq I(X_1^n;Z(X_1^n)) \leq I(X_1^n;Z^\star(X_1^n)),
    \end{equation*}
    and so
    \begin{equation*}
        \lim_{n\to \infty} \frac{|I(X_1^n;Z(X_1^n))-I(X_1^n;Z^\star(X_1^n))|}{n} \leq \lim_{n\to \infty} \frac{1+2(\log(n+1)+1)}{n} = 0.
    \end{equation*}
    This implies that $I(X;Z(X))=I(X;Z^\star(X))$, as desired.

We would like to apply \cref{thm:dense-code-gen} to the \emph{$01$-trimming multi-trace} version of any runlength-dependent deletion channel.
More precisely, this corresponds to the setting where the $T\geq 1$ channels $Z_1,\dots,Z_T$ are $01$-trimming versions of runlength-dependent deletion channels (see \cref{sec:trimming}).
To be able to apply \cref{thm:dense-code-gen} in this case it suffices to establish the hypothesis of \cref{thm:multi-trace-cap} (\cref{eq:multi-trace-extra}).
For simplicity we assume that $Z_1,\dots,Z_T$ all behave like the $01$-trimming version of the same channel $Z$ with runlength-dependent deletions. 
The argument generalizes easily to the case where traces come from channels with runlength-dependent deletions with different parameters.
Take $Z^\star_1,\dots,Z^\star_T$ to be channels behaving exactly like the well-behaved channel $Z^\star$ described above in this section.
Then, it suffices to show that for any input process $X$ we have
\begin{equation}\label{eq:hyp-multitrace-restate}
    I(X;Z_1(X),\dots,Z_T(X))=I(X;Z^\star_1(X),\dots,Z^\star_T(X)).
\end{equation}

Fix an arbitrary process $X=(X_i)_{i\in\N}$ and an arbitrary $n\in\N$.
First, note that $X_1^n\to Z^\star_i(X_1^n)\to Z_i(X_1^n)$ for each $i$, and so
\begin{equation*}
    X_1^n \to Z^\star_1(X_1^n),\dots, Z^\star_T(X_1^n) \to Z_1(X_1^n),\dots,Z_T(X_1^n).
\end{equation*}
Consequently,
\begin{equation}\label{eq:non-neg-diff}
    I(X_1^n;Z^\star_1(X_1^n),\dots,Z^\star_T(X_1^n))-I(X_1^n;Z_1(X_1^n),\dots,Z_T(X_1^n)) \geq 0.
\end{equation}
To complement this, consider for each $i$ the side information
\begin{equation*}
    W_i = (L_{1,i}, L_{2,i},D_{1,i},D_{2,i},B_{1,i},B_{2,i}, T_{0,i}, T_{1,i}).
\end{equation*}
Here,
$(L_{1,i},D_{1,i},B_{1,i})$ (resp.\ $(L_{2,i},D_{2,i},B_{2,1})$) denote, respectively, the length of the first (resp.\ last) run of $X_1^n$, the number of bits deleted from that run by $Z_i(X_1^n)$ \emph{before the trimming process}, and the bit value of that run.
Moreover, $T_{0,i}$ (resp.\ $T_{1,i}$) denotes the number of $0$s (resp.\ $1$s) trimmed from the prefix (resp.\ suffix) after the runlength-dependent deletions have been applied.
Note that $Z_i(X_1^n)$ and $W_i$ together completely determine $Z^\star_i(X_1^n)$ for each $i$, and so, setting $W=(W_1,\dots,W_T)$, we have that
\begin{equation*}
    X_1^n\to Z_1(X_1^n),\dots,Z_T(X_1^n),W\to Z^\star_1(X_1^n),\dots,Z^\star_T(X_1^n).
\end{equation*}
Note that $H(W_i) \leq 6\log(n+1)+2$, and so $H(W)\leq T(6\log(n+1)+2)$.
Therefore, recalling also \cref{eq:non-neg-diff},
\begin{equation*}
    |I(X_1^n;Z^\star_1(X_1^n),\dots,Z^\star_T(X_1^n))-I(X_1^n;Z_1(X_1^n),\dots,Z_T(X_1^n))| \leq H(W)\leq T(6\log(n+1)+2),
\end{equation*}
and so
\begin{equation*}
    \lim_{n\to\infty} \frac{|I(X_1^n;Z^\star_1(X_1^n),\dots,Z^\star_T(X_1^n))-I(X_1^n;Z_1(X_1^n),\dots,Z_T(X_1^n))|}{n} \leq \lim_{n\to\infty} \frac{T(6\log(n+1)+2)}{n}=0,
\end{equation*}
since $T$ is constant.
This yields \cref{eq:hyp-multitrace-restate}, as desired.
Formally, we get the following corollary directly from applying \cref{thm:multi-trace-cap,thm:dense-code-gen}.

\begin{corollary}
\label{cor:dense-code-RL-dep}
    Let $Z$ be a $01$-trimming $T$-trace runlength-dependent deletion channel.
    Then, for any $\eps,\zeta>0$ there exist $\gamma=\gamma(\eps,\zeta)\in(0,1/2)$ and integers $b=b(\eps,\zeta)$ and $t(\eps,\zeta)$ that depend only on $\eps$ and $\zeta$ such that for any $t\geq t(\eps,\zeta)$ there exists a code $\cC$ with blocklength $n=tb$, rate $R\geq \ICap(Z)-\eps$, decoding error probability at most $\eps$ such that for all codewords $c\in\cC$ we have $\gamma\zeta n\leq w(c_i^{i+\zeta n})\leq (1-\gamma)\zeta n$ for all $i\in[(1-\zeta)n]$.    
\end{corollary}

\begin{remark}
    \em
    \Cref{cor:dense-code-RL-dep} also holds for the $10$, $11$, $00$-trimming $T$-trace runlength-dependent deletion channel, whose definitions are analogous to the one above.
\end{remark}

\section{Warmup: efficient capacity-achieving codes for channels with runlength-dependent deletions} \label{sec:efficient-single-trace}

In \Cref{cor:dense-code-RL-dep}, for the case of $T = 1$, we showed that there is a code that achieves capacity on the $00$-trimming runlength-dependent channel.\footnote{Recall from \cref{sec:trimming} that this means that after the deletions are performed, the channel also trims the first and last runs of zeros.} Furthermore, we showed that this code has a large density of $1$s in every ``not too short'' interval. 

In this section, we will closely follow the arguments of Pernice, Li, and Wootters~\cite{PLW22} who showed how one can transform any code for the binary deletion channel (and, in fact, more i.i.d. synchronization channels) into an explicit and efficient code with a negligible loss in the rate. 
With some needed modifications, we will show how to transform (non-explicit and non-efficient) capacity-achieving codes for a smaller class of runlength-dependent channels into explicit and efficient capacity achieving codes.

We start by defining the class of runlength-dependent channels for which we aim to construct efficient codes.

\begin{definition} \label{def:run-length-bounded-channel}
    Let $M\in \mathbb{N}$ and $\mu\in (0,1)$. A runlength-dependent deletion channel with deletion probability function $d:\mathbb{N} \rightarrow [0,1]$ is called $\bdcRLB$ if $d$ is non-decreasing and for all $\ell \geq M$, we have $d(\ell) = d(M) < 1 - \mu$.\footnote{The monotonicity assumption is for simplicity and also since it makes sense to assume that longer runs are more likely to suffer from higher deletion rate.}
\end{definition}

The theorem we will prove in this section is as follows.
\begin{theorem} \label{thm:efficient-bounded-rl}
    Let $\varepsilon > 0$. There exists an explicit family of binary codes $\{C_i\}_{i=1}^{\infty}$ for the channel $\bdcRLB$ where the block length of $C_i$ goes to infinity as $i\rightarrow \infty$ and\footnote{The channel parameters are fixed and do not depend on the block lengths of the codes.}
    \begin{enumerate}
        \item $C_i$ is encodable in linear time and decodable in quasi-linear time.
        \item The decoding failure probability is $\exp(-\Omega(n_i))$ where $n_i$ is the block length of $C_i$.
        \item The rate of the $C_i$ is 
        $R > \textup{Cap}(\bdcRLB) - \varepsilon$.
    \end{enumerate}
\end{theorem}

Before formally presenting our construction, we  briefly recall, in an informal way, the construction of \cite{PLW22}. 
As in \cite{GL19,con2022improved}, the construction of \cite{PLW22} is based on code concatenation where the outer code, $\Cout$, is taken to be the code of \cite{HS21}. 
The inner code of \cite{PLW22}, denoted as $\Cin$, is taken to be a capacity-achieving code over the trimming  binary deletion channel that is ``dense''. 
More precisely, our \Cref{cor:dense-code-RL-dep} for the single trace setting can be seen as an extension of their \cite[Proposition 3.4]{PLW22} (stated here as \Cref{lem:dense-stat-erg}). Thus, a concatenated codeword in $\Cout \circ \Cin$ is of the form $\Cin(\sigma_1)\circ \cdots \circ \Cin(\sigma_n)$ where $(\sigma_1, \ldots, \sigma_n)$ is an outer codeword.

Then, in the concatenated codeword, every two adjacent inner codewords are separated using a large run of zeros (termed as buffers). Adding these buffers was done also in \cite{GL19,con2022improved} and the goal is to reduce the loss of synchronization between the receiver and the sender. The final codeword that is transmitted through the channel is 
\[
\Cin(\sigma_1)\circ 0^B \circ  \cdots \circ 0^B \circ \Cin(\sigma_n)
\]
where $B$ is the length of the buffer.

Now, the decoder which receives a corrupted version of the transmitted codeword consists of three steps. First, identifying the buffers following the simple rule: every run of zeros of length greater than some parameter is identified as a buffer. 
With standard concentration bounds (and appropriate parameters), one can show that almost all buffers are correctly identified as buffers. 
Moreover, with high probability there are very few ``spurious'' buffers inside the inner codewords (that is, the channel created a long run of zeros inside an inner codeword and the decoder mistakenly identified it as a buffer). Here we use the fact that in every ``not too short'' interval in an inner codeword there are many $1$s. 
After this step, every string between two buffers is decoded using decoder of the inner code. Note here that the buffer identification step might trim the codeword from the left and right. However, the inner code achieves capacity on the trimming version of the channel and thus most of the corrupted inner codeword are decoded successfully.

Finally, we run the decoder of the outer codeword which can correct from a small amount of insertions and deletions (insdel errors). 
Observe that each deleted buffer and spurious buffer contributes at most $3$ insdel errors in the outer codeword symbols, and that a failure in the inner code's decoding algorithm results in $2$ insdel errors. 
Thus, as long as with high probability the number of all of these errors can be made as small as we want, then the outer code of rate $1 - \varepsilon - \delta$ that can correct from $\delta$ insdel errors can handle those errors.
Concluding, we have that with high probability decoding is successful and the rate is close as we want to the rate of the inner code which achieves capacity.

\subsection{Construction}
Let $\varepsilon$ be the desired gap to capacity. 
\paragraph{Outer and inner codes.}
The coding scheme uses code concatenation.
For the outer code, as in previous works that construct binary codes for synchronization channels, we shall use a code that can correct from \emph{adversarial} (worst-case) indel (insertions and deletions) errors. We start by defining the relevant metric.
\begin{definition}
    The \emph{edit distance} between two strings $s,s'$, denoted by $\ed(s,s')$, is the minimal number of insertions and deletions needed to convert $s$ into $s'$.
\end{definition}

We use the code by Haeupler, Rubinstein, and Shahrasbi~\cite{haeupler2019near} that can correct adversarial insdel errors.
\begin{theorem}[\cite{haeupler2019near}] \label{thm:hs-code}
    For every $\epsout,\delout\in(0,1)$ there exists a family of codes $\{\cC_n\}_{n\in\N}$, where $\cC_n$ has blocklength $n$, of rate $\Rout = 1-\delout-\epsout$ over an alphabet $\Sigma$ of size $|\Sigma|=O_{\epsout}(1)$ that can correct $\delout n$ adversarial insdel errors.
    The codes $\cC_n$ support linear time encoding and quasi-linear time decoding.
\end{theorem}
In other words, this code has the property that the edit distance between any two distinct codewords is at least $2\delout n$.
We will denote by $\Cout$ the outer code given by \Cref{thm:hs-code} where we set $\delout = \varepsilon/5$. 
By $\Cin$, we denote the inner code that achieves capacity on the $0$-trimming $\bdcRLB$ channel, given by \Cref{cor:dense-code-RL-dep} by setting $\varepsilon$ (in the corollary) to $\varepsilon/50$ and $\zeta$ (in the corollary) to be $\nu/3$ where $\nu = \varepsilon\cdot \mu/5$.
We also make sure that each codeword in $\Cin$ starts and ends with a $1$ simply by adding a $1$ bit at the beginning and at the end. Clearly, this does not affect the asymptotics.
Formally, our encoding process is as follows:

    \paragraph{Encoding.} Given as input a message $x\in \Sigma^{\Rout n}$, we encode it with the code given in \Cref{thm:hs-code} to obtain an outer codeword $c^{(\textup{out})} = (\sigma_1, \ldots, \sigma_n) \in C_{\textup{out}} \subset \Sigma^n$. Then, every symbol in $c^{(\textup{out})}$, $\sigma_i \in \Sigma = \{0,1\}^{m \cdot \Rin}$, is encoded using the inner code to a codeword that we denote  $c^{(\textup{in})}_{\sigma_i}$. We thus get a codeword in the concatenated code 
    \[ 
	\left( c^{(\textup{in})}_{\sigma_1},\ldots,    c^{(\textup{in})}_{\sigma_n} \right) \in        C_{\textup{out}} \circ C_{\textup{in}} \;.
    \]
    Now that we have a codeword in the concatenated code, we add an additional layer of encoding that is crucial for preserving synchronization. Every two adjacent inner codewords are separated by a buffer of zeros of length $\ceil{\nu \cdot m/(1-d(M))}$ where recall that $\nu$ is a small constant such that $\nu = \mu \cdot \eps/5$ (recall that, by \Cref{def:run-length-bounded-channel}, $\mu$ is such that $d(M) \leq 1 - \mu$ and that $d(M)$ is the largest deletion probability the channel can impose).
    
    \begin{remark}
        \em
        We note that in order for the concatenation to make sense, we must have $\Rin m = \ceil{\log_2 |\Sigma|} = \ceil{\log_2 (O_{\epsout}(1))}$.
        We shall choose a small enough $\epsout$ so that the length of the buffer will be $\geq M$, which would imply that the channel will apply deletions to it with probability $d(M)$.
    \end{remark}

\paragraph{Rate.}
    The codeword length is $mn + (n-1)\ceil{\nu m/(1 - d(M)}$ and, therefore, the rate of the code is
    \begin{align*}
        \mathcal{R} &= \frac{\log_2|\Sigma|^{\Rout n}}{mn + (n-1)\ceil{\nu m/(1 - d(M)}}\\
        &\geq \frac{\Rin \Rout}{1 + \nu/(1 - d(M)) + 1/m}\\
        &\geq \frac{(1 - \varepsilon/5 -\epsout)\Rin}{1 + \varepsilon/5 + 1/m} \\
        &\geq \frac{(1 - \varepsilon/4)\Rin}{1 + \varepsilon/4}\\
        & \geq (1 - \varepsilon/2)\Rin 
    \end{align*}
    where the second inequality follows since $1 - d(M) > \mu$. The third inequality follows by our assumption that $\delout = \varepsilon/5$ and taking $\epsout$ to be small enough (so that $\varepsilon/5 + \epsout \leq \varepsilon/4$ and $1/m$ is at most $\varepsilon/5$). 
    Now, since $\Rin = \textup{Cap}(\bdcRLB) - \varepsilon/50$, we get that $\mathcal{R}\geq \textup{Cap}(\bdcRLB) - \varepsilon$.
\paragraph{Decoding.}    The decoding consists of the following step:
    \begin{enumerate}
        \item \label{decode:iden-buffers} Decoding buffers: Identify all runs of the symbol $0$ in the received string that are of length at least $\frac{\nu m}{2}$ and declare that these runs are buffers. 

        Denote by $\tilde{c}_1,\ldots, \tilde{c}_s$ the strings in between the buffers (and discard the buffers). Note that each $\tilde{c}_i$ starts and ends with a $1$. 
        
        \item Inner decoding: \label{decode:inner-code-decode} Decode each $\tilde{c}_i$ (brute force) into an outer code symbol $\tilde{\sigma}_i$.

        \item Outer decoding: \label{decode:outer-code-decode} Run the outer code decoding algorithm on $\sigouttil = (\tilde{\sigma}_1, \ldots, \tilde{\sigma}_s)$ to obtain $\tilde{x}$ and return it. 
    \end{enumerate}

\subsection{Analysis}

We start by upper bounding the probability of identifying a sent buffer.
\begin{claim}\label{clm:plw-like-deleted-buffer}
    The probability that a specific buffer is not identified is at most $\exp(-\Omega_{\eps,\mu}(m))$.
\end{claim}
\begin{proof}
    This happens when the channel deletes too many bits from a buffer so that less than $\nu m/2$ bits of the transmitted buffer survived the channel, and we did not identify this buffer when we identify buffers in Step~\ref{decode:iden-buffers}. 
        
    Let $Z$ be the random variable that corresponds to the number of bits from that buffer that survived the transmission through the $\bdcRLB$. Clearly, $Z \sim \Bin(\ceil{\nu m / (1 - d(M))}, 1 - d(M))$ and then, By Chernoff's bound
        \begin{align*}
            \Pr \left[ Z < \frac{\nu m}{2}  \right] \leq \Pr \left[ Z < \left(1 - \frac{1}{2} \right)\nu m \right] < \exp\left( - \frac{1}{8} \nu m \right) \;.
        \end{align*}
\end{proof}
Now, we bound the probability that a \emph{spurious buffer} is created by the channel.
\begin{claim} \label{clm:plw-like-spurious-buffer}
    Let $\Cin(\sigma_i)$ be an inner codeword that is transmitted through $\bdcRLB$ and let $y$ denote its output. 
    The probability that $y$ contains a run of zeros of length at least $\nu m/2$ is at most $\exp(-\Omega_{}(m))$.
\end{claim}
\begin{proof}
    Clearly, in order for the channel to create a run of zeros of length $\nu m/2$, it must be that the channel deleted all $1$s from an interval of length at least $\nu m/2$. 
        Fix an interval of length $\nu m/3$ in an inner codeword. 
        By \Cref{cor:dense-code-RL-dep} and recalling that $\zeta = \nu/3$, this interval contains at least $\gamma \nu m/3$ ones for some $\gamma = \gamma(\varepsilon/2, \zeta)$.
        Let $a_i, i\in [M-1]$ be the number of ones in this interval that belong to runs of length $i$ and denote by $a_M$ the number of bits that belong to runs of length $\geq M$ in this interval. 
        Deleting all these 1s happens with probability 
        \[
        d(1)^{a_1} d(2)^{a_2} \cdots d(M)^{a_M} \leq \left(\frac{3 \cdot \left(\sum_{i=1}^M a_i d(i)\right)}{\gamma \nu m}\right)^{\gamma \nu m/3} \leq (1 - \mu)^{\gamma \nu m/3} = \exp(-\Omega_{\varepsilon,\mu,\gamma}(m))\;,
        \]
        where the first inequality follows from the weighted AM-GM inequality and the second inequality is due to the assumption that $d(1) \leq d(2) \leq \cdots \leq d(M)\leq 1- \mu$ and that $\sum_{i=1}^M a_i = \gamma \nu m/3$.
        Union bounding over all such intervals, we get also that the probability of the existence of a spurious buffer inside an inner codeword is at most $m\cdot \exp(-\Omega_{\varepsilon,\mu,\gamma}(m)) =  \exp(-\Omega_{\varepsilon,\mu,\gamma}(m))$.    
\end{proof}

\begin{remark}
    \em
    Note that the above argument bounds the probability that a spurious buffer is identified inside an inner codeword. However, it can be that the decoder identifies several spurious buffer.
        The maximal number of spurious buffers the decoder can identify inside an inner codeword is at most $2/\nu = O(1)$. 
\end{remark}

\begin{claim} \label{clm:plw-like-wrong-inner-dec}
    Let $0^{B} \circ \Cin(\sigma_i) \circ 0^B$ be an inner codeword surrounded by two buffers and assume that it is transmitted through the channel $\bdcRLB$. 
    Assume that the buffers were identified by the algorithm and that no spurious buffers were created. 
    Denote by $\tilde{c}_j$ the corresponding string obtained after removing the buffers. 
    Then, the probability that $\tilde{c}_j$ is decoded correctly to $\sigma_i$ is at least $1 - \varepsilon/2$.
\end{claim}

\begin{proof}
    Observe that $\tilde{c}_j$ is the output of the $00$-trimming $\bdcRLB$ on the string $\Cin(\sigma_i)$.
    Thus, by \Cref{cor:dense-code-RL-dep}, the decoding failure probability is $\varepsilon/2$.
\end{proof}

With these claims, we can now prove \Cref{thm:efficient-bounded-rl}, exactly as in \cite{PLW22} (the proof is given here for completeness). 

\begin{proof}[Proof of \Cref{thm:efficient-bounded-rl}]
        Our goal is to show that with probability $1-\exp(-\Omega(n))$, it holds that $\ed(\sigma^{(\textup{out})},\tilde{\sigma}^{(\textup{out})}) \leq \delta_{\textup{out}} n$. This would imply that Step~\ref{decode:outer-code-decode} in our decoding algorithm succeeded and the correct message is returned. 
        There are three types of error that can increase the edit distance between $\sigout$ and $\sigouttil$: A deleted buffer, a spurious buffer, and wrong inner decoding. 

    We analyze the contribution of each of the error types (deleted buffer, spurious buffer and wrong inner decoding) on $\ed(\sigma^{(\textup{out})},\tilde{\sigma}^{(\textup{out})})$. 
    A deleted buffer causes two inner codewords to merge, and thus be decoded incorrectly by the inner code's decoding algorithm. This introduces two deletions and one insertion in the outer code level. 
    Similarly, one can verify that a single spurious buffer inside an inner codeword introduces one deletion and two insertions. Furthermore, $b$ spurious buffers inside an inner codeword introduce at most $b +1$ insertions and one deletion. Finally, a wrong inner decoding causes a substitution which is equivalent to one deletion followed by one insertion.

    As mentioned above, the outer decoding algorithm fails if $\ed(\sigma^{(\textup{out})},\tilde{\sigma}^{(\textup{out})}) > \delout n$. For this to happen, at least one of the following must occur:
    \begin{enumerate}
        \item $\delout n/9$ deleted buffers.
        \item $\delout n/9$ spurious buffers.
        \item $\delout n/9$ wrong inner decodings.
    \end{enumerate}
    Observe that for large enough $m$ (equivalently, for small enough $\epsout$), we get that
    the probability for a deleted buffer $\exp(-\Omega_{\varepsilon,\mu}(m))< \delout/10$ and thus, by Chernoff bound, the probability that there are at least $\delout n/9$ deleted buffers is at most $\exp(-\Omega(n))$. 
    Furthermore, since the expected number of spurious buffers between two adjacent real buffers is $\exp(-\Omega_{\varepsilon,\mu, \gamma}(m)) < \delout/ 10$ for large enough $m$ and the maximal number of such spurious buffers is $2/\nu$, we can apply Hoeffding's bound (\Cref{lem:hoeff}), and get that the total number of all spurious buffers is $< \delout n /9$ with probability $1 - \exp(-\Omega (n))$. 
    Finally, since a decoding of an inner codeword fails with probability $\varepsilon/50 = \delout/10$ (assuming the buffers surrounding it were detected and there were no spurious buffers inside), we get that the probability that more than $\delout n/9$ of these inner decodings failed is at most $\exp(-\Omega(n))$.
    The claim about the decoding failure probability follows by applying a simple union bound.

    We now justify the time complexity of our encoding and decoding algorithms. The complexity of the encoder is as follows. The encoder of the outer code runs in linear time. Then, the encoding of each outer code symbol using the inner code is performed in constant time, and thus the encoding of all the $n$ symbols is done in $O(n)$. Finally, adding the buffers also takes linear time.
    The decoding complexity is dominated by the outer code's decoding time which is quasi-linear. Indeed, identifying the buffers and decoding all corrupted inner codewords takes $O(n)$.
\end{proof}

    \begin{remark} 
        \em
        We emphasize the order in which we choose the parameters of the scheme and, in particular, we make sure that there is no circular dependency.
        First, observe that $\mu$ and $M$ are given by the channel $\bdcRLB$.
        
        We first choose the constant $\varepsilon$ which is the gap to capacity we want to achieve. This sets $\delout,\nu$, $\zeta$ and $\gamma$. 
        Then, we choose small enough $\epsout$ to ensures that all the failure probabilities computed in the proofs of \Cref{clm:plw-like-deleted-buffer} and \Cref{clm:plw-like-spurious-buffer} are indeed smaller than $\delout /10$ (recall that deleted buffer and spurious buffer happen with probability $\exp(-\Omega(m)) = \exp(-\Omega(\log_2 (|\Sigma|/\Rin))) = (O_{\epsout}(1))^{-\Omega(1)}$).
    \end{remark}

\begin{remark} \label{rem:strngth-of-technique}
    \em
	We remark that the proof technique applied in this section (which is based on \cite{PLW22}) can be applied to any runlength-dependent channel for which we can design buffers to separate inner codewords and prove that with high probability we can identify ``most'' of these buffers correctly and that with high probability the number of spurious buffers is ``small''.
	More specifically, all one needs to do is prove analogous claims to \Cref{clm:plw-like-deleted-buffer} and \Cref{clm:plw-like-spurious-buffer}.
	
	In this section, we followed a relatively easy example, the $\bdcRLB$ channel. Indeed, the buffers we used are simply long runs of zeros, and for the proofs of \Cref{clm:plw-like-deleted-buffer} and \Cref{clm:plw-like-spurious-buffer}, we used the nature of channel (the deletions inside a run are i.i.d.) and the fact (given in \Cref{cor:dense-code-RL-dep}) that our inner codewords contain many $1$s in every not-too-short interval (and thus, spurious buffers are not likely to be created).	
\end{remark}

\section{Efficient capacity-achieving codes for multi-trace channels with runlength-dependent deletions}
\label{sec:multi-trace}

In the previous section, we constructed efficient capacity-achieving codes for runlength-dependent deletion channel. In this section, our objective is to construct efficient capacity-achieving codes for the multi-trace version of the channel. Thus, one can see the result of the previous section as special case of the result we will obtain in this section. 
More specifically, our goal is to transform the structured capacity-achieving codes from \Cref{cor:dense-code-RL-dep} for a multi-trace runlength-dependent deletion channel into explicit and efficient capacity-achieving codes for the same channel.

This time, our techniques will follow the main ideas of Brakensiak, Li, and Spang \cite{BLS20} who constructed explicit and efficient codes of rate $1-\varepsilon$ that can be reconstructed from $\exp(O_d(\log^{1/3}(\varepsilon^{-1})))$ traces of the binary deletion channel with parameter $d$ (BDC$_d$, for short).

We start by defining the class of multi-trace, runlength-dependent channels for which we aim to construct efficient codes.
\begin{definition} \label{def:multi-trace-run-length-bounded-channel}
    Let $\bdcMTRLB$ be a runlength-dependent channel that produces $T$ traces on each input. As in \Cref{def:run-length-bounded-channel}, the deletion function is defined by a non-decreasing function $d(\ell):\mathbb{N} \rightarrow [0,1]$, number $\mu \in (0,1)$, and an integer $M$ such that $d(\ell) = d(M) < 1- \mu$ for all $\ell \geq M$.
\end{definition}

The theorem we will prove in this section is as follows.
\begin{theorem} \label{thm:efficient-bounded-multi-trace-rl}
    Let $\bdcMTRLB$ be a runlength-dependent channel that complies with \Cref{def:run-length-bounded-channel} where $\mu$ and $M$ are constants.
    For every $\varepsilon > 0$, there exists a family of binary codes $\{C_i\}_{i=1}^{\infty}$ for the $\bdcMTRLB$ where the block length of $C_i$ goes to infinity as $i\rightarrow \infty$ and
    \begin{enumerate}
        \item $C_i$ is encodable in quasi linear time and decodable in quadratic time.
        \item The decoding failure probability is $\exp(-\Omega(n))$.
        \item The rate of the $C_i$ is 
        $R > \textup{Cap}(\bdcMTRLB) - \varepsilon$.
    \end{enumerate}
\end{theorem}

\begin{remark}
    \em
	Note that this section's main result recovers the result of the previous section with worse time complexity. 
	This is because in the decoding procedure, we simply use the matching algorithm given by \cite{HS21} instead of the improved, more involved indexing scheme given in \cite{haeupler2019near}. It can be the case that these techniques also work in this case. 
    However, for simplicity, we are using the simple matching algorithm of \cite{HS21} in a black-box way (see \Cref{lem:hs-matching-alg}).
\end{remark}
Before diving into the technical details, we present a broader picture, including a description of the main ideas of \cite{BLS20}.

Let $x$ be a random string in $\zo^n$. The average trace reconstruction problem (which has been extensively studied in the literature \cite{batu2004reconstructing,holden2018subpolynomial,chase2019new,holden2020subpolynomial,chen2022near}, just to name a few) asks for the minimal number of traces $T$ for which the reconstruction of $x$ succeeds with high probability (the randomness is over the randomness of the channel and the choice of $x$). 
Although there has been significant interest in this question, there is still an exponential gap between the lower bound ($\widetilde{\Omega}(\log^{5/2}n)$ \cite{chase2019new}) and upper bound on $T$ ($\exp(\widetilde{O}(\log^{1/5} n))$ \cite{Rub23}).

Motivated by the challenges in developing DNA-based storage systems, \cite{CGMR20} and then a subsequent work \cite{BLS20} introduced the \emph{coded trace reconstruction} problem, a natural extension of the trace reconstruction problem. 
In this setting, the goal is to design a code with rate close to $1$ as possible such that any codeword can be reconstructed from (very) few traces, with high probability. 
In both works, the authors provide constructions of efficient codes with rate $1 - \varepsilon$ that are trace-reconstructable $O_\varepsilon (1)$ (see Remarks 1.6 and 1.7 in \cite{BLS20}, which compare the results of both papers).

The main theorem of \cite{BLS20} (Theorem 1.4) turns an upper bound on $T$ in the average case into an explicit and efficient code with rate $1 - \varepsilon$ that is efficiently trace-reconstructable. Applying their result with the current best upper bound, they that the number of traces is $\exp(O_d(\log ^{1/3}(\varepsilon^{-1})))$.

The major problem in generalizing the concatenated scheme from the previous section to the multi-trace setting is that we need to synchronize the traces. 
Indeed, assume that we have an inefficient code that achieves capacity in the multi-trace setting and that we use the same encoding process as in the previous section and denote by $a_1, \ldots, a_n$ the encoded inner codewords. 
Let $z_1, \ldots, z_T$ be the received traces. We can easily show that in most traces, most of the buffers are detected and that there are not too many spurious buffers. 
Thus, from each trace $t\in [T]$ we extract $n_t$ binary strings $b^{(t)}_1, \ldots, b^{(t)}_{n_t}$ where each one of these strings is a trace of an inner codeword.
However, to perform trace reconstruction to recover, say $a_i$, we need to know in each trace $t$, which $j\in [n_t]$ is such that $b^{(t)}_j$ is the corresponding trace of $a_i$ in the $t$-th trace.

To overcome this problem, \cite{BLS20} used two inner codes in their encoding process. We shall do the same thing here.
The first code encodes the information symbols and is robust in the multi-trace setting. 
The second inner code encodes synchronization symbols (these are symbols that do not carry data information but are used to preserve synchronization in the presence of insertions and deletions)
and can be reliably recovered from only a \emph{single} trace of the channel. 
More precisely, the encoding process takes a pair of information symbol and synchronization symbol $(r_i, s_i)$ and output $(a_i, b_i)$ where $a_i$ and $b_i$ are binary codewords in two different
inner codes, as described above. Then, we shall place the buffers in a slightly different way than what we did in the previous section (also slightly different than what \cite{BLS20} did), but the main idea is similar.

Note that now, before performing the ``inner'' trace reconstruction step, we add another step (as in \cite{BLS20}) that aligns each trace using the synchronization symbols (in each trace most of them are decoded correctly). 
Then, in each trace, we know the right location of most of the corrupted inner codewords, and we can run the inner trace reconstruction.

\subsection{Auxiliary results}
In this subsection, we define and cite several ingredients that we will use in our construction.
\paragraph{Synchronization strings.}
Synchronization strings are special \emph{indexing strings} that were first introduced in \cite{HS21}.
\begin{definition}
    A string $S\in \Sigma^n$ is an $\eta$-synchronization string if for every $1\leq i < j<k \leq n$, it holds that $\ed(S[i,j), S[j,k)) > (1 - \eta)(k-i)$.
\end{definition}
The following theorem shows that one can efficiently construct such strings over a `not too large' alphabet.
\begin{theorem}[Theorem 1.2, \cite{cheng2019synchronization}] \label{thm:sync-string-alphabet-size}
For every natural number \( n \) and every \( \eta \in (0,1) \), there exists a polynomial-time algorithm that constructs a \( \tau \)-synchronization sequence over an alphabet of size \( O(\eta^{-2}) \).
\end{theorem}

Synchronization strings were the main ingredient in \cite{HS21} for constructing codes that can correct insdel errors. 
The main idea behind \cite{HS21} is to take a code $\cC$ that can correct substitutions and erasures and combine it with a single $\eta$-synchronization string $S=s_1\cdots s_n$ of length $n$. 
The resulting code is 
\[
\cC^{\text{ID}} := \{ ((c_1, s_1),\ldots, (c_n,s_n)) \mid c\in \cC  \}\;.
\]
Note that the synchronization string is the same for all codewords and therefore does not carry information. Therefore, to ensure that the synchronization string has a negligible impact on the rate, the alphabet size of the code $\cC$ must be large compared to the alphabet size of the synchronization string.

The purpose of the synchronization symbols is to transform the insdel errors into erasures and substitutions. This is achieved by a ``matching'' algorithm given by the following lemma.
\begin{lemma}[\protect{\cite[Lemma 2.2]{HS2021-survey}}] \label{lem:hs-matching-org}
    Let $S=s_1s_2\cdots s_n$ be an $\eta$ synchronization string. 
    There exists an algorithm that on input $(m'_1,s'_1), \ldots, (m'_{n'}, s'_{n'})$, and $S=s_1\cdots s_n$, guesses the position of all received symbols in the sent string such that the position of all but $O(\sqrt{\eta} \cdot n)$ of the symbols that are not deleted are guessed correctly.
    This algorithm runs in time $O_{\eta}(n^2)$.
\end{lemma}

The algorithm which yields this lemma \cite[Algorithm 1]{HS2021-survey} essentially transforms the $\delta n$ insdel errors applied to a codeword $((c_1, s_1),\ldots,(c_n,s_n))$ into erasures and substitutions that are applied to $(c_1, \ldots, c_n)$. 
Specifically, the algorithm outputs $(y_1, \ldots, y_n)\in (\Sigma_{\cC} \cup \{ \perp \})^n$ such that $(y_1, \ldots, y_n)$ can be obtained from $(c_1, \ldots, c_n)$ by performing $e$ erasures and $t$ substitutions where $e + 2t \leq \delta n + 12 \sqrt{\eta}n$.

We will rephrase \Cref{lem:hs-matching-org} so that it fits our decoding algorithm. 
We start with a definition of a \emph{true indexing matching}.
\begin{definition}
    Let $S=s_1\cdots s_n$ be a string and let $S' = s'_1\cdots s'_m$ be a string that is obtained from $S$ via $\delta n$ insdel errors. 
    A \emph{true indexing matching} between $S$ and $S'$ consists of a set $I\subseteq [n]$ and a map $\Gamma:I \to [m]$ such that for all $a\in I$:
    \begin{itemize}
        \item The symbol in the $a$-th position in $S$ ($s_a$) was not deleted. 
        
        \item The new position of $s_a$ in $S'$ is $\Gamma(a)$. 
    \end{itemize}
\end{definition}
\begin{lemma}[\protect{\cite[Section 2, implicit]{HS2021-survey}}]
\label{lem:hs-matching-alg}
    Let $S=s_1s_2\cdots s_n$ be an $\eta$ synchronization string and let $S' = s'_1\cdots s'_m$ be a string that is obtained from $S$ via $\delta n$ insdel errors. 
    Then, there exists a matching algorithm that produces indices $i_1< \ldots < i_t$ and $j_1 < \ldots< j_t$ such that
    \begin{itemize}
        \item $t \geq n -\delta n - 2\sqrt{\eta}n$.
        \item $s_{i_{\ell}} = s'_{j_{\ell}}$ for all $\ell \in [t]$.
        
        \item There are at most $\sqrt{\eta} n$ indices $\ell \in [t]$ such that $\Gamma(i_{\ell}) \neq j_{\ell}$.
    \end{itemize}
\end{lemma}
\paragraph{Good codes correcting substitutions.} 
For our scheme we will need to use an outer code over a large (but constant) alphabet that can correct substitutions. 
In \cite{BLS20}, the authors adapted the efficient binary codes given in \cite{guruswami2005linear} into efficient codes over any alphabet that is a power of $2$. They used these codes to construct codes over \emph{large} (but constant) alphabets that are trace reconstructible.
\begin{proposition} [\protect{\cite[Proposition 2.17]{BLS20}}]\label{prop:bls-outer-code}
    For every $\varepsilon$ and $\Sigma$ whose size is a power of $2$, there exists an infinite
    family of codes over $\Sigma$ of rate $1 - \varepsilon$ encodable in linear time and decodable in linear time from up
    to a fraction $\frac{1}{40} \varepsilon^3$ of worst-case substitution errors.
\end{proposition}

\begin{remark}
    \em
    We note that in \cite{BLS20}, for their binary codes, the authors used a construction of Justesen \cite{justesen1972class} that gives rate $1- \varepsilon$ codes correcting $\Theta(\varepsilon^2/\log(\varepsilon^{-1}))$ for \emph{all} large enough block length (see \cite[Proposition 2.17]{BLS20}). The caveat is that the run time of the encoder and decoder is quadratic. 
    In this paper, we focus on constructing an infinite family of capacity achieving codes, and thus we choose to use \Cref{prop:bls-outer-code}.
\end{remark}
\newcommand{\Bigbuf}{\frac{\nu}{16(1-d(M))} \cdot n_R}
\newcommand{\Smallbuf}{\floor{10^3 \cdot \beta \cdot \log n_R}}
\newcommand{\Thrbigbuf}{(1 - d(M))B/2}
\newcommand{\nout}{n_{\textnormal{out}}}
\newcommand{\nR}{\mu^{-1} \cdot \varepsilon^{-1} \cdot \log(1/\varepsilon)}
\subsection{Construction}
\label{sec:bls-like-const}
\paragraph{Parameters}
We start by introducing the components of our construction which include an outer code that can correct from substitutions and two inner codes that we obtain from \Cref{cor:dense-code-RL-dep}. We shall introduce parameters throughout the description of the construction. To aid readability, these are also listed in \Cref{table:params-multi-trace}. 
We note here that there was no attempt to optimize the parameters. Let $\varepsilon$ be the desired gap to capacity.

\begin{table}
\begin{center}
\begin{tabular}{ |c|c|c| } 
    \hline
    Parameter & Value & Description/Comments\\
    \hline
    $\varepsilon$ &  & Gap to capacity \\ 
    \hline
    $T$ & Constant & Number of traces \\
    \hline
    $\nout$ & $\rightarrow \infty$ & Length of $\Cout$ \\
    \hline
    $\delout$ & $\varepsilon^3/40$ & Substitution correction capability of $\Cout$ \\
    \hline
    $n_R$ & Large enough constant & Length of the inner code $C_R$ \\
    \hline
    $n_S$ & $\log n_R$ & Length of the inner code $C_S$ \\
    \hline
    $\mathcal{R}$ & Constant& Rate of the inner code $C_R$ \\
    \hline
    $\mu$ & Constant & Channel parameter \\
    \hline
    $M$ & Constant & Channel parameter \\
    \hline
    $\nu$ & $\varepsilon \cdot \mu$ & small constant for buffer size \\
    \hline
    $B$ & $\Bigbuf$ & Buffer size \\
    \hline
    $\xi$ & $\varepsilon^{4}/T$ & Number of bad pairs (see \Cref{def:good-pair-trace-align} below)\\
    \hline
    $\varepsilon_R, \varepsilon_S$ & $\varepsilon^{4}/T$ & Decoding error probability of $C_R$, $C_S$, respectively\\
    \hline
    $\zeta$ & $(1 - d(M))\nu/4$ & For \Cref{cor:dense-code-RL-dep} \\
    \hline
    $\gamma$ & $\gamma(\varepsilon_R, \zeta)$ & Given by \Cref{cor:dense-code-RL-dep} \\
    \hline
    $\eta$ & $\varepsilon^8/T$ & Synchronization parameter \\
    \hline

\end{tabular}
\end{center}
\caption{Parameters of the scheme. All parameters are constant with respect $\nout$.}
\label{table:params-multi-trace}
\end{table} 
\paragraph{Outer code.} Let $\Cout$ be a code of length $\nout$ and rate $1 - \varepsilon/4$ over the alphabet $\Sigma_R$ that can correct $\delout = \varepsilon^3/40$ fraction of worst-case substitution errors. This code exists by \Cref{prop:bls-outer-code}. 

\paragraph{Inner codes.} As in \cite{BLS20}, we shall encode the data and the synchronization symbols separately using two inner codes. Set $\zeta = (1 - d(M))\nu/4$ and $\varepsilon_R=\varepsilon_S = \varepsilon^4/T$. 
\begin{enumerate}
    \item \textbf{Code for the synchronization symbols.} Let $\eta = \varepsilon^8/T$ and let $S = s_1\cdots s_{\nout}$ be an $\eta$-synchronization string of length $\nout$ over the alphabet $\Sigma_S$ where, by \Cref{thm:sync-string-alphabet-size}, we have that $|\Sigma_S| = O(T^2/\varepsilon^{16})$. However, we clearly can use a larger alphabet. 
    Let $C_S$ be a code for the $\bdcRLB$ (this is the single trace version, i.e., for $T = 1$) channel guaranteed by \Cref{cor:dense-code-RL-dep}  where $\mathcal{R}(C_S)$ denotes its rate ($\mathcal{R}(C_S)$ can be as close to capacity as we want). $C_S$ is equipped with an encoder $\Enc_S: \Sigma_S \to \zo^{n_S}$ and a decoder $\Dec_S:\zo^* \to \Sigma_S$ such that for every codeword $c\in C_S$ the probability that $\Dec_S$ fails to decode $c$ after transmission through the $\bdcRLB$ is at most $\varepsilon_S$. We also assume that $C_S$ starts with a $1$ bit and ends with a $0$ bit.
    We shall assume that
    $\log|\Sigma_S| = (\mathcal{R}(C_S))^{-1}\cdot \log n_R$ where $\mathcal{R}(C_S)$ where $n_R$ is a large enough constant to be determined at the end. 
    
    By \Cref{cor:dense-code-RL-dep}, the length of $C_S$ which is denoted by $n_S$ is at least $n(\varepsilon_S, \zeta)$ for some function $n$.

    \item \textbf{Code for the content symbols.} For the content symbols, i.e., the symbols of the outer code $\Cout$, we will use again the code guaranteed by \Cref{cor:dense-code-RL-dep}, but now for $T$ traces. Namely, we get a code $C_R$ of length $n_R$, with rate $\mathcal{R}(C_R) = \textup{Cap}(\bdcMTRLB) - \varepsilon_R$ that is $(T,\varepsilon_R)$-trace reconstructable for the $\bdcRLB$ channel. 
    $C_R$ is equipped with an (inefficient) encoder $\Enc_R: \Sigma_R \to \zo^{n_R}$ where $\log|\Sigma_R| = (\mathcal{R}(C_R))^{-1}\cdot n_R$ and an (inefficient) decoder $\Dec_R:\zo^* \to \Sigma_R$ such that for every codeword $c\in C_R$, the probability that $\Dec_R$ fails to decode $c$ after transmission through the $\bdcMTRLB$ is at most $\varepsilon_R$. We assume that $C_R$ starts with a $0$ bit and ends with a $1$ bit.

    Note again that by \Cref{cor:dense-code-RL-dep}, $n_R$ is at least $n(\varepsilon_R, \zeta)$ for some function $n$.
\end{enumerate}

Since $n_R$ and $n_S$ can be as large as we want, starting from some $n(\varepsilon_S, \zeta)$, we make sure that $n_R = 2^{n_S}$. Let $\gamma$ is the constant obtained by \Cref{cor:dense-code-RL-dep} when given the parameters $\varepsilon_S$ and $\zeta$ as declared here.  

\paragraph{Encoding.} 
 
Let $\nu = \varepsilon \cdot \mu$ and set $B = \Bigbuf$. By choosing a large enough $n_R$, we make sure that $B>M$ where recall that $M$ is a channel parameter (for every run of length at least $M$, the deletion probability applied on each bit of this run is $d(M)$).

The encoding process is given next.
\begin{enumerate}
    \item Let $m\in \Sigma_R^{k_{\text{out}}}$ be a message. Encode it using the outer code $\Cout$ to obtain $(r_1, \ldots, r_{\nout}) \in \Sigma_R^{\nout}$.
    \item Encode each $r_i$ and $s_i$ using the respective inner codes. Namely, let $a_i = \Enc_R (r_i)$ and $b_i = \Enc_S(s_i)$. 
    \item Add buffers of the symbol $0$ of length $B$ between every $a_i$ and $b_i$ and add buffers of the symbol $1$ length $B$ between $b_i$ and $a_i$.
    The codeword has the following form
    \[
    c = a_1 \circ 0^{B} \circ b_1 \circ 1^{B} \circ a_2 \circ 1^{B} \circ b_2 \circ \cdots \circ a_{\nout} \circ 0^{B} \circ b_{\nout} \circ 1^{B} \;. 
    \]
\end{enumerate}

\paragraph{Rate.} The length of the codeword is $(n_R + 2B + n_S) \cdot \nout$. Thus, the rate is
\begin{align}
    \mathcal{R} &= \frac{\log |\Sigma_R|^{(1 - \varepsilon/4)\nout}}{(n_R + 2B + n_S) \cdot \nout}\\
    &= \frac{(1-\varepsilon/4)\cdot \mathcal{R} (C_R) \cdot n_R}{n_R + 2B + n_S} \nonumber \\
    &=\frac{(1 - \varepsilon/4) \cdot \mathcal{R}(C_R)}{1 + \frac{2\nu}{16(1 - d(M))} + \frac{\log n_R}{n_R}} \nonumber \\
    &\geq \frac{(1-\varepsilon/4)\cdot \mathcal{R}(C_R)}{1 + \varepsilon/4} \nonumber \\
    &\geq (1 - \varepsilon/2) \cdot \mathcal{R}(C_R) \label{eq:rate-multi-trace-const}\\
    &\geq \textup{Cap}(\bdcMTRLB) - \varepsilon\;,
\end{align}
where the first inequality holds since, for large enough $n_R$, we have $(\log n_R)/n_R \leq \varepsilon/8$ and also since $1 - d(M) \leq \mu$ (by the channel's definition), $\nu/(4(1 - d(M)))\leq \varepsilon/4$. The last inequality follows since $\mathcal{R}(C_R) = \textup{Cap}(\bdcMTRLB) - \varepsilon_R$ where $\varepsilon_R = \varepsilon^4/T < \varepsilon/2$.

\paragraph{Decoding.} Let $z^{(t)}$ be the received string at the $t$-th trace. 
\begin{enumerate}
    \item Trace alignment using the synchronization symbols.
    \begin{enumerate}
        \item \label{item:big-buffers} Identify $1$-buffers. Every run of zeros of length greater than $\Thrbigbuf$ is identified as a $1$-buffer.
        
        Let $z^{(t)}_1, \ldots, z^{(t)}_{n_t}$ be the strings in between the buffers.
    
        \item \label{item:small-buffers} Identify $0$-buffers. For every $i\in [n_t]$, in $z^{(t)}_i$, every run of zeros of length greater than $\Thrbigbuf$ is identified as a $0$-buffer. If there are no $0$-buffers in $z^{(t)}_i$ or there are more than a single $0$-buffer, discard this $z^{(t)}_i$. Otherwise, divide this $z^{(t)}_i$ into two parts according to the $0$-buffer that was found.
        
        Let $(x_1^{(t)}, y_1^{(t)}), \ldots, (x_{n'_t}^{(t)}, y_{n'_t}^{(t)})$ be the pairs that correspond to the $s^{(t)}_i$ containing a single $0$-buffer.
        
        \item \label{item:sync-symb-dec} For every $i\in [n'_t]$, decode the synchronization symbol from $y^{(t)}_{i}$. Namely, $\tilde{s}^{(t)}_i = \Dec_S(y^{(t)}_{i})$.
    
        \item \label{item:sync-match-alg} Run the algorithm from \Cref{lem:hs-matching-alg} on  $\tilde{s}^{(t)}_1,\ldots, \tilde{s}^{(t)}_{n'_t}$ to obtain indices $\tilde{i}^{(t)}_1,\ldots, \tilde{i}^{(t)}_{n''_t}$ and $\tilde{j}^{(t)}_1,\ldots, \tilde{j}^{(t)}_{n''_t}$.
    
        \item \label{item:set-up-trace} For every $p \in [n''_t]$, let $\tilde{a}^{(t)}_{\tilde{i}^{(t)}_{p}}:=x^{(t)}_{\tilde{j}^{(t)}_{p}}$ and for $p \in [\nout]\setminus  \{ \tilde{i}^{(t)}_1,\ldots, \tilde{i}^{(t)}_{n''_t} \}$, set $\tilde{a}^{(t)}_{p} = \perp$.
    \end{enumerate}
    \item \label{item:trace-decode} Trace reconstruction. For all $i\in [\nout]$, let $\tilde{r_i} = \Dec(\tilde{a}^{(1)}_i,\ldots, \tilde{a}^{(T)}_i)$.
    \item \label{item:outer-decoder} Outer code correction. Run $\Dec_{\text{out}}(\tilde{r}_1, \ldots, \tilde{r}_{\nout})$ and return the output.
    \end{enumerate}

\subsection{Analysis}
The following claims (Claims~\ref{clm:deleted-1-buffer}-\ref{clm:spu-1-buffers}) bound the probabilities of \emph{bad} events that are related to the correct identification of the buffers.
The first claim bounds the probability that a $1$-buffer is not detected in Step~(\ref{item:big-buffers}). 
\begin{claim} \label{clm:deleted-1-buffer}
    The probability that a $1$-buffer is not identified is at most $\exp(-\Omega_{\varepsilon,\mu}(n_R))$.
\end{claim}
\begin{proof}
    The length of a buffer is $B=\Bigbuf$. 
    In order for the decoding algorithm to consider it as a run in an inner codeword, the channel must delete at least $\Thrbigbuf$ bits from it. 
    Since every bit in the buffer is deleted independently with probability $d(M)$ (recall that $B> M$), the expected length of a buffer after going through the channel is $(1 - d(M)) B$. 
    By applying the Chernoff bound, we get that the probability that less than $(1 - d(M))B/2$ bits of buffer survived the transmission is at most 
    \[
    \exp(-(1 - d(M))B/8) \leq \exp(-\nu n_R/128) \leq \exp(-\Omega_{\varepsilon,\mu}(n_R))\;.
    \]
\end{proof}
We are now interested in bounding the probability that a $0$-buffer is not identified. 
\begin{claim} \label{clm:deleted-0-buffer}
    Let $z_i^{(t)}$ be a string obtained after Step~(\ref{item:big-buffers}). Assume that $z_i^{(t)}$ does not contain $1$s that belong to a $1$-buffer that was not identified by the algorithm and that the buffers that were identified before and after $z_i^{(t)}$ are indeed two adjacent $1$-buffers. 
    Then, the probability that the $0$-buffer inside $z_i^{(t)}$ not identified is at most $\exp(-\Omega_{\varepsilon,\mu}(n_R))$.
\end{claim}
\begin{proof}
    The proof is identical to the proof of \Cref{clm:deleted-1-buffer}.
\end{proof}

Our next goal is to bound the probability of detecting ``spurious'' $0$-buffers inside a $z_i^{(t)}$. That is, we upper bound the probability that the channel deletes many consecutive $1$s in such a way that a long run of $0$s is created where all of its bits belong to an inner codeword and the algorithm mistakenly thinks that it is a $0$-buffer.
\begin{claim} \label{clm:spurious-0-buf}
    Let $z^{(t)}_i$ be a string obtained after Step~(\ref{item:big-buffers}). 
    Assume that $z_i^{(t)}$ does not contain $1$s that belong to a $1$-buffer that was not identified by the algorithm and the buffers that were identified before and after $z_i^{(t)}$ are indeed two adjacent buffers. 
    Then, the probability that the algorithm identifies a $0$-buffer such that all of the bits of that buffer are bits of an inner codeword is at most $\Omega_{\varepsilon, \mu, \gamma}(n_R))$. \footnote{Note that if some of the bits belong to a $0$-buffer, then the $0$s that originally belonged to the inner codeword are merged to the $0$-buffer. Thus, a spurious $0$ buffer occurs only when all its $0$ bits belong to an inner codeword.}
\end{claim}
\begin{proof}
    The computation is similar to the one done in the previous section, yet a bit more delicate. Since the $1$-buffers surrounding $z_i^{(t)}$ were identified and correspond to adjacent $1$-buffers in the original codeword, $z_i^{(t)}$ contains bits that correspond to $a_{i'}\circ 0^B \circ b_{i'}$ for some $i'$.
    We write $z_i^{(t)}$ as $
    z_i^{(t)} = \tilde{a}_{i'}\circ 00\cdots0 \circ \tilde{b}_{i'}$
    where all the bits of $\tilde{a}_{i'}, \tilde{b}_{i'}$ correspond to $a_{i'},b_{i'}$, respectively, and the run of zeros in between contains $0$ bits that correspond to the $0$-buffer in the original codeword.
    
    In order to create a run of zeros (in which all the zeros correspond to inner codewords) of length at least $\Thrbigbuf$, the channel must delete all $1$s from at least an interval of this size. Note here that the length of $\tilde{b}_{i'}$ is $\leq \log n_R$ which is less than $\Thrbigbuf$, the threshold from which a buffer is considered (for large enough $n_R$). Thus, a spurious $0$-buffer can be created only inside $\tilde{a}_{i'}$.

    By the properties of our inner codes and by recalling that $\zeta = (1 - d(M))\nu/4$, any interval of length $\zeta n_R = (1 - d(M))B/4$ must contain at least $\gamma \zeta n_R=\gamma(1 - d(M))B/4$ ones.
    Let $a_i, i\in [M-1]$ be the number of ones in this interval that belong to runs of length $i$ and denote by $a_M$ the number of bits that belong to runs of length $\geq M$.
    Deleting all these ones happens with probability
    \begin{align*}
    d(1)^{a_1}d(2)^{a_2} \cdots d(M)^{a_M} &\leq \left( \frac{4\cdot \sum_{i=1}^{M}a_i d(i)}{\gamma (1 - d(M))B} \right)^{\frac{\gamma (1 - d(M))B}{4}} \\
    &\leq (1- \mu)^{\gamma\nu n_R/8} \\
    &\leq \exp(-\Omega_{\varepsilon, \mu, \gamma}(n_R))
    \end{align*}
    where the first inequality follows from the weighted AM-GM inequality and the second inequality is due to the assumption that $d(1) \leq d(2) \leq \cdots \leq d(M)\leq 1- \mu$ and that $\sum_{i=1}^M a_i = \gamma (1 - d(M))B/4$. 

    Now, union bounding over all such intervals, the probability that there exists a spurious buffer is at most $n_R \cdot \exp(-\Omega_{\varepsilon, \mu, \gamma}(n_R)) = \exp(-\Omega_{\varepsilon, \mu, \gamma}(n_R))$.
\end{proof}
The final claim regarding the buffers bounds the probability that a spurious $1$-buffer is created inside the output of channel in an encoded pair $(a_i,b_i)$.
\begin{claim} \label{clm:spu-1-buffers}
    Assume that $a_i\circ 0^B \circ b_i$ is transmitted through the channel and let $\hat{z}$ be the received output. 
    Then, the probability that a spurious $1$-buffer is identified in $\hat{z}$ is at most $\exp(-\Omega_{\varepsilon, \mu, \gamma}(n_R))$.
    Moreover, the expected number of spurious $1$-buffers in $\hat{z}$ is at most $\exp(-\Omega_{\varepsilon, \mu, \gamma}(n_R))$ and the maximal number of spurious buffers in $\hat{z}$ is at most $32\nu^{-1}$.
\end{claim}
\begin{proof}
    The claim about the probability to have a spurious $1$-buffer in $\hat{z}$ (and the bound on the expected number of spurious $1$-buffers) follows by the same argument as in \Cref{clm:spurious-0-buf}. The size of a buffer is at least $(1 - d(M))B/2$. Thus, the number of spurious buffers is at most 
    $\frac{2n_R}{(1 - d(M))B}\leq \frac{32}{\nu}$.
\end{proof}

Next, we define a pair $(x^{(t)}_i, y^{(t)}_i)$ (these pairs are obtained in Step~\ref{item:small-buffers}) to be a \emph{good pair} if there exists a $j\in [\nout]$ such that all the bits of $x^{(t)}_i$ belong to $a_{j}$ and all the bits of $y^{(t)}_i$ belong to $b_{j}$. 
Equivalently, this means that the $1$-buffers that are to the left and to the right of $a_{j} \circ 0^B\circ b_{j}$ were identified (and there were no spurious $1$-buffers inside) and the $0$-buffer between $a_j$ and $b_j$ was identified (and there were no spurious $0$-buffers). 
\begin{definition} \label{def:good-pair-trace-align}
    Let $z^{(t)}$ be the $t$-th trace and let $(x_1^{(t)}, y_1^{(t)}), \ldots, (x_{n'_t}^{(t)}, y_{n'_t}^{(t)})$ be the output of Step (\ref{item:small-buffers}). 
    We call a pair $(x_i^{(t)}, y_i^{(t)})$ a \emph{good pair} if there exists $j\in [\nout]$ such that all the bits of $x_i^{(t)}$ are bits of $a_j$ and all the bits of $y_i^{(t)}$ are bits of $b_j$.
\end{definition}

Our next proposition states that, with high probability, after Step~(\ref{item:small-buffers}), there are at least $(1 - \xi)\cdot \nout$ good pairs. Formally, we have the following.
\begin{proposition} \label{prop:number-good-pairs}
    Let $z^{(t)}$ be the $t$-th trace. Let $(x_1^{(t)}, y_1^{(t)}), \ldots, (x_{n'_t}^{(t)}, y_{n'_t}^{(t)})$ be the output of Step (\ref{item:small-buffers}). 
    Then, with probability $1 - \exp(-\Omega(\nout))$, at least $\nout(1 - \xi)$ of them are good pairs. 
\end{proposition}

\begin{proof}
    Let $z^{(t)}_1, \ldots, z^{(t)}_{n_t}$ be the strings obtained in Step~\ref{item:big-buffers} after identifying $1$-buffers.
    We start by counting how many $z^{(t)}_i$s are such that the $1$-buffers to the right and to the left of the $z^{(t)}_i$ are two adjacent \emph{genuine} $1$-buffers (by \emph{genuine} we mean that bits of the identified buffer are indeed bits of a $1$-buffer and not bits of run in an inner codeword). We shall call such a $z^{(t)}_i$ a \emph{good} $z^{(t)}_i$.

    By \Cref{clm:deleted-1-buffer}, the probability that a specific $1$-buffer is deleted is at most $\exp(-\Omega_{\varepsilon,\mu}(n_R)) < \xi/10$ where the inequality holds for large enough $n_R$ (recall that $\xi = \varepsilon^4/T$ and that $n_R$ will be a constant that will be chosen in the end). Thus, the probability that more than $\xi \nout/9$ genuine $1$-buffers are not identified is at most $\exp(-\Omega(\nout))$.
    By \Cref{clm:spu-1-buffers}, the expected number of spurious $1$-buffers between two adjacent genuine $1$-buffers is at most $\exp(-\Omega_{\nu, \mu, \gamma} (n_R) ) < \xi/10$ and the maximal number of spurious buffers in the channel output on $a_j\circ 0^B \circ b_j$ is at most $32\nu^{-1}$.
    Thus, we can apply Hoeffding's bound (\Cref{lem:hoeff}), and get that the probability of having more than $\xi \nout/9$ spurious $1$-buffers in the whole received string is at most $\exp(-\Omega(\nout))$.
    Thus, with probability $1 - \exp(-\Omega(\nout))$ at least $\nout (1 - \xi /9)$ genuine $1$-buffers were identified and at most $\xi \nout/9$ spurious buffers were identified. Consequently, it is easy to observe that there must be at least $\nout (1 - 5\xi/9)$ $z^{(t)}_i$s that are good $z^{(t)}_i$s. Indeed, every deleted genuine buffer in the worst case transforms two potential good $z_i$s into one bad $z_i$, and every spurious $1$-buffer transforms one potential good $z_i$ into two bad $z_i$s.

    We now turn to look inside the good $z^{(t)}_i$s. A good $z^{(t)}_i$ will be transformed into a good pair (see \Cref{def:good-pair-trace-align}) if the $0$-buffer is correctly identified and there are no spurious $0$-buffers inside. 
    According to \Cref{clm:deleted-0-buffer}, the $0$-buffer is deleted with probability $\exp(-\Omega_{\varepsilon,\mu}(n_R)) < \xi/10$, and according to \Cref{clm:spurious-0-buf} the probability to have a spurious $0$-buffer in $z^{(t)}_i$ is at most $\exp(-\Omega_{\varepsilon,\mu,\gamma}(n_R)) < \xi/10$. Thus, with probability at least $1 - \xi/5$, a good $z^{(t)}_i$ is transformed into a good pair $(x_i^{(t)}, y_i^{(t)})$. 
    Note that by the behavior of the channel, the event that $z_i^{(t)}$ is transformed into a good pair is independent of all other $z_i^{(t)}$.
    Thus, applying a Chernoff bound we get that with probability $1 - \exp(-\Omega(\nout))$, there are at least $(1 - \xi/4)\cdot \nout(1 - 5\xi/9) \geq \nout(1 - \xi)$ good pairs.
    \end{proof}
    \begin{remark}
        \em
        We emphasize that in our construction $n_R$ is a constant (with respect to $\nout$) that can be made as large as we wish. 
        Specifically, we need $n_R$ to be such that all the inequalities that involve it in the proof of \Cref{prop:number-good-pairs} will hold. Also, recall from the construction that $n_R$ is also at least $2^{n(\varepsilon_S,\zeta)}$ for some function $n$.
        This implies that there exists $n_{R,0} = n_{R,0}(\varepsilon, \varepsilon_S, \mu,\gamma,\zeta)$ such that for every $n_R > n_{R,0}$, all the inequalities are true.
    \end{remark}
Our next claim shows that, with high probability, the synchronization symbol is successfully decoded in almost all of the good pairs.
\begin{claim} \label{clm:correct-sync-symb}
    With probability at least $1 - \exp(-\Omega(\nout))$, there are at least $(1- 2\varepsilon_S) \cdot \nout(1 - \xi)$ good pairs $(x_i^{(t)}, y_i^{(t)})$ for which $y_i^{(t)}$ was correctly decoded.
\end{claim}
\begin{proof}
    $y_i^{(t)}$ is the output of a $01$-trimming $\bdcRLB$ on some input $x$. According to the inner code that we chose for the synchronization symbols, we have the assumption that the probability of decoding error is at most $\varepsilon_S$. Thus, by the independentness induced by the channel, we can apply a Chernoff bound (\Cref{lem:chernoff}) and get the claimed result.
\end{proof}

Now, we prove the correctness of our decoding algorithm.
\begin{proposition} \label{prop:multi-trace-failure-prob}
    Fix the parameters as described in \Cref{sec:bls-like-const} and \Cref{table:params-multi-trace}. Let $m\in \Sigma_{\textup{out}}^{(1 - \varepsilon/4)\nout}$ be a message that is encoded to a codeword $c$ using our encoding algorithm. Let $z^{(1)},\ldots, z^{(T)}$ be $T$ traces of $c$ under the channel $\bdcRLB$.
    Then, with probability $1 - \exp(-\Omega(\nout))$, the decoding algorithm on input $z^{(1)},\dots,z^{(T)}$ outputs the correct message $m$. 
\end{proposition}
\begin{proof}
    For a fixed trace, say the $t$-th trace, by \Cref{prop:number-good-pairs} and \Cref{clm:correct-sync-symb}, with probability at least $1 - \exp(-\Omega(\nout))$ there are at least $(1 - 2\varepsilon_S)(1 - \xi) \nout $ of pairs $(x_i,y_i)$ for which the decoding of $y_i$ resulted in the correct synchronization symbol. Also, observe that the number of total pairs (good and bad), denoted by $n'_t$ in Step~(\ref{item:small-buffers}) can be upper bounded by $n'_t \leq n_t \leq (1 + 2\xi/9)\nout$. Indeed, $n_t$ is the number of $z_i$s in Step~(\ref{item:big-buffers}), and clearly in Step~(\ref{item:small-buffers}) this number can only decrease. Furthermore, note that $n_t - \nout \leq 2\cdot \#\textup{spurious buffers}$, and recall that by the proof of \Cref{prop:number-good-pairs}, with probability $1- \exp(-\Omega(\nout))$ the number of spurious buffers is at most $\xi \nout/9$.
    
    This implies that the symbols $\tilde{s}_1^{(t)}, \ldots, \tilde{s}_{n'_t}^{(t)}$ obtained in Step~(\ref{item:sync-symb-dec}) can be obtained from $S=s_1 \cdots s_n$ by performing at most $(2\varepsilon_S + \xi) \nout$ deletions and at most \[
    n'_t - (1- 2\varepsilon_S - \xi) \nout 
    \leq (1 + 2\xi/9)\nout - (1- 2\varepsilon_S + \xi) \nout  \leq (2\varepsilon_S +2\xi)\nout
    \] 
    insertions. In total, there are at most $(4\varepsilon_S + 3\xi) \nout$ insdel errors between $S$ and $\tilde{s}_1^{(t)}, \ldots, \tilde{s}_{n'_t}^{(t)}$.
    By \Cref{lem:hs-matching-alg}, the number of correctly matched synchronization symbols is at least $\nout (1 - 4\varepsilon_S - 3\xi - 3\sqrt{\eta})$. 
    Thus, after Step~(\ref{item:set-up-trace}), we have $(\tilde{a}_1^{(t)}, \ldots, \tilde{a}_{\nout}^{(t)})$ where at least $\nout (1 - 4\varepsilon_S - 3\xi - 3\sqrt{\eta})$ indices $j\in [\nout]$ are such that $\tilde{a}_{j}^{(t)}$ is a trace of a $10$-trimming $\bdcRLB$ channel applied to $a_{j}$. In this case, we say that the \emph{trace-alignment procedure succeeded} and it happens with probability $1 - \exp(-\Omega(\nout))$.
    Union bounding over all $T$ traces, we have that with probability at least $1 - T \cdot \exp(-\Omega(\nout)) = 1 - \exp(-\Omega(\nout))$, the trace-alignment procedure succeeds for all traces. 

    Now, given that the trace-alignment procedure succeeded for all traces, we turn to count for how many indices $j \in [\nout]$ the vector $(\tilde{a}^{(1)}_{j}, \ldots, \tilde{a}^{(T)}_{j})$ represents indeed $T$ traces of $a_{j}$. This happens if for all traces the pair $(a_{j},b_{j})$, after being encoded and transmitted through the channel complies with the following conditions for the $t$-th trace, for all $t\in\{1,\dots,T\}$: 
    \begin{itemize}
        \item After Step~(\ref{item:big-buffers}), there exists an index $i$ such that \emph{all} the bits in $z^{(t)}_i$ correspond to $a_{j} \circ 0^{B} \circ b_{j}$.
        \item After Step~(\ref{item:small-buffers}), the pair $(x^{(t)}_i, y^{(t)}_i)$ is a good pair. In particular, all the bits of $x^{(t)}_i$ correspond to $a_{j}$ and all the bits of $y^{(t)}_i$ correspond to $b_{j}$.
        \item After Step~(\ref{item:sync-symb-dec}), we have $\tilde{s}^{(t)}_i = s_{j}$, namely, the decoding of $y^{(t)}_i$ to the synchronization symbol succeeded.
        \item After Step~(\ref{item:sync-match-alg}), we get that the positioning of $\tilde{s}^{(t)}_i$ in the synchronization string $s_1 \cdots s_n$ is correct. 
    \end{itemize}

    In each trace, there are at most $(4\varepsilon_S + 3\xi + 3\sqrt{\eta}) \nout$ indices $j\in [\nout]$ for which one of the conditions does not hold. 
    Thus, in total there are at most $T \cdot (4\varepsilon_S + 3\xi + 3\sqrt{\eta}) \nout$ indices $j \in [\nout]$ for which there is a trace $t\in [T]$ such that one of the conditions above does not hold.

    Now, recall that our code $C_R$ is such that the decoding failure probability under the $10$-trimming $\bdcMTRLB$ channel is $\varepsilon_R$. Thus, the probability that there are more than $(1 - 2\varepsilon_R) \cdot \nout(1 - 4T\varepsilon_S - 3T\xi - 3T\sqrt{\eta})$ failures when trace-decoding these symbols is at most $\exp(-\Omega(\nout))$. 
    Recalling that $\xi = \varepsilon_R = \varepsilon_S = \varepsilon^4/T$ and $\eta = \varepsilon^8/T$, with probability $1 - \exp(-\Omega(n))$ we get that for at least a
    \[
    \left(1 - \frac{2\varepsilon^4}{T}\right) \cdot \left( 1 - 10\varepsilon^4 \right) \geq 1 - 12\varepsilon^4 \geq 1 - \frac{\varepsilon^3}{40} \;
    \]
    fraction of indices $j \in [\nout]$ the trace-decoding procedure of $C_R$ succeeds in Step~\ref{item:trace-decode}.
    Therefore, with probability $1 - \exp(-\Omega(n))$ the Hamming distance between $(r_1, \ldots, r_n)$ (the original outer codeword) and  $\tilde{r}_1, \ldots, \tilde{r}_n$ is at most $\varepsilon^3 \nout/40$ and thus, by \Cref{prop:bls-outer-code}, the decoder of the outer code can decode it.
\end{proof}

We now prove \Cref{thm:efficient-bounded-multi-trace-rl}.
\begin{proof}[Proof of \Cref{thm:efficient-bounded-multi-trace-rl}]
    Note that the claim about the rate is given in \Cref{sec:bls-like-const}, specifically, the rate is computed in \cref{eq:rate-multi-trace-const}. 
    Further, observe that the claim about the decoding failure probability is given by \Cref{prop:multi-trace-failure-prob}. 
    Thus, we are left to justify the claim about the complexity of the encoding and decoding algorithms. 

    We start with the encoding algorithm. In the first step we run the encoding algorithm of the inner codes which runs in time $O(\nout)$. 
    Then, in the second step, each symbol in the outer codeword is encoded using the encoder of $C_R$ and every symbol of the synchronization symbol is encoded using the encoder of $C_S$. Recall that the lengths of $C_R$ is $n_R$ (and the length of $C_S$ is $\log n_R$) which is constant with respect to $\nout$.
    Therefore, the complexity of this step is $O(\nout)$. Finally, placing the buffers in the third step also takes linear time. Thus, the encoder takes $O(\nout)$ time.

    Now, we analyze the complexity of the decoding algorithm. In the trace-alignment step, Steps~\ref{item:big-buffers} and~\ref{item:small-buffers} which search for buffers are performed in linear time. In Step~\ref{item:sync-symb-dec}, we run the decoder of $C_S$ on at most $\nout$ symbols and thus this is done in $O(\nout)$ time. In Steps~\ref{item:sync-match-alg} and~\ref{item:set-up-trace} we run the matching algorithm from \Cref{lem:hs-matching-alg} which runs in time $O(\nout^2)$ and then simply reposition the output.
    Thus, the trace-alignment step takes $O(\nout^2)$ time.
    Then, Step~\ref{item:trace-decode} which decodes every trace using the decoding algorithm of $C_R$ (again, recall that $n_R$ is constant with respect to $\nout$), takes $O(\nout)$ time. Finally, the decoding of the corrupted outer codeword in Step~\ref{item:outer-decoder} takes $O(\nout)$ according to \Cref{prop:bls-outer-code}.
    We conclude that the decoder takes $O(\nout^2)$ time.
\end{proof}

\begin{remark}
    \em
    We emphasize here the order by which the parameters are set in order to make sure that there are no circular dependencies.
    First, the channel parameters are $\gamma, M, T,$ and $\mu$. 
    We let $\varepsilon$ be the desired gap to capacity we want to achieve. This sets the parameters $\varepsilon_S, \varepsilon_R, \xi, \nu, \eta,\delout$, and $\gamma$.
    Now, we choose large enough $n_R$ so that all the constraints hold (those in the \Cref{prop:number-good-pairs} and those imposed by the \Cref{cor:dense-code-RL-dep} and the construction).

    We emphasize that all the parameters mentioned are independent of $\nout$.
\end{remark}

\section{Lower bounds on the capacity of a threshold deletion channel}\label{sec:lowerbounds}

In this section, we consider a simple example of a runlength-dependent channel and compute lower bounds on its capacity. The channel is defined next.
\begin{definition}
    Let $\bdctzerod$ be a runlength-dependent deletion channel that acts as follows on every input run of length $\ell$. If $\ell<\tau$, the channel leaves this run as is. Otherwise (if $\ell \geq \tau$), every bit of this run is deleted independently with probability $d$.

\end{definition}
As discussed in \cref{sec:intro}, the motivation to define and study such a channel comes from error patterns observed in DNA-based storage systems.  
We provide two lower bounds on the capacity of $\bdctzerod$, which we plot for the special cases $\tau=2,3$. 
We emphasize here that since this channel is a runlength-dependent deletion channel and complies with \Cref{def:run-length-bounded-channel}, both lower bounds can be turned into explicit and efficient codes using \Cref{thm:efficient-bounded-rl}.

\subsection{First lower bound}

Here we closely follow the approach of Diggavi and Grossglauser \cite{diggavi2006information} and Drinea and Mitzenmacher \cite{drinea2006lower} and provide a simple lower bound on the capacity of $\bdctzerod$. We prove the following.
\begin{theorem} \label{thm:dg-like-bound}
    Let $d\in [0,1]$ and $\tau\in \mathbb{N}$ be such that $d \cdot \frac{\tau+1}{2^\tau} \leq 1/2$. The capacity of $\bdctzerod$ is at least $1 - h\left(d \cdot \frac{\tau+1}{2^\tau} \right)$.
\end{theorem}

Before proving \Cref{thm:dg-like-bound}, we first recall McDiarmid's concentration inequality, which extends Hoeffding's concentration lemma. 
We then use this inequality to prove that, with high probability, the fraction of bits that belong to runs of length at least $\tau$ in a uniformly random binary string is close to $\frac{\tau+1}{2^\tau}$.

\begin{lemma}[{McDiarmid's inequality,  \cite[Lemma 1.2]{mcdiarmid1989method}}]\label{lem:mcdiarmid}
    Let $X_1, \ldots, X_n$ be independent random variables, where each $X_k$ takes values in a set $A_k$. Assume there exists a function $f: A_1 \times \ldots \times A_n \to \mathbb{R}$ and constants $c_1, \ldots, c_n$ such that for every $k \in [n]$ and for any two vectors $a, a' \in A_1 \times \ldots \times A_n$ that differ only in their $k$-th coordinate we have
    \[
    |f(a) - f(a')|\leq c_k \;.
    \]
    Then, we have that 
    \[
     \Pr \left[ \left|f(X_1, \ldots, X_n) - \bbE[f(X_1, \ldots, X_n)] \right| \geq t \right] \leq 2\exp\left( \frac{-2t^2}{\sum_{k=1}^nc_k^2}\right)\;.
    \]
\end{lemma}
\begin{proposition}\label{prop:long-bits}
    Let $X$ be a uniformly random binary string in $\zo^N$ and let $\tau$ be constant with respect to $N$.
    Then, with probability $1 - \exp(-\Theta(N^{1/3}))$, the number of bits that belong to runs of length at least $\tau$ is at most $(\frac{\tau+1}{2^\tau} + \Theta(N^{-1/3}))\cdot N$.
\end{proposition}
\begin{proof}
    Let $X_i$ denote the $i$-th bit of $X$. A run of length $j < N$ beginning at position $s \in [2, N-j]$ means that $X_{s-1} \neq X_s$, the bits $X_s, X_{s+1}, \ldots, X_{s+j-1}$ are all equal, and $X_{s+j} \neq X_s$. 
    The probability of this event is precisely $2/2^{j+2} = 2^{-j-1}$, where the factor of $2$ accounts for the choice of $X_s$, which may be either $0$ or $1$.
    The probability that a run of length $j < N$ starts at position $s=1$ or $s=N-j + 1$ is $2/2^{j+1} = 2^{-j}$ (in this case, one of the boundaries of the run does not exist).
    Denote by $Z_j$, the number of runs of length exactly $j$ in $X$. We have
    \begin{equation} \label{eq:expected-num-of-runs-j}
        \bbE [Z_j] = 2\cdot 2^{-j} + (N - j-1) \cdot 2^{-j-1} = \frac{N - j + 3}{2^{j +1}}\;.
    \end{equation}
    Now observe that $Z_j$ is a a function of $X_1, \ldots, X_n$ which are independent random variables taking values in $\zo$. 
    Also, note that changing a single bit in a binary strings can change the number of runs of length exactly $j$ by at most $2$. Thus, we apply \Cref{lem:mcdiarmid} with $c_k = 2$ for all $k\in [n]$ and $t = N^{-1/3}$, and get
    \begin{align*}
    \Pr\left[Z_j < (1 - N^{-1/3})\bbE[Z_j])\right] \leq \exp(-\Theta(N^{1/3})),
    \end{align*}

    Now, as $\tau$ is constant with respect to $N$, we conclude that the probability that for all $i\in [\tau-1]$ we have $Z_j \geq (1 - N^{-1/3}) \cdot \bbE[Z_j]$ is $1 - \tau \cdot \exp(-\Theta (N^{1/3})) = 1 -  \exp(-\Theta (N^{1/3}))$. 

    In this case, the number of bits in $X$ that belong to runs of length $\geq \tau$ is 
    \begin{align*}
    N - \sum_{j=1}^{\tau-1} j \cdot Z_j &\geq N - (1 - N^{-1/3}) \sum_{j=1}^{\tau-1} \frac{(N - j + 3)j}{2^{j +1}} \\
    &= N - (1 - N^{-1/3})\cdot \left(\frac{N}{2} \cdot \sum_{j=1}^{\tau-1} \frac{j}{2^j} -  \sum_{j=1}^{\tau-1} \frac{3j - j^2}{2^{j+1}} \right)\\
    &= N - (1 - N^{-1/3})\cdot \left( \frac{N}{2}\cdot \left(2 - \frac{\tau + 1}{2^{\tau - 1}}\right) - O(1) \right) \\
    &= \frac{\tau + 1}{2^{\tau}} \cdot N - \Theta (N^{2/3}) \;,
    \end{align*}
    where the second equality is due to the identity $\sum_{i=1}^k i/2^i = 2 - (k+2)/2^k$ and our assumption that $\tau$ is constant with respect to $N$ making the second sum $O(1)$.
    The proposition follows.
    \end{proof}

To prove \Cref{thm:dg-like-bound}, we will need the following well-known result which gives the exact size of the insertion ball around a string.
\begin{lemma}[\protect{\cite[Equation 24]{levenshtein2001efficient}}] \label{lem:lev-superseq}
    Let $y\in \zo^{|y|}$ be a string. The number of strings $x\in \zo^n$ that contain $y$ as a subsequence is $\sum_{i=0}^{n - |y|} \binom{n}{i}$.
\end{lemma}
    
We are now ready to prove \Cref{thm:dg-like-bound}. The improvement over the bound of $1 - h(d)$ \cite{diggavi2006information} for the binary deletion channel comes from the observation that a typical output $Y$ in this channel is of length $\approx (1 - d\cdot \frac{\tau+1}{2^\tau})N$, which is greater than $\approx (1 - d)N$. The proof will closely follow the proof of \cite[Theorem 4.2]{diggavi2006information} and incorporate the necessary modifications.

\begin{proof}[Proof of \Cref{thm:dg-like-bound}]
     Generate a random codebook with $2^{NR}$ i.i.d. codewords chosen uniformly from $\zo^N$.
    Let $X$ be a transmitted codeword and $Y$ be the output of the channel on $X$. 
    By \Cref{prop:long-bits}, with probability at least $1 - \exp(-\Theta(N^{-1/3}))$, the number of bits that can be deleted in $X$ (i.e., the number of bits that belong to runs of length at least $\tau$) is at most $N(\frac{\tau+1}{2^\tau} + \Theta(N^{-1/3}))$. 
    On these bits, the channel acts in an i.i.d. fashion and deletes each bit with probability $d$. 
    Thus, conditioned on $X$ having at most $N(\frac{\tau+1}{2^\tau} + \Theta(N^{-1/3}))$ bits that belong to runs of length at least $\tau$, the probability that more than $(d + N^{-1/3})\cdot N(\frac{\tau+1}{2^\tau} + \Theta(N^{-1/3}))$ bits are deleted from $X$ is at most $\exp(-\Theta(N^{1/3}))$.

    Thus, conditioned on $X$ having at most $N(\frac{\tau+1}{2^\tau} + \Theta(N^{-1/3}))$ bits that belong to runs of length at least $\tau$, we have that 
    \begin{equation} \label{eq:Y-len}
        |Y| \geq N - \left(d \cdot \frac{\tau+1}{2^\tau}N + \Theta(N^{2/3})\right) = N\left( 1 - d\cdot \frac{\tau+1}{2^\tau} - \Theta(N^{-1/3})\right)
    \end{equation}
    with probability at least $1 - \exp(-\Theta(N^{1/3}))$.
    
    Our decoding algorithm is as follows. Upon receiving $Y$, if $|Y| < N\left( 1 - d\cdot \frac{\tau+1}{2^\tau} - \Theta(N^{-1/3})\right)$, we declare ``short'' error, and denote this event by $P_{\text{short}}$. Otherwise, we check if $Y$ is a subsequence of a single codeword in $\cC$. If it is, then it must be $X$ and we declare success. If $Y$ is a subsequence of two or more codewords, we declare a collision error. Denote this event by $P_{\text{col}}$.

    We turn to computing the collision error, assuming that the length of $Y$ is at least $m$ where $m$ is the right-hand side of \cref{eq:Y-len}.
    Let $X_1$ be a transmitted codeword, and let $X_2$ be a random string of length $N$. As in \cite{diggavi2006information}, we first upper bound the probability that the decoding algorithm declares collision error because of $X_2$ as 
    \begin{align*}
        P_{\text{col},X_2}&:=\sum_{Y,|Y|\geq m} \Pr[Y \textup{ is subsequence of } X_2 \mid X_1] \Pr[Y|X_1] \\
        &=\sum_{j = m}^N \left[ \left( \sum_{i=0}^{N-j} \binom{N}{i} \cdot 2^{-N}\right) \cdot \left(\sum_{Y,|Y| = j} \Pr[Y|X_1] \right) \cdot \Pr[|Y| = j] \right] \\
        &\leq 2^{-N}\sum_{j = m}^N   \sum_{i=0}^{N-j} \binom{N}{i}\\
        &\leq 2^{-N} \cdot(N- m + 1)^2 \cdot \binom{N}{N - m}\\
        &\leq 2^{-N(1 - h(\frac{\tau+1}{2^{\tau}}\cdot d + o(1)) - o(1))}
    \end{align*}
    The second equality follows by \Cref{lem:lev-superseq} and taking into consideration that the size of the insertion ball depends only on the length of $Y$. 
    The second inequality is due to the assumption that $d \cdot \frac{\tau+1}{2^\tau} < \frac{1}{2}$ which implies that $N - m < N/2$. The last inequality follows by the bound $\binom{N}{\alpha N} < 2^{N h(\alpha)}$ for $\alpha \leq 1/2$.
    Thus, we union bound over all $X_2 \in \cC$ and get that the $P_{\text{col}} = 2^{NR}P_{\text{col},X_2}$. Thus, 
    \begin{align*}
        P_{\text{e}} &=  P_{\text{col}} + P_{\text{short}} \\
        &\leq 2^{N(R-(1 - h(\frac{\tau+1}{2^\tau}\cdot d + o(1)) - o(1)))} + \exp(-\Theta (N^{1/3}))\;.
    \end{align*}
    As a result, for any constant $\varepsilon>0$, for $R = 1 - h(\frac{\tau+1}{2^\tau}\cdot d) - \varepsilon$ it holds that $\lim_{N\rightarrow\infty} P_e \rightarrow 0$.
\end{proof}

\subsection{Second lower bound}
In this section we provide another lower bound on the capacity of $\bdctzerod$. 
For this lower bound, we will greedily construct binary codes that can be reliably decoded against an adversary that is very restricted with the deletions that he can apply. 
We note that this construction is inspired by the binary inner code given in \cite{con2022improved}. However, the codes we design here are more structured and tailored for this particular channel.
For $v_1,\ldots,v_{\tau}$ such that $\sum_{i=1}^{\tau}v_i = 1$, we define $H((v_1, \ldots, v_\tau)) := \sum_{i=1}^{\tau} -v_i \log v_i$.

We prove the following.
\begin{theorem} \label{thm:bdc-tau-thm}
    Let $d \in (0,1)$, $M$ and $\tau$ be positive integers, and $\beta_1, \ldots, \beta_\tau \in [0,1]$ be such that $\sum_{i=1}^\tau i\beta_i = 1$. 
    Denote $\beta = \sum_{i=1}^{\tau}\beta_i$.
    Then, the capacity of $\bdctzerod$ is at least
    \[
        \frac{\beta\cdot H\left(\frac{1}{\beta}\cdot(\beta_1, \ldots, \beta_\tau)\right) - (2\beta_\tau + g(d)) h\left( \frac{g(d)}{2\beta_{\tau}+g(d)}\right)- h(g(d)) }{(1 - \beta_{\tau} \tau) + \beta_{\tau} M},
        \]
        where 
        \[
            g(d) = 2\tau \cdot d^M + \sum_{i=1}^{\tau}(\tau - i)\binom{M}{i}(1-d)^i d^{M-i}.
        \]
\end{theorem}

We start by borrowing a lemma from the pioneering work of Levenshtein \cite{levenshtein1966binary} who first considered recovering from worst-case insertions and deletions. 
The lemma gives an upper bound on the deletion ball of a string that depends on the number of runs of the string. Formally, we have the following. 

\begin{lemma}[\cite{levenshtein1966binary}] \label{lem:del-ball-lev-bound}
	Let $s$ be a string and denote by $r(s)$ the number of runs in $s$. There are at most 
	\[
	\binom{r(s) + \ell - 1}{\ell}
	\]
	different subsequences of $s$ of length $|s| - \ell$.
\end{lemma}

    We start with defining a set of strings from which we shall pick our codewords.

    \begin{definition}
        Let $\beta_1,\ldots, \beta_{\tau}\in [0,1]$ such that $\sum_{i=1}^{\tau} i\beta_i = 1$.
        Let $S_{\beta_1, \ldots, \beta_{\tau}} \subset \{0,1\}^N$ be a set that contains all strings that 
        \begin{enumerate}
            \item Contain only runs of length $\leq \tau$.
            \item The number of runs of length $i\in [\tau]$ is exactly $\beta_i m$.
        \end{enumerate}
    \end{definition}

    The next claim gives a lower bound on the size of $S_{\beta_1, \ldots, \beta_\tau}$.
    \begin{claim} \label{clm:S-betas-set-size}
        Given $S_{\beta_1, \ldots,\beta_\tau}$, denote $\beta = \sum_{i=1}^{\tau} \beta_i$. Then,
        \[
        \frac{1}{N}\log_2 |S_{\beta_1, \ldots, \beta_{\tau}}| \geq H\left( \frac{1}{\beta}(\beta_1, \ldots, \beta_\tau) \right) - o(1)\;.
        \]
    \end{claim}
    \begin{proof}
        Clearly, we have that
        \[
        |S_{\beta_1, \ldots, \beta_{\tau}}| = \frac{(\beta N)!}{(\beta_1 N)! \cdots (\beta_\tau N)!}\;.
        \]
        Then, the bound follows by using Stirling's approximation (e.g., \cite[Lemma 7.3]{mitzenmacher2005probability}).
    \end{proof}

    We will construct a code $C \subset S_{\beta_1, \ldots, \beta_\tau}$ that is robust against an adversary that performs a $\delta$ fraction of restricted deletions. More formally, we have the following definition.
    \begin{definition} \label{def:spec-adv}
        We call an adversary a \emph{$(\delta,\tau)$-restricted adversary}, if given a string $s = r_1 \circ r_2\circ \cdots \circ r_{\beta N}$, he is allowed to perform deletions only to runs $r_i$ such that $r_i$ or $r_{i-1}$ are of length $\tau$.
    \end{definition}

    Our next two claims will compute the size of a ``deletion ball'' and the size of an ``insertion ball'' under the $(\delta,\tau)$-restricted adversary. 
    \begin{claim} \label{clm:adv-del-bal}
        Let $s\in S_{\beta_1, \ldots, \beta_\tau}$. Then, $(\delta,\tau)$-restricted adversary defined in \Cref{def:spec-adv} can produce at most
        \[
        \binom{2\beta_{\tau}N + \delta N}{\delta N} 
        \]
        different subsequences of $s$.
    \end{claim}
    \begin{proof}
        The number of runs of length $\tau$ is $\beta_{\tau} N$. Therefore, by \Cref{def:spec-adv}, the $(\delta, \tau)$-restricted adversary can apply deletions only to $2\beta_{\tau} N$ runs. 
        As in the proof of \Cref{lem:del-ball-lev-bound} in \cite{levenshtein1966binary}, we observe that each subsequence can be described by the number of deletions applied to each run. The proposition follows by the equivalence to the ``stars and bars'' paradigm. 
    \end{proof}
    
    To bound the number of different strings $s\in S_{\beta_1, \ldots, \beta_{\tau}}$ for which the $(\delta,\tau)$-restricted adversary can turn $s$ into $s'$, we use a lemma from \cite{guruswami2017deletion} that provides a general upper bound on the number of strings in $\zo^N$ that contain a particular string $s'\in \zo^{N - \delta N}$ as a subsequence.
    \begin{lemma}[\protect{\cite[Lemma 7]{guruswami2017deletion}}] \label{lem:GW-ins-ball-bound}
        Let $\delta\in (0,1/2)$ and let $s'\in \zo^{N - \delta N}$. The number of strings $s\in \zo^N$ that contain $s'$ as a subsequence is at most $\delta N \binom{N}{\delta N}$.
    \end{lemma}
    Since $S_{\beta_1, \ldots, \beta_\tau}\subseteq \zo^{N}$, we have the following claim
    \begin{claim} \label{clm:adv-ins-bal}
        Let $s'\in \zo^{N - \delta N}$ be a string that was obtained the $(\delta,\tau)$-restricted adversary that was given in a string in $S_{\beta_1, \ldots, \beta_\tau}$.
        Then, there are at most $\delta N\binom{N}{\delta N}$
        strings $s\in S_{\beta_1, \ldots, \beta_\tau}$ 
        that can be transformed to $s'$ by the $(\delta,\tau)$-restricted adversary.
    \end{claim}
    
    Combining these two claims, we can claim the existence of a code $C \subseteq S_{\beta_1, \ldots, \beta_{\tau}}$ with a certain rate that is robust against the $(\delta,\tau)$-restricted adversary (defined in \Cref{def:spec-adv}).
    \begin{proposition}
    \label{prop:inner-code}
        Let $\delta\in (0,1)$. Let $\beta_1, \ldots, \beta_\tau \in [0,1]$ be such that $\sum_{i=1}^\tau i\beta_i = 1$ and denote $\beta = \sum_{i=1}^\tau \beta_i$. There is a code $C\subset S_{\beta_1, \ldots, \beta_\tau}$ with rate
        \[
        H\left( \frac{1}{\beta}(\beta_1, \ldots, \beta_\tau) \right) - (2\beta_\tau + \delta) h\left( \frac{\delta}{2\beta_{\tau}+\delta}\right)- h(\delta) - o(1)
        \]
        that is robust against the $(\delta,\tau)$-restricted adversary.
    \end{proposition}
    \begin{proof}
        We construct the code greedily. Each time we add a string $c$ from $S_{\beta_1, \ldots, \beta_{\tau}}$ to our codebook $C$, we need to discard from $S_{\beta_1, \ldots, \beta_{\tau}}$ all the strings that are \emph{confusable} with $c$. That is, all strings in $s\in S_{\beta_1, \ldots, \beta_{\tau}}$ for which the $(\delta,\tau)$-restricted adversary can transform $s$ and $c$ into the same string.
        
        By \Cref{clm:adv-del-bal} and \Cref{clm:adv-ins-bal}, each string added to the code is confusable with at most 
        \begin{align*}
        \delta N\cdot \binom{2\beta_2 N + \delta N}{\delta N} \cdot \binom{N}{\delta N}
        \end{align*}
        strings in $S_{\beta_1, \ldots, \beta_{\tau}}$. Then, we have
        \begin{align*}
            \frac{1}{N}\log_2 |C| &\geq \frac{1}{N}\log_2|S_{\beta_1, \ldots, \beta_\tau}| - \frac{1}{N}\log_2\left| \delta N\cdot \binom{2\beta_2 N + \delta N}{\delta N} \cdot \binom{N}{\delta N} \right| \\
        \end{align*}
        The proposition follows by \Cref{clm:S-betas-set-size} and by applying the Stirling approximation.
    \end{proof}

    We now introduce our construction, which basically takes the code guaranteed by \Cref{prop:inner-code} and blows up the runs of length $\tau$.
    \begin{construction} \label{const:code-for-bdc-long-del}
        Let $M \geq \tau$ be an integer and let $\cC'$ be the code guaranteed by \Cref{prop:inner-code}. 
        Define $\cC$ to be the code obtained from $\cC'$ by transforming each run of length $\tau$ to a run of length $M$. 
        The rate of $\cC$ is 
        \[
        \frac{H\left( \frac{1}{\beta}(\beta_1, \ldots, \beta_\tau) \right) - (2\beta_\tau + \delta) h\left( \frac{\delta}{2\beta_{\tau}+\delta}\right)- h(\delta)}{1 + (M - t)\beta_{\tau}} - o(1) \;.
        \]
    \end{construction}
    Our next theorem states that this code is robust against the channel $\bdctzerod$. We denote by $\Bin(n, p)$ the binomial distribution.
    \begin{theorem}\label{thm:decode-correct}
        Let $M$ and $\tau$ be integers where $M\geq \tau$. Let $\beta_1, \ldots, \beta_\tau \in [0,1]$ be such that $\sum_{i=1}^\tau i\beta_i = 1$ and denote $\beta = \sum_{i=1}^\tau \beta_i$. Let $\delta \in (0,1/2)$ and denote by $\cC$ the code from \Cref{const:code-for-bdc-long-del} when given $\delta, \beta_1, \ldots, \beta_\tau$ and $M$.
        Let $Z \sim \Bin(M, 1-d)$ and for every nonnegative integer $i$, define
        \[
            P^{(\tau) \rightarrow (i)}:= \Pr[Z = i]
        \]
        and let 
        \begin{equation} \label{eq:constraint}
            \alpha := \beta_\tau \cdot \left(2 \tau  P^{(\tau) \rightarrow (0)} + \sum_{i=1}^{\tau - 1} (\tau - i) P^{(\tau) \rightarrow (i)}\right) \;. 
        \end{equation}
        Let $c\in C$ be a codeword and let $y$ be the output of the channel $\bdctzerod$ on $c$.
        If $\alpha<\delta$,
        there exists an algorithm that given $y$ as input outputs $c$ with probability $1 - \exp(-\Theta (n))$. 
    \end{theorem}
    
    \begin{proof}

    To prove the theorem, we will present a decoding algorithm and prove that it succeeds with probability $1 - \exp(-\Omega (N))$.
    Let $C$ be the code given in \Cref{const:code-for-bdc-long-del} and denote by $C'$ the associated code given by \Cref{prop:inner-code}. That is, for every $c'\in C'$ we obtain $c\in C$ by turning the $\tau$-runs into $M$-runs.
    Assume that $c$ was transmitted, assume that the channel output on $c$ was $\tilde{c}$, and denote by $c'$ the respective codeword in $C'$. Next, we give the decoding algorithm.
    \paragraph{Decoding algorithm.}
    \begin{enumerate}
        \item \label{decode:blow-up} For every run in $\tilde{c}$, if its length is at least $\tau$, decode it as a run of length $\tau$. Denote the output of this step as $\tilde{c}'$ and note that $\tilde{c}'$ contains only runs of length $\leq \tau$.
        
        \item \label{decode:inner-code-decode2} 
        Find a codeword in $C$ such that the $(\delta, \tau)$-restricted adversary can produce $\tilde{c}$ from this codeword.
        Return this codeword.
    \end{enumerate}

        Recall that the channel can apply deletions only to runs with length at least $\tau$ and that we have exactly $\beta_\tau N$ such runs in $c$ (and in $c'$). Denote these runs in $c'$ by $r_1, \ldots, r_{\beta_\tau N}$.
        We denote by $Z_i, i\in [\beta_\tau N]$ the random variables that correspond to the number of bits that survived after transmitting the blowed up version $r_i$, of length $M$. 
        Clearly, $Z_i \sim \Bin(M,1-d)$.
        Define the following random variable for $i\in [\beta_\tau m]$:
        \[
        X_i := 
        \begin{cases}
            0, \qquad \,\,\, \textup{if } Z_i \geq \tau\\
            \tau - Z_i, \,\textup{if } 0< Z_i < \tau \\
            2\tau \quad \;\,\,\,\,\, \textup{ if } Z_i = 0 
        \end{cases}\;.
        \]
        We claim that $\ed(c', \tilde{c}') \leq \sum_{i=1}^{\beta_\tau m}X_i$, and in particular it bounds the number of deletions that were performed to $c'$ in order to obtain $\tilde{c}'$ (recall that $c'$ is the codeword of $C'$). Indeed, 
        \begin{itemize}
            \item If $Z_i \geq \tau$, then in Step~\ref{decode:blow-up}, the algorithm decodes it as runs of length $\tau$ and no deletion occurred in the run $r_i$. 
            \item If $0 < Z_i< \tau$, then in Step~\ref{decode:blow-up}, the algorithm decodes it as runs of length $Z_i$ and thus $\tau - Z_i$ deletions happened to $r_i$.
            \item If $Z_i = 0$, then this run is completely deleted by the channel. In this case, the run before and the run after merge to a single run of the same symbol. Then, in Step~\ref{decode:blow-up}, the merged run is decoded to a single run. 
            In this case, the channel caused $\tau$ deletions (the run $r_i$ of length $\tau$ was completely deleted) and the decoding algorithm caused another at most $\leq \tau$ deletions.
            More formally, assume that we have the following sequence of runs $r_1 \circ r_2 \circ \cdots \circ r_{2\ell+1}\circ r_{2\ell + 2}$ in $c'$ and that the runs $r_2, r_4, \ldots, r_{2\ell}$ were deleted by the channel and that $r_1$ and $r_{2\ell+1}$ and $r_{2\ell + 2}$ were \emph{not deleted} by the channel. The decoding algorithm sees a long run of the same symbol and decodes it as a single run. 
            Then, in the worst case scenario, the decoding algorithm deletes all the bits that correspond to $r_3, \ldots, r_{2\ell+1}$. Indeed, if the $|r_1| < \tau$ then the length of the decoded merged run is at least $r_1$, and if $|r_1| = \tau$ the length of the decoded merged run is at least $Z_1$.
            In both cases, the number of additional deletions caused by the merge and the decoding algorithm is at most $|r_3| + |r_5| + \cdots + |r_{2\ell+1}| \leq \tau \cdot \ell$. Therefore, in total, we suffered at most $2\tau \ell$ deletions.

        \end{itemize}
        We have shown that $\ed(c', \tilde{c}') \leq \sum_{i=1}^{\beta_\tau m}X_i$. We also observe that $\tilde{c}'$ can be obtained from $c'$ by the $(\delta,\tau)$-restricted adversary. Indeed, by the cases analyzed above and the definition of the $X_i$s, every deletion is applied to a run of length $\tau$ or to the subsequent run. 
        Our next goal is to provide an upper bound for the variable $X := \sum_{i=1}^{\beta_\tau N} X_i$ that holds with high probability. First, observe that
        \begin{equation} \label{eq:bX-upper-bound}
        \bbE[X] \leq \left(2\tau P^{(\tau) \rightarrow (0)} + \sum_{i=1}^{\tau - 1} (\tau - i) P^{(\tau) \rightarrow (i)}\right) \cdot \beta_\tau N \;.
        \end{equation}
        On the other hand, since the event $0 < Z_i < \tau$ causes exactly $\tau - Z_i$ deletion and the event $Z_i=0$ causes at least $\tau$ deletion, we can also lower bound $\bbE[X]$ by
        \begin{equation}\label{eq:bX-lower-bound}
            \bbE[X] \geq \left(\tau P^{(\tau) \rightarrow (0)} + \sum_{i=1}^{\tau - 1} (\tau - i) P^{(\tau) \rightarrow (i)}\right) \cdot \beta_\tau N \;.
        \end{equation}
        
        Note that the $X_i$s are independent random variables (by the nature of the channel). Also, $0\leq X_i \leq 2\tau$ and they have finite first and second moments. 
        Thus, we apply \Cref{lem:hoeff} and get that for any $\nu > 0$, 
        \begin{equation} \label{eq:X-concentrated}
            \Pr[X > (1 + \nu)\bbE[X]] \leq \exp\left( - \frac{2 \nu^2 (\bbE[X])^2}{\beta_\tau N \cdot 4\tau^2}\right) \leq \exp(-\Omega (N))\;,
        \end{equation}
        where the second inequality follows from the lower bound for $\bbE[X]$ in \cref{eq:bX-lower-bound}.
        We have 
        \begin{align*}
            \Pr\left[\ed\left(c, \tilde{c}\right)  > \delta N\right] &\leq \Pr[X > \delta N] \\
            &= \Pr\left[X > \left(1 + \frac{\delta - \alpha}{\alpha}\right) \alpha N\right]\\
            &\leq \Pr\left[X > \left(1 + \frac{\delta - \alpha}{\alpha}\right) \bbE [X]\right]\\
            &\leq \exp(-\Omega(N)) \;.
        \end{align*}
        The first inequality follows by the claim that $X$ is an upper bound on $\ed\left(c', \tilde{c}'\right)$. The second inequality follows by the assumption $\bbE[X]\leq \alpha N$ and the last inequality follows by the assumption that $\delta > \alpha$ and substituting $\nu = (\delta - \alpha) / \alpha$ in \cref{eq:X-concentrated}.
        To conclude, the probability that $\ed\left(c', \tilde{c}'\right) \leq \delta N$ is at least $1 - \exp(-\Omega(N))$. Now, since the code $C'$ is robust against the $(\delta,\tau)$-restricted adversary, and since we showed that all the deletions performed by the channel or the algorithm are of this kind, the theorem follows.
    \end{proof}

    \subsection{Achievable rates for $\tau = 2$ and $\tau = 3$}
    We provide plots of the lower bounds on the capacity of the $\bdctzerod$ channel derived in previous sections, where we focused on $\tau=2$ (\cref{fig:tau-2}) and $\tau=3$ (\cref{fig:tau-3}). Observe that both our lower bounds given by \Cref{thm:dg-like-bound} and \Cref{thm:bdc-tau-thm} do not provide efficient coding schemes by default.
    Nevertheless, by invoking \Cref{thm:efficient-bounded-rl}, these two lower bounds can be transformed into explicit and efficient codes for the $\bdctzerod$ channel. 
    Also, note that the lower bound implied by \Cref{thm:bdc-tau-thm} contains several parameters ($\beta_1,\ldots,\beta_\tau, M$).
    Therefore, we implemented a greedy search on these parameters to numerically maximize the lower bound and ran it for every deletion probability $d$, with two-digit precision.

    We note that in the case of $\tau=3$, the transmitter can send strings in $\{0,1\}^N$ that contain only runs of length $1$ or runs of length $2$ and the receiver receives those strings without any error.
    Thus, a baseline lower bound for the capacity of $\bdctzerod$ when $\tau = 2$ is $\log_2 ((1+\sqrt5)/2) -o(1) \leq 0.6943$. 
    Indeed, the number of such strings of length $n$ is given by the $n$-th Fibonacci number, defined by the recurrence $F(n) = F(n-1) + F(n-2)$.
    Our lower bounds improve on this baseline lower bound.
    \begin{figure}
		\centering	\begin{tikzpicture}
		
		\begin{axis}[
            width = 12cm,
		scaled y ticks = false,
		tick label style={/pgf/number format/fixed },
		axis lines = left,
		xlabel = $d$,
		ylabel = {$Rate$},
		xtick={0.1,0.2,0.3,0.4,0.5,0.6,0.7,0.8,0.9,1},
		ytick={0,0.2,0.4,0.6,0.8,1}, 
		legend pos=north east,
		ymajorgrids=true,
		grid style=dashed,
		]

            \addplot [
		domain=0:0.67, 
		samples=100, 
		color=red,
		]
		{1 + 0.75*x*ln(0.75*x)/ln(2) + (1- 0.75*x)*ln(1-0.75*x)/ln(2)};
		\addlegendentry{$1 - h\left(\frac{3}{4}d\right)$ \Cref{thm:dg-like-bound}}

            \addplot[
		color=blue,
            mark size=0.6pt,
		mark=*,
		]
		coordinates {
                (0.0,0.6942)
                (0.01,0.6024)
                (0.02,0.5455)
                (0.03,0.5393)
                (0.04,0.5315)
                (0.05,0.5221)
                (0.06,0.5115)
                (0.07,0.4997)
                (0.08,0.487)
                (0.09,0.4733)
                (0.1,0.4587)
                (0.11,0.448)
                (0.12,0.4438)
                (0.13,0.439)
                (0.14,0.4337)
                (0.15,0.4278)
                (0.16,0.4213)
                (0.17,0.4143)
                (0.18,0.4067)
                (0.19,0.3986)
                (0.2,0.3899)
                (0.21,0.3832)
                (0.22,0.3794)
                (0.23,0.3752)
                (0.24,0.3706)
                (0.25,0.3656)
                (0.26,0.3602)
                (0.27,0.3544)
                (0.28,0.3481)
                (0.29,0.3413)
                (0.3,0.3363)
                (0.31,0.3326)
                (0.32,0.3285)
                (0.33,0.3241)
                (0.34,0.3193)
                (0.35,0.3142)
                (0.36,0.3086)
                (0.37,0.3028)
                (0.38,0.2993)
                (0.39,0.2955)
                (0.4,0.2913)
                (0.41,0.2867)
                (0.42,0.2817)
                (0.43,0.2763)
                (0.44,0.2724)
                (0.45,0.2686)
                (0.46,0.2643)
                (0.47,0.2597)
                (0.48,0.2546)
                (0.49,0.2501)
                (0.5,0.2462)
                (0.51,0.2419)
                (0.52,0.2371)
                (0.53,0.2321)
                (0.54,0.2282)
                (0.55,0.2239)
                (0.56,0.2191)
                (0.57,0.2144)
                (0.58,0.2102)
                (0.59,0.2057)
                (0.6,0.2007)
                (0.61,0.1966)
                (0.62,0.1919)
                (0.63,0.1871)
                (0.64,0.1828)
                (0.65,0.1779)
                (0.66,0.1735)
                (0.67,0.1687)
                (0.68,0.1641)
                (0.69,0.1593)
                (0.7,0.1547)
                (0.71,0.1498)
                (0.72,0.1452)
                (0.73,0.1405)
                (0.74,0.1356)
                (0.75,0.1307)
                (0.76,0.1259)
                (0.77,0.1211)
                (0.78,0.1162)
                (0.79,0.1112)
                (0.8,0.1063)
                (0.81,0.1013)
                (0.82,0.0963)
                (0.83,0.0913)
                (0.84,0.0862)
                (0.85,0.0811)
                (0.86,0.0759)
                (0.87,0.0708)
                (0.88,0.0656)
                (0.89,0.0603)
                (0.9,0.055)
                (0.91,0.0497)
                (0.92,0.0444)
                (0.93,0.039)
                (0.94,0.0335)
                (0.95,0.028)
                (0.96,0.0225)
                (0.97,0.017)
                (0.98,0.0114)
                (0.99,0.0057)

        }; 
        \addlegendentry{\Cref{thm:bdc-tau-thm}}
        
        \end{axis}
	\end{tikzpicture}
	\caption{Lower bounds on the capacity of the $\bdctzerod$ for $\tau=2$.} \label{fig:tau-2}
	\end{figure}
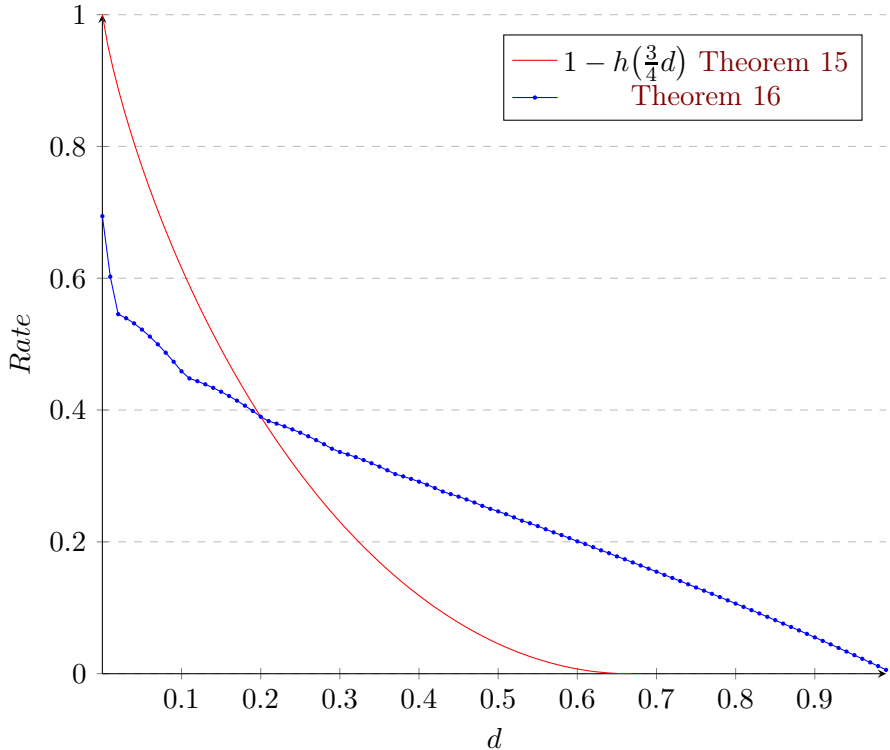

    \begin{figure}
		\centering
		\begin{tikzpicture}
		
		\begin{axis}[
            width = 12cm,
		scaled y ticks = false,
		tick label style={/pgf/number format/fixed },
		axis lines = left,
		xlabel = $d$,
		ylabel = {$Rate$},
		xtick={0.1,0.2,0.3,0.4,0.5,0.6,0.7,0.8,0.9,1},
		ytick={0,0.2,0.4,0.6,0.8,1}, 
		legend pos=north east,
		ymajorgrids=true,
		grid style=dashed,
		]
		
		\addplot [
		domain=0:1, 
		samples=100, 
		color=black,
		]
		{0.6943};
		\addlegendentry{$0.6943$}

            \addplot [
		domain=0:1, 
		samples=100, 
		color=red,
		]
		{1 + (x/2)*ln(x/2)/ln(2) + (1- x/2)*ln(1-x/2)/ln(2)};
		\addlegendentry{$1 - h\left(\frac{1}{2}d\right)$ \Cref{thm:dg-like-bound}}
        
            \addplot[
		color=blue,
            mark size=0.6pt,
		mark=*,
		]
		coordinates {
                    (0.0,0.8791)
                    (0.01,0.8321)
                    (0.02,0.8075)
                    (0.03,0.8036)
                    (0.04,0.799)
                    (0.05,0.7937)
                    (0.06,0.7878)
                    (0.07,0.7816)
                    (0.08,0.7753)
                    (0.09,0.7687)
                    (0.1,0.7622)
                    (0.11,0.7594)
                    (0.12,0.7571)
                    (0.13,0.7545)
                    (0.14,0.7518)
                    (0.15,0.7489)
                    (0.16,0.7458)
                    (0.17,0.7426)
                    (0.18,0.7395)
                    (0.19,0.7361)
                    (0.2,0.7329)
                    (0.21,0.7312)
                    (0.22,0.7295)
                    (0.23,0.7278)
                    (0.24,0.7259)
                    (0.25,0.7241)
                    (0.26,0.7221)
                    (0.27,0.7201)
                    (0.28,0.7181)
                    (0.29,0.716)
                    (0.3,0.7148)
                    (0.31,0.7137)
                    (0.32,0.7124)
                    (0.33,0.7111)
                    (0.34,0.7098)
                    (0.35,0.7085)
                    (0.36,0.7072)
                    (0.37,0.7059)
                    (0.38,0.7052)
                    (0.39,0.7045)
                    (0.4,0.7037)
                    (0.41,0.7028)
                    (0.42,0.7018)
                    (0.43,0.7008)
                    (0.44,0.7002)
                    (0.45,0.6997)
                    (0.46,0.6992)
                    (0.47,0.6987)
                    (0.48,0.6981)
                    (0.49,0.6975)
                    (0.5,0.697)
                    (0.51,0.6965)
                    (0.52,0.6959)
                    (0.53,0.6952)
                    (0.54,0.6945)
                    (0.55,0.6939)
                    (0.56,0.6933)
                    (0.57,0.6927)
                    (0.58,0.6927)
                    (0.59,0.6927)
                    (0.6,0.6927)
                    (0.61,0.6927)
                    (0.62,0.6927)
                    (0.63,0.6927)
                    (0.64,0.6927)
                    (0.65,0.6927)
                    (0.66,0.6927)
                    (0.67,0.6927)
                    (0.68,0.6927)
                    (0.69,0.6927)
                    (0.7,0.6927)
                    (0.71,0.6927)
                    (0.72,0.6927)
                    (0.73,0.6927)
                    (0.74,0.6927)
                    (0.75,0.6927)
                    (0.76,0.6927)
                    (0.77,0.6927)
                    (0.78,0.6927)
                    (0.79,0.6927)
                    (0.8,0.6927)
                    (0.81,0.6927)
                    (0.82,0.6927)
                    (0.83,0.6927)
                    (0.84,0.6927)
                    (0.85,0.6927)
                    (0.86,0.6927)
                    (0.87,0.6927)
                    (0.88,0.6927)
                    (0.89,0.6927)
                    (0.9,0.6927)
                    (0.91,0.6927)
                    (0.92,0.6927)
                    (0.93,0.6927)
                    (0.94,0.6927)
                    (0.95,0.6927)
                    (0.96,0.6927)
                    (0.97,0.6927)
                    (0.98,0.6927)
                    (0.99,0.6927)
	
                };
                \addlegendentry{\Cref{thm:bdc-tau-thm}}
		\end{axis}
		\end{tikzpicture}
		\caption{Lower bounds on the capacity of the $\bdctzerod$ for $\tau=3$.} 
        \label{fig:tau-3}
	\end{figure}
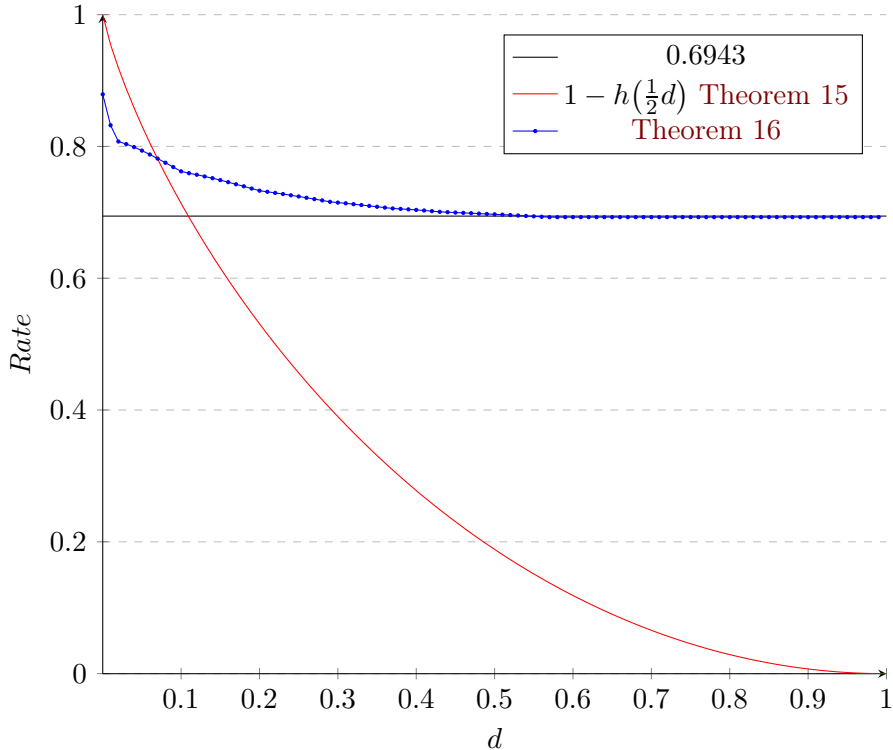

\section{Acknowledgments}
We thank Elena Grigorescu for many fruitful discussions on this problem and Olgica Milenkovic for insightful discussions that motivated the study of channels with runlength-dependent deletions.
We also thank the anonymous IEEE Transactions on Information Theory reviewers for many insightful comments that improved this work.

\bibliographystyle{alpha}
\bibliography{refs}

\appendix

\section{Proof of \cref{lem:id-dob-gen}}\label{sec:proof-id-dob-gen}

For a discrete random variable $V$ with probability mass function $p_V$, define $h_V(v)=-\log p_V(v)$.
Note that for any map $\phi$ defined on the support of $V$ we have
\begin{equation}\label{eq:local-data-proc}
    h_V(v) - h_{\phi(V)}(\phi(v))\geq 0.
\end{equation}

Also, if $|\phi^{-1}(z)|\leq M_\phi$ for all $z$, then
\begin{align}
    \E_{v\sim V}[h_V(v)-h_{\phi(V)}(\phi(v))] &= \sum_v p_V(v) \log\left(\frac{p_{\phi(V)}(\phi(v))}{p_V(v)}\right)\nonumber\\
    &= \sum_z p_{\phi(V)}(z) \sum_{v\in\phi^{-1}(z)}\frac{p_V(v)}{p_{\phi(V)}(z)} \log\left(\frac{p_{\phi(V)}(z)}{p_V(v)}\right)\nonumber\\
    &= \sum_z p_{\phi(V)}(z) H(V|\phi(V)=z)\nonumber\\
    &\leq \sum_z p_{\phi(V)}(z) \log M_\phi\nonumber\\
    &= \log M_\phi. \label{eq:UB-exp-info-density}
\end{align}
The rearrangement of the sum in the second equality is possible due to Tonelli's theorem, since all terms are non-negative.
The inequality uses the fact that $H(V|\phi(V)=z)\leq \log |\phi^{-1}(z)|$, and $|\phi^{-1}(z)|\leq M_\phi$ for all $z$ by hypothesis.

Now, we may write
\begin{align*}
    &i_{X,A}(x,a)-i_{X,\phi(A)}(x,\phi(a)) \\
    &= (h_X(x) + h_A(a) - h_{X,A}(x,a)) - (h_{X}(x)+h_{\phi(A)}(\phi(a))-h_{X,\phi(A)}(x,\phi(a)))\\
    &=(h_A(a)-h_{\phi(A)}(\phi(a))) + (h_{X,\phi(A)}(x,\phi(a))-h_{X,A}(x,a)).
\end{align*}
Then, by the triangle inequality and \cref{eq:local-data-proc} applied to $V=A\mapsto \phi(A)$ and $V=(X,A)\mapsto \widetilde{\phi}(X,A)$ with $\widetilde{\phi}(x,a)=(x,\phi(a))$,
\begin{align*}
    |i_{X,A}(x,a)-i_{X,\phi(A)}(x,\phi(a))| &\leq |h_A(a)-h_{\phi(A)}(\phi(a))| + |h_{X,A}(x,a)-h_{X,\phi(A)}(x,\phi(a))| \\
    &= (h_A(a)-h_{\phi(A)}(\phi(a))) + (h_{X,A}(x,a)-h_{X,\phi(A)}(x,\phi(a))).
\end{align*}
This means that, applying \cref{eq:UB-exp-info-density} to $\phi$ and $\widetilde{\phi}$ (and noting that $|\widetilde{\phi}^{-1}(x,z)|=|\phi^{-1}(z)|$ for all $x$ and $z$), we get
\begin{align*}
    &\E_{(x,a)\sim (X,A)}[|i_{X,A}(x,a)-i_{X,\phi(A)}(x,\phi(a))|]\\
    &\leq \E_{(x,a)\sim (X,A)}[h_A(a)-h_{\phi(A)}(\phi(a))] + \E_{(x,a)\sim (X,A)}[h_{X,A}(x,a)-h_{X,\phi(A)}(x,\phi(a))]\\
    &\leq 2\log M_\phi,
\end{align*}
as desired.

\end{document}